\documentclass[a4paper,aps,pra,notitlepage,superscriptaddress,longbibliography,nofootinbib,twocolumn,10pt,reprint]{revtex4-2}

\usepackage[utf8]{inputenc}
\usepackage{amsmath,amsthm,amssymb,amsfonts}
\usepackage{braket}
\usepackage{graphicx}
\usepackage{hyperref}
\usepackage[bb=boondox]{mathalfa}
\usepackage{yhmath} 

\usepackage{array}
\newcolumntype{C}[1]{>{\centering\let\newline\\\arraybackslash\hspace{0pt}}m{#1}}

\usepackage{tikz,tikz-3dplot,pgfplots,tikz-cd}
\usetikzlibrary{shapes,arrows,patterns,3d}
\usetikzlibrary{decorations.markings,arrows,decorations.pathreplacing,calc,intersections,through,backgrounds}

\newcommand{\ccaption}[2]{\caption[#1]{\textit{#1.} #2}}

\usepackage{mathdots}

\DeclareMathOperator*{\Moplus}{\text{\raisebox{0.25ex}{\scalebox{0.7}{$\bigoplus$}}}}
\newcommand{\spare}[1]{\left[ {#1} \right]}

\newtheorem*{result1}{Result 1}
\newtheorem*{result2}{Result 2}
\newtheorem*{result3a}{Result 3a}
\newtheorem*{result3b}{Result 3b}
\newtheorem*{result4}{Result 4}

\newtheorem{proposition}{Proposition}
\newtheorem{lemma}{Lemma}
\newtheorem{corollary}{Corollary}

\definecolor{gray}{gray}{0.5}
\definecolor{G}{rgb}{0.6,0,0}
\definecolor{J}{rgb}{0.6,0,0.6}
\definecolor{O}{rgb}{0,0,0.6}
\definecolor{dgreen}{rgb}{0,0.6,0}
\definecolor{dred}{rgb}{0.6,0,0}
\definecolor{lred}{rgb}{1,.8,.8}

\newcommand{\ii}{\mathrm{i}}
\newcommand{\ie}{{\it i.e.},\ }
\newcommand{\eg}{{\it e.g.},\ }

\newcommand{\id}{\mathbb{1}} 
\newcommand{\tr}{\operatorname{tr}}
\newcommand{\Tr}{\operatorname{Tr}}
\newcommand{\sh}{\operatorname{sh}}
\newcommand{\ch}{\operatorname{ch}}

\pgfplotsset{compat=1.18}

\begin{document}

\title{Representation theory of Gaussian unitary transformations\texorpdfstring{\\}{} for bosonic and fermionic systems}

\author{Tommaso Guaita}
\email{tommaso.guaita@fu-berlin.de}
\affiliation{Dahlem Center for Complex Quantum Systems, Arnimallee 14, 14195 Berlin, Germany}

\author{Lucas Hackl}
\email{lucas.hackl@unimelb.edu.au}
\affiliation{School of Mathematics and Statistics, The University of Melbourne, Parkville, VIC 3010, Australia}
\affiliation{School of Physics, The University of Melbourne, Parkville, VIC 3010, Australia}

\author{Thomas Quella}
\email{thomas.quella@unimelb.edu.au}
\affiliation{School of Mathematics and Statistics, The University of Melbourne, Parkville, VIC 3010, Australia}

\begin{abstract}
Gaussian unitary transformations are generated by quadratic Hamiltonians, \ie Hamiltonians containing quadratic terms in creations and annihilation operators, and are heavily used in many areas of quantum physics, ranging from quantum optics and condensed matter theory to quantum information and quantum field theory in curved spacetime. They are known to form a representation of the metaplectic and spin group for bosons and fermions, respectively. These groups are the double covers of the symplectic and special orthogonal group, respectively, and our goal is to analyze the behavior of the sign ambiguity that one needs to deal with when moving between these groups and their double cover. We relate this sign ambiguity to expectation values of the form $\braket{0|\exp{(-\ii \hat{H})}|0}$, where $\ket{0}$ is a Gaussian state and $\hat{H}$ an arbitrary quadratic Hamiltonian. We provide closed formulas for $\braket{0|\exp{(-\ii \hat{H})}|0}$ and show how we can efficiently describe group multiplications in the double cover without the need of going to a faithful representation on an exponentially large or even infinite-dimensional space. Our construction relies on an explicit parametrization of these two groups (metaplectic, spin) in terms of symplectic and orthogonal group elements together with a twisted $\mathrm{U}(1)$ group.
\end{abstract}

\maketitle

\tableofcontents

\clearpage

\section{Introduction}
Gaussian quantum states are relevant in various areas of theoretical physics, ranging from quantum optics~\cite{gerry2023introductory} and condensed matter~\cite{bogoliubov1947theory,bogoljubov1958new} to quantum information~\cite{wang2007quantum,weedbrook2012gaussian,adesso2014continuous,holevo2019quantum} and quantum field theory~\cite{birrell1984quantum,parker2009quantum,derezinski2013mathematics}. They exist for both bosonic and fermionic systems and are characterized by a number of properties, which make them particularly suitable for analytical computations. This includes the \emph{infamous} Wick's theorem, which allows one to reduce higher order correlation functions to the knowledge of the two-point correlation function, typically known as the covariance matrix. Another defining property is that their characteristic function $\chi$ on the classical phase space is of Gaussian form, \ie $\chi(w)\propto e^{-w_a X^{ab}w_b}$, where $X^{ab}$ is a positive-definite bilinear form on the dual phase space.

A closely related concept is that of Gaussian unitaries, which are those unitary transformations that map pure Gaussian states onto pure Gaussian states. While they naturally describe the time evolution of non-interacting systems, they are also used in numerical tools for the simulation of interacting systems, such as Hartree-Fock~\cite{echenique2007mathematical,kraus2010generalized} and the Kohn-Sham approach to density functional theory~\cite{kohn1965self}. They have also been successfully employed in analyzing impurity problems~\cite{bravyi_complexity_2017,schuch2019matrix}. Gaussian unitaries further play an important role in constructions such as generalized Gaussian states~\cite{shi2018variational} and the associated generalized Wick's theorem~\cite{hackl_geometry_2020,guaita_generalization_2021}.

Gaussian unitaries are generated by exponentials of quadratic Hamiltonians, \ie where each term of the Hamiltonian is a quadratic expression of creation/annihilation operators. For bosons, such Hamiltonians form a representation of the symplectic Lie algebra $\mathfrak{sp}(2N,\mathbb{R})$ on the associated Fock space. For fermions this similarly becomes a representation of the orthogonal Lie algebra $\mathfrak{so}(2N,\mathbb{R})$. It is known~\cite{woit2017quantum}, however, that the Gaussian unitaries constructed by exponentiating these Hamiltonians do not form a representation of the symplectic Lie group $\mathrm{Sp}(2N,\mathbb{R})$ or orthogonal Lie group $\mathrm{SO}(2N,\mathbb{R})$. Rather, they from a representation of their respective double covers, namely the metaplectic group $\mathrm{Mp}(2N,\mathbb{R})$ and the spin group $\mathrm{Spin}(2N,\mathbb{R})$, as discussed in~\cite{folland1989harmonic,Simon1993,derezinski2013mathematics,woit2017quantum}.

Nonetheless, it is common practice in many applications to deal with Gaussian unitaries by parametrizing them with matrices in the groups $\mathrm{Sp}(2N,\mathbb{R})$ or $\mathrm{SO}(2N,\mathbb{R})$~\cite{bravyi_complexity_2017,holevo_evaluating_2001,wang2007quantum,shi2018variational}. While this is a powerful approach for handling these operators efficiently, it leads to a parametrization that cannot capture the full structure of the system and is inevitably ambiguous. This ambiguity typically manifests in the form of signs and phases that are ill-defined if one tries to compute them only in terms of symplectic or orthogonal matrices. While in many concrete applications these phases are not physical quantities and their ambiguity does not cause any practical consequence, it is not hard to find situations where they are of absolute practical relevance. In a recent example, Dias and König~\cite{Dias:2023arXiv230712912D, Dias:2024arXiv240319059D} have recognised that this incomplete parametrization leads to undefined relative phases when using Gaussian superpositions to simulate quantum circuits and they dedicate significant effort in their work to correctly compute and account for these phases. Another significant example arises whenever one tries to compute quantities of the form $\braket{0|e^{-\ii\hat{H}}|0}$ for some quadratic Hamiltonian $\hat{H}$. One can attempt to compute this quantity in terms of a certain matrix $M\in\mathrm{Sp}(2N,\mathbb{R})$ / $\mathrm{SO}(2N,\mathbb{R})$ that can be naturally defined from $\hat{H}$. However this approach will lead to results that are correct only up to a sign. The information about this sign indeed is simply not contained in $M$.

The only natural way to resolve these issues is to move to a parametrisation of Gaussian unitaries that is based on the metaplectic and spin groups. The main objective of our work is to present a complete and concrete construction of how this can be done in practice. Our results encompass all previous approaches to compute the phase-dependent quantities above and give a systematic way to understand and reproduce the previous results and to derive new ones in an ambiguity-free way.

To understand the problem in more detail, consider the following linear operators
\begin{align}
    \hat{q}_j=\frac{1}{\sqrt{2}}(\hat{a}^\dag_j+\hat{a}_j)\,, \quad\quad     
    \hat{p}_j=\frac{\ii}{\sqrt{2}}(\hat{a}^\dag_j-\hat{a}_j) \,,
\end{align}
for $j=1,\dots,N$, where $\hat{a}_j$ are the annihilation operators of an $N$-mode system. These operators are usually referred to as quadratures for bosons and Majorana operators for fermions and we will collect them in the vector
\begin{align}
    \hat{\xi}\equiv(\hat{q}_1,\dots,\hat{q}_N,\hat{p}_1,\dots,\hat{p}_N)\,.
\end{align}
We will consider arbitrary quadratic Hamiltonians
\begin{equation}
    \hat{H}=(\ii) \sum^{2N}_{a,b=1}h_{ab}\hat{\xi}^a\hat{\xi}^b
\end{equation}
where $h_{ab}$ is real and chosen to make $\hat{H}$ Hermitian. The factor $\ii$ is present in the fermionic case but not in the bosonic one. These Hamiltonians are closed under the commutator, namely if $\hat{H}_1$ and $\hat{H}_2$ are quadratic so is $[\hat{H}_1,\hat{H}_2]$. We can uniquely\footnote{For this uniqueness it is important that we fixed $h_{ab}$ to be real. Indeed, if we allow $h_{ab}$ to have an imaginary part, then the set of quadratic Hamiltonians would contain arbitrarily many Hamiltonians that differ only by terms proportional to the identity. All these Hamiltonians would obviously give rise to the same $K$. Fixing $h_{ab}\in\mathbb{R}$ completely resolves this ambiguity.} label each such Hamiltonian by a $2N$-by-$2N$ matrix $K$ defined by its action on the linear operators $[\hat{H},\hat{\xi}]=-K\hat{\xi}$. The matrix $K$ will be an element of the symplectic algebra $\mathfrak{sp}(2N,\mathbb{R})$ for bosons and of the orthogonal algebra $\mathfrak{so}(2N,\mathbb{R})$ for fermions. This mapping is an algebra homomorphism, that is if $\hat{H}_1$ and $\hat{H}_2$ are labelled by $K_1$ and $K_2$ then $[\hat{H}_1,\hat{H}_2]$ is labelled by $[K_1,K_2]$. In this sense the quadratic Hamiltonians give a representation of the matrix algebras  $\mathfrak{sp}(2N,\mathbb{R})$ and $\mathfrak{so}(2N,\mathbb{R})$ on the Fock space.

Consider now the group of unitary transformations that is generated by taking arbitrary products of exponentials $\mathcal{U}=e^{-\ii\hat{H}}$ of quadratic Hamiltonians $\hat{H}$. We refer to them as Gaussian unitaries. Similarly to before, these operators can be associated to a $2N\times 2N$ matrix $M$ by their action on the linear operators $\mathcal{U}^\dag\hat{\xi}\mathcal{U}=M\hat{\xi}$. As we might expect, the matrix $M$ will be an element of the symplectic group $\mathrm{Sp}(2N,\mathbb{R})$ for bosons or of the orthogonal group $\mathrm{SO}(2N,\mathbb{R})$ for fermions, and if $\mathcal{U}_1$ and $\mathcal{U}_2$ are labelled by $M_1$ and $M_2$ then $\mathcal{U}_1\mathcal{U}_2$ is labelled by $M_1M_2$. However this mapping is not injective, meaning that it cannot be a group representation. Indeed, it is a known fact (albeit not completely trivial) that for each $M$ there exist exactly two different Gaussian unitaries (related by an opposite sign) that both satisfy  $\mathcal{U}^\dag\hat{\xi}\mathcal{U}=M\hat{\xi}$.

To see a very easy example of this fact, consider the simplest system consisting of a single degree of freedom. An associated quadratic Hamiltonians is given by $\hat{H}= \frac{t}{4}(\hat{q}^2+\hat{p}^2)=t(\hat{a}^\dag\hat{a}+\frac{1}{2})$ for bosons or $\hat{H}= \frac{\ii t}{4}(\hat{q}\hat{p} - \hat{p}\hat{q})=t(\hat{a}^\dag\hat{a}-\frac{1}{2})$ for fermions. Now consider the Gaussian unitaries that can be generated by exponentiating these Hamiltonians. It is clear that both the unitaries $\mathcal{U}=\hat{\id}$ and $\mathcal{U}=-\hat{\id}$ can be obtained by choosing $t=4\pi$ and $t=2\pi$ respectively. Both these unitaries however have exactly the same trivial action on the linear operators given by $M=\id$.

As hinted above, all this can be seen as a consequence of the fact that Gaussian unitaries do not form a proper representation of $\mathrm{Sp}(2N,\mathbb{R})$ or $\mathrm{SO}(2N,\mathbb{R})$ but rather of their double cover groups, $\mathrm{Mp}(2N,\mathbb{R})$ and $\mathrm{Spin}(2N,\mathbb{R})$. Thus labelling Gaussian unitaries by matrices in $\mathrm{Sp}(2N,\mathbb{R})$ or $\mathrm{SO}(2N,\mathbb{R})$ leads to an intrinsically ambiguous parametrization. To solve this issue one needs to parametrize each Gaussian unitary by an element of the double cover groups. The aim of this paper is to explicitly construct this parametrization. 

We will introduce a parametrization of the elements of the double cover groups in terms of a tuple $(M,\psi)$ and show how each such element unambiguously identifies a single unitary $\mathcal{U}(M,\psi)$. We will further derive the group multiplication rules for the double cover group, which in turn allows us to identify the unitary arising from any product  $\mathcal{U}(M_1,\psi_1)\mathcal{U}(M_2,\psi_2)$. Finally, we will explain how to compute the unique unitary $\mathcal{U}=e^{-\ii\hat{H}}$ in the double cover that results from exponentiating a quadratic Hamiltonian $\hat{H}$, given access to a description of $\hat{H}$, \ie we will show how to compute $(M,\psi)$ such that $\mathcal{U}(M,\psi)=e^{-\ii\hat{H}}$. In the process we will also compute the related quantities $\braket{0|e^{-\ii\hat{H}}|0}$ and $\braket{0|\hat{\xi}^{a_1}\cdots\hat{\xi}^{a_d} e^{-\ii\hat{H}}|0}$, including their correct sign. Our findings allow a complete treatment of the full fermionic or bosonic Gaussian unitary group without any remaining ambiguity.

This manuscript is structured as follows: Section~\ref{sec:Setup} introduces the mathematical foundations of Gaussian states and unitaries, sets up notation and outlines the basic ideas of our construction. In section~\ref{sec:DoubleCover}, we then develop the mathematical machinery required to parametrize the two double cover groups (metaplectic and spin groups) and prove relevant lemmata and propositions. Our main results are then presented in section~\ref{sec:unitary-rep-results}, where we show explicitly how our parametrization of the double cover groups relates to Gaussian unitary transformations and where we derive key analytical formulas. We illustrate these constructions explicitly in section~\ref{sec:CaseStudies} for the simplest non-trivial examples for both bosons and fermions. section~\ref{sec:Applications} then discusses how our construction and findings are relevant in the context of variational methods and numerical simulations of complex quantum systems. Finally, we conclude in section~\ref{sec:Discussion} by discussing our findings and how it relates to other works.\\

\textbf{Distinction from previous work.} Gaussian unitaries have been previously studied in various contexts and play an important role when using Gaussian quantum states. In the following, we briefly review how our findings differ from previous work.\vspace{-1mm}
\begin{itemize}
\setlength\itemsep{-1.5mm}
    \item The fact that Gaussian unitaries form a representation of the double cover groups $\mathrm{Mp}(2N,\mathbb{R})$ and $\mathrm{Spin}(2N,\mathbb{R})$ is a well-known fact~\cite{folland1989harmonic,Simon1993,derezinski2013mathematics,woit2017quantum}. However, we are not aware of another explicit construction of the double cover that relates it directly to the complex phase of the unitary's expectation value with respect to a Gaussian state. 
    \item Crucial progress was made by Rawnsley~\cite{rawnsley2012universal}, who introduced an explicit parametrization for the universal cover of the symplectic group. In our work, we show how Rawnsley's circle function relates to the expectation value of Gaussian unitaries with respect to a Gaussian state. Moreover, we extend his construction to the fermions.
    \item Some expressions that we derive were already studied in some form in previous works. For example, the double cover group multiplication rule can be seen to be related to the ``three state overlap'' $\braket{J_1|J_2}\braket{J_2|J_3}\braket{J_3|J_1}$, where $\ket{J_i}$ are Gaussian states. Expressions for this were derived for bosons~\cite{PhysRevA.61.022306} and for fermions~\cite{bravyi_complexity_2017} using phase space methods. Interestingly, our work gives alternative expressions for them, derived with completely independent methods and without reference to the phase space formalism.
    \item The works of Dias and König~\cite{Dias:2023arXiv230712912D,Dias:2024arXiv240319059D} address a very closely related topic to ours, that is the parametrization of Gaussian states in a way that consistently deals with their global phase. Gaussian unitaries are more general objects than Gaussian states and indeed their results can be re-derived as a consequence of ours (see section~\ref{sec:Applications}). We believe that our work complements these previous results by clarifying the group theoretic structure of the relevant objects. 
\end{itemize}

\newpage
\section{\label{sec:Setup}Setup and idea}
In this section, we review some notation and define the problem. Most importantly, we explain how to describe $N$-mode Gaussian states and Gaussian untaries in terms of linear complex structures and elements of the symplectic or orthogonal Lie algebras of dimension $2N$. We further comment on the nature of the corresponding orthogonal and symplectic Lie groups and their double covers. We finally introduce a picture of these groups in terms of a principal fibre bundle, which will help visualise their structure. The mathematical formalism of this section closely follows~\cite{hackl2018aspects,hackl2020geometry,hackl2021bosonic}.

\subsection{Hilbert space and operators}
We consider a bosonic or fermionic system with $N$ degrees of freedom, realized on the associated Fock space. Any Hilbert space operator can be constructed from $2N$ independent linear observables $\hat{\xi}^a$ that span a real vector space $V$, the phase space. Consequently, the components $\hat{\xi}^a$ satisfy the canonical commutation or anticommutation relations
\begin{align}
  \label{eq:CCR}
    [\hat{\xi}^a,\hat{\xi}^b]&=\ii\Omega^{ab}\,,&&\textbf{(bosons)}\\
    \{\hat{\xi}^a,\hat{\xi}^b\}&=G^{ab}\,,&&\textbf{(fermions)}
\end{align}
where $\Omega^{ab}$ is a symplectic form and $G^{ab}$ a positive definite metric. We can always fix a standard basis, known as either bosonic quadrature operators (position and conjugate momentum) or fermionic Majorana modes\footnote{Instead of the Hermitian operators $\hat{\xi}^a$, we can also take complex linear combinations to construct the raising and lowering operators $\hat{a}_i=\frac{1}{\sqrt{2}}(\hat{q}_i+\ii\hat{p}_i)$ and $\hat{a}_i^\dagger=\frac{1}{\sqrt{2}}(\hat{q}_i-\ii\hat{p}_i)$.},
\begin{align}
    \hat{\xi}^a\equiv(\hat{q}_1,\dots,\hat{q}_N,\hat{p}_1,\dots,\hat{p}_N)\,,\label{eq:xi-basis}
\end{align}
such that above forms are represented by the matrices
\begin{align}
    \Omega^{ab}\equiv\left(\begin{array}{cc}
    0 & \id \\
    -\id & 0
    \end{array}\right)\quad\text{and}\quad G^{ab}\equiv\left(\begin{array}{cc}
    \id & 0 \\
    0 & \id
    \end{array}\right)\label{eq:Omega-G-standard-form}
\end{align}
for bosons and fermions, respectively. We will also introduce their inverses
\begin{align}\label{eq:gomega}
    \omega_{ab}=(\Omega^{-1})_{ab}\equiv\left(\begin{array}{cc}
    0 & -\id \\
    \id & 0
    \end{array}\right)\,,\, g_{ab}=(G^{-1})_{ab}\equiv\left(\begin{array}{cc}
    \id & 0 \\
    0 & \id
    \end{array}\right)\,.
\end{align}
Sometimes, we will suppress the index notation and just write matrix equations, such as $\Omega\omega=\id$, which will be short hand for $\Omega^{ac}\omega_{cb}=\delta^a{}_b$ and so on. In particular, repeated indices represent a tensor contraction / Einstein's summation convention. Note that we can only contract a lower with an upper index or vice versa, which also limits the possible matrix products, \eg $G\Omega$ would be an invalid expression as $G^{ac}\Omega^{cb}$ would not contract an upper and a lower index. Mathematically, this is of course due to the fact that the product of two matrices representing bilinear forms is not basis-independent.

Having introduced the symplectic form $\Omega$ and the metric $G$, we can define the respective symplectic and orthogonal Lie groups and Lie algebras
\begin{align}
\begin{split}\label{eq:group-G}
    \mathrm{Sp}(2N,\mathbb{R})&=\left\{M\in\mathrm{SL}(2N,\mathbb{R})\,\big|\,M\Omega M^{\intercal}=\Omega\right\},\\
    \mathrm{SO}(2N,\mathbb{R})&=\left\{M\in\mathrm{SL}(2N,\mathbb{R})\,\big|\,MG M^{\intercal}=G\right\},
\end{split}\\
\begin{split}\label{eq:algebra-g}
    \mathfrak{sp}(2N,\mathbb{R})&=\left\{K\in\mathfrak{gl}(2N,\mathbb{R})\,\big|\,K\Omega=-\Omega K^{\intercal}\right\},\\
    \mathfrak{so}(2N,\mathbb{R})&=\left\{K\in\mathfrak{gl}(2N,\mathbb{R})\,\big|\,KG=-GK^{\intercal}\right\},
\end{split}
\end{align}
where we represent $\mathrm{SL}(2N,\mathbb{R})$ as invertible, unit-determinant linear maps $M^a{}_b: V\to V$ on the phase space $V$ and $\mathfrak{gl}(2N,\mathbb{R})$ as linear maps $K^a{}_b: V\to V$. We note that we restricted our attention to special linear transformations (those with $\det(M)=1$) when defining the invariance groups in \eqref{eq:group-G} since this ensures that the groups in question are connected. We will routinely use $\mathcal{G}$ and $\mathfrak{g}$ to refer to
\begin{align}
  \label{eq:DefGroup}
    \mathcal{G}&=\begin{cases}
    \mathrm{Sp}(2N,\mathbb{R}) & \textbf{(bosons)}\\
    \mathrm{SO}(2N,\mathbb{R}) & \textbf{(fermions)}
    \end{cases}\,,\\
    \mathfrak{g}&=\begin{cases}
    \mathfrak{sp}(2N,\mathbb{R}) & \textbf{(bosons)}\\
    \mathfrak{so}(2N,\mathbb{R}) & \textbf{(fermions)}
    \end{cases}\,.
\end{align}
In view of the compatibility of $\mathcal{G}$ and $\mathfrak{g}$ with the commutation relations \eqref{eq:CCR}, it is a natural question whether the actions of $\mathcal{G}$ and $\mathfrak{g}$ on $V$
lift to the full Hilbert space of the system, the Fock space. And while this is the case for the Lie algebra $\mathfrak{g}$, section~\ref{sec:GaussianUnitaries} will show that we will need to consider a double cover $\widetilde{\mathcal{G}}$ in the case of the group.

\subsection{Complex structure and Gaussian states}\label{sec:complex-Gaussian}
Given a bosonic or fermionic system with symplectic form $\Omega^{ab}$ or positive definite metric $G^{ab}$, respectively, we define a compatible complex structure $J$ as a linear map $J: V\to V$, such that
\begin{align}
    J^2=-\id\,,
\end{align}
and such that it satisfies the compatibility conditions
\begin{align}
\begin{aligned}\label{eq:Compatibility}
    J^a{}_c\Omega^{cd}(J^\intercal)_d{}^b&=\Omega^{ab}\,,\,-J^a{}_c\Omega^{cb}>0&&\textbf{(bosons)}\\[1mm]
    J^a{}_cG^{cd}(J^\intercal)_d{}^b&=G^{ab}&&\textbf{(fermions)}
\end{aligned}
\end{align}
with $(J^{\intercal})_a{}^b=J^b{}_a$. We note that $-J^a{}_c\Omega^{cb}$ is a symmetric bilinear form and the condition $-J^a{}_c\Omega^{cb}>0$ requires it to be positive definite.

Given a compatible complex structure $J$, we define the Gaussian state $\ket{J}$ as a normalized solution to the equations\footnote{The choice of $J$ can be interpreted as choosing the space of annihilation operators that annihilate $\ket{J}$, so Gaussian states are the same as the vacuum states in Fock quantization with respect to a given set of annihilation operators.}
\begin{align}
    \tfrac{1}{2}(\delta^a{}_b+\ii J^a{}_b)\hat{\xi}^b\ket{J}=0\,,
    \label{eq:annihilation-op}
\end{align}
where the $\frac{1}{2}$ is chosen, such that $\tfrac{1}{2}(\delta^a{}_b+\ii J^a{}_b)$ is a projector. The state $\ket{J}$ can be shown to be unique up to a complex phase~\cite{hackl2021bosonic} and this allows one to think of the matrix $J$ and the state $\ket{J}$ interchangeably. By assumption, we have $\braket{J|J}=1$. One can define $\hat{\xi}^a_{\pm}=\tfrac{1}{2}(\delta^a{}_b\mp\ii J^a{}_b)\hat{\xi}^b$ as covariant versions of creation ($\hat{\xi}^a_+$) and annihilation ($\hat{\xi}^a_-$) operators, such that~\eqref{eq:annihilation-op} simplifies to $\hat{\xi}^a_-\ket{J}=0$. Complex structures have been used to parametrize Gaussian states in the context of quantum fields in curved spacetime~\cite{ashtekar1975quantum} and more recently to study their entanglement and complexity properties~\cite{bianchi2015entanglement,bianchi2018linear,hackl2018circuit,chapman2019complexity}.

We can further derive the relations\footnote{While one can extend the family Gaussian states to also include states $\ket{J,z}$ with non-zero $z^a=\braket{J,z|\hat{\xi}^a|J,z}$, we restrict to $z=0$.}
\begin{align}
    \braket{J|\hat{\xi}^a|J}&=0\,,\label{eq:z-displacement}\\
    \braket{J|\hat{\xi}^a\hat{\xi}^b|J}&=\left\{\begin{array}{ll}
    \tfrac{1}{2}(-J^a{}_c\Omega^{cb}+\ii\Omega^{ab}) &  \textbf{(bosons)}\\[2mm]
    \tfrac{1}{2}(G^{ab}+\ii J^a{}_cG^{cb}) & \textbf{(fermions)}
    \end{array}\right.\,.\nonumber
\end{align}
By making appropriate identifications, the second equation can be unified as
\begin{align}
    \braket{J|\hat{\xi}^a\hat{\xi}^b|J}&=\frac{1}{2}(G^{ab}+\ii\Omega^{ab})\,,
\end{align}
where we introduced the new forms $G^{ab}:=-J^a{}_c\Omega^{cb}$ for bosons and $\Omega^{ab}:=J^a{}_cG^{cb}$ for fermions. We emphasize that in each case only one of the two structures $G$ and $\Omega$ depends on the state $J$, while the respective other structure is fixed by the canonical commutation or anti-commutation relations. For bosons, $G$ is state dependent, while for fermions, $\Omega$ is state dependent. For bosons, we refer to $G$ as the (bosonic) covariance matrix, while for fermions, we refer to $\Omega$ as the (fermionic) covariance matrix.

Whenever we fix a pure Gaussian state $\ket{J}$, the classical phase space is naturally equipped with the triangle of compatible Kähler structures $(G,\Omega,J)$ related by the compatibility condition
\begin{align}
    J=\Omega G^{-1}=-G\Omega^{-1}\,.\label{eq:compatibility-cond}
\end{align}
Each of these structures defines its own invariance group, namely $\mathrm{Sp}(2N,\mathbb{R})$ and $\mathrm{SO}(2N,\mathbb{R})$, as introduced in~\eqref{eq:group-G}, and the new group $\mathrm{GL}(N,\mathbb{C})$ with Lie algebra $\mathfrak{gl}(N,\mathbb{C})$
\begin{align}
    \mathrm{GL}(N,\mathbb{C})&=\left\{M\in\mathrm{GL}(2N,\mathbb{R})\,\big|\,MJ=JM\right\},\label{eq:GLNC}\\
    \mathfrak{gl}(N,\mathbb{C})&=\left\{K\in\mathfrak{gl}(2N,\mathbb{R})\,\big|\,KJ=JK\right\},\label{eq:glNC}
\end{align}
where $\mathrm{GL}(2N,\mathbb{R})$ is represented by invertible linear maps $M^a{}_b: V\to V$. We will discuss in section~\ref{sec:complex-linear-group} how the real $2N$-by-$2N$ matrices commuting with $J$ introduced in~\eqref{eq:GLNC} are isomorphic to complex $N$-by-$N$ matrices. The conditions in \eqref{eq:GLNC} and \eqref{eq:glNC} simply mean that the elements can be regarded as $\mathbb{C}$-linear (not only $\mathbb{R}$-linear) maps if we identify $J$ with the imaginary unit $\ii$. 

\begin{figure}[t]
    \centering
\begin{tikzpicture}[scale=-4]
    \begin{scope}[shift={(-.015*2,2*.00866025)}]
    \draw[fill=red,fill opacity=.5]
    (0,0)--(1,0)--(.5,-.866025)--cycle;
    \end{scope}
    \draw[fill=red,fill opacity=.5]
    (1,0)--(1.5,-.866025)--(.5,-.866025)--cycle;
    \draw[fill=blue,fill opacity=.5] (.9,-.866025)--(1.4,0)--(1.9,-.866025)--cycle; 
    \draw[fill=purple,fill opacity=.3] (1.2,-.34641)--(1.7,-1.21244)--(.7,-1.21244)--cycle;
    \node[scale=1.3,rotate=-60] at ($(1.2,-0.69282)+(30:3.2mm)$) {$\mathrm{Sp}(2N,\mathbb{R})$};
    \node[scale=1.3,rotate=60] at ($(1.2,-0.69282)+(150:3.2mm)$) {$\mathrm{SO}(2N,\mathbb{R})$};
    \node[scale=1.3,fill=white,rounded corners=2pt,inner sep=1pt,fill opacity=.9,rotate=60] at ($(1.2,-0.69282)+(150:5.4mm)$) {$\mathrm{O}(2N,\mathbb{R})$};
    \node[scale=1.3,rotate=60] at ($(1.2,-0.69282)+(150:7.4mm)$) {$\mathrm{O}^-(2N,\mathbb{R})$};
    \node[scale=1.3] at ($(1.2,-0.69282)+(-90:3.2mm)$) {$\mathrm{GL}(N,\mathbb{C})$};
    \node[scale=1.3,white] at (1.2,-0.69282) {$\mathrm{U}(N)$};
\end{tikzpicture}
    \ccaption{Illustration of 2-out-of-3 property}{We show how the three groups $\mathrm{O}(2N,\mathbb{R})$, $\mathrm{Sp}(2N,\mathbb{R})$ and $\mathrm{GL}(N,\mathbb{C})$ intersect to form the unitary group $\mathrm{U}(N)$. In particular, we see that intersecting all three groups is equivalent to intersecting any two out of the three groups. Moreover, we see that only the component $\mathrm{SO}(2N,\mathbb{R})\subset\mathrm{O}(2N,\mathbb{R})$ connected to the identity matters, while $\mathrm{O}^-(2N,\mathbb{R})\subset\mathrm{O}(2N,\mathbb{R})$ is the subset (not subgroup) of group elements not connected to the identity. \textit{This figure is reproduced from~\cite{hackl2021bosonic}}.}
    \label{fig:2-out-of-3}
\end{figure}

The three groups intersect as visualized in figure~\ref{fig:2-out-of-3} and their intersection satisfies the famous 2-out-of-3 property, which states that the intersection of any two groups is equivalent to intersecting all three. This is a consequence of the compatibility condition~\eqref{eq:compatibility-cond}, which ensures that any two of the three structures $(G,\Omega,J)$ define the third one. The intersection consists of those group elements $M\in\mathcal{G}$ that commute with $J$, \ie they form the stabilizer group respectively Lie algebra
\begin{align}
    \mathrm{U}(N)&=\left\{M\in\mathcal{G}\,\big|\,MJ=JM\right\}\,,\label{eq:U-group}\\
    \mathfrak{u}(N)&=\left\{K\in\mathfrak{g}\,\big|\,KJ=JK\right\}\label{eq:u-algebra}
\end{align}
which happens to be isomorphic to the group $\mathrm{U}(N)$ with Lie algebra $\mathfrak{u}(N)$. We further introduce the subspace
\begin{align}
    \mathfrak{u}_\perp(N)=\{K_+\in \mathfrak{g}\,|\, K_+J=-JK_+\}\,,\label{eq:u-perp}
\end{align}
which consists of the Lie algebra elements $K_+$ that anti-commute with $J$. This space is not a Lie subalgebra. One can show~\cite{hackl2021bosonic} that this space is the orthogonal complement of $\mathfrak{u}(N)$ in $\mathfrak{g}$ with respect to the Killing form $\mathcal{K}(K_1,K_2)\propto \Tr(K_1K_2)$. The orthogonal decomposition $\mathfrak{g}=\mathfrak{u}(N)\oplus\mathfrak{u}_\perp(N)$ will play a crucial role as Cartan decomposition on the Lie algebra level, as discussed in section~\ref{sec:cartan}. Using the dimension $\dim\mathfrak{g}=N(2N\pm 1)$ for bosons $(+)$ and fermions $(-)$ and $\dim\mathfrak{u}(N)=N^2$, we find $\dim\mathfrak{u}_\perp(N)=N(N\pm 1)$.

\subsection{Relative complex structure}
Given a $2N$-dimensional phase space $V$ that is either equipped with a symplectic form $\Omega$ for bosons or a metric $G$ for fermions, we can ask \emph{how many} complex structures $J$ exist that satisfy the compatibility conditions~\eqref{eq:Compatibility}. Having chosen a basis, where $\Omega$ or $G$, respectively, take the standard forms~\eqref{eq:Omega-G-standard-form}, it is easy to check that a linear map $J$ with the matrix representation
\begin{align}\label{eq:J-standard-form}
    J\equiv\begin{pmatrix}
    0 & \id\\
    -\id & 0
    \end{pmatrix}
\end{align}
satisfies all the conditions of a complex structure, as introduced before. Vice versa, it is also not difficult to show that for \emph{every} complex structure $J$ satisfying the compatibility conditions~\eqref{eq:Compatibility}, we can find a basis, such that $J$ is represented by~\eqref{eq:J-standard-form}, while also retaining the standard forms~\eqref{eq:Omega-G-standard-form} for $\Omega$ or $G$, respectively~\cite{hackl2021bosonic}. This follows from the fact that $J$ is a symplectic or orthogonal group element with eigenvalues $\pm \ii$ of multiplicity $N$ each, for which it is known that it can be block diagonalized in a symplectic or orthogonal basis, respectively.\footnote{While all orthogonal group elements are diagonalizable over the complex numbers, this is not true for a general symplectic group element. However, for $J$ this property follows from $J^2=-\id$.}

Based on these considerations, it follows that given a complex structure $J$ satisfying~\eqref{eq:Compatibility} we can reach any other such complex structure by applying a group transformation $M\in\mathcal{G}$, \ie all other complex structures are
\begin{align}
    J_M=MJM^{-1}
\end{align}
for an appropriately chosen $M$. Note, however, that it is easy to see that two different $M$ and $\tilde{M}$ may give rise to the same complex structure $J_M=J_{\tilde{M}}$. This will exactly be the case, when $M=\tilde{M}u$, where $u$ is an element of the stabilizer group $\mathrm{U}(N)$ of $J$, defined in~\eqref{eq:U-group}, which we can see from computing
\begin{align}
    \hspace{-2mm}J_{M}=MJM^{-1}=\tilde{M}uJu^{-1}\tilde{M}^{-1}=\tilde{M}J\tilde{M}^{-1}=J_{\tilde{M}}.
\end{align}
We can thus recognize the manifold of all complex structures as the quotient
\begin{align}
    \mathcal{M}=\mathcal{G}/\!\sim\,=\mathcal{G}/\mathrm{U}(N)\label{eq:def-M-manifold}
\end{align}
with respect to the equivalence relation
\begin{align}
    M\sim \tilde{M}\quad\Leftrightarrow\quad M\tilde{M}^{-1}\in \mathrm{U}(N)\,.\label{eq:equivalence-relation}
\end{align}
As every linear complex structure corresponds to pure Gaussian state $\ket{J}$, we recognize $\mathcal{M}$ also as the manifold of Gaussian states.\footnote{Recall that we restricted ourselves to Gaussian states with $\braket{J|\hat{\xi}^a|J}=0$, as explained in the context of~\eqref{eq:z-displacement}.} Remembering that $\mathcal{G}=\mathrm{Sp}(2N,\mathbb{R})$ for bosons and $\mathcal{G}=\mathrm{SO}(2N,\mathbb{R})$ for fermions, we can recognize the quotient $\mathcal{M}$ as a symmetric space of type CI for bosons and of type DIII for fermions~\cite{windt2021local,hackl2021bosonic}. Using the same dimension counting argument as for $\mathfrak{u}_\perp(N)$ from~\eqref{eq:u-perp}, we find $\dim\mathcal{M}=N(N\pm1)$.

Given two complex structures $J$ and $J_M$, we can define the relative complex structure
\begin{align}
  \label{eq:RelativeComplexStructure}
    \Delta_M=-J_MJ=-MJM^{-1}J\,.
\end{align}
The spectrum of $\Delta_M$ contains all the basis-invariant information about the relations of the two Gaussian states $J$ and $J_M$. In particular, the expectation value $\braket{J|e^{-\ii \hat{H}}|J}$ will be up to an overall sign fully determined by $\Delta_M$ whenever $e^{-\ii \hat{H}}$ maps the Gaussian state $\ket{J}$ to the Gaussian state $e^{-\ii \hat{H}}\ket{J}=\ket{J_M}$. Such unitaries that map Gaussian states to Gaussian states will be discussed in detail in the subsections~\ref{sec:gaussian-generators} and~\ref{sec:GaussianUnitaries} below.

\subsection{Cartan decomposition}\label{sec:cartan}
The Cartan decomposition of a Lie group $\mathcal{G}$ and associated Lie algebra $\mathfrak{g}$ is a well-known object in Lie theory. For a semi-simple Lie algebra, a Cartan decomposition is commonly defined through an involution $\theta: \mathfrak{g}\to\mathfrak{g}$ with $\theta^2=\id$.\footnote{One usually requires that the bilinear form $B_\theta(K_1,K_2)=-\mathcal{K}(K_1,\theta K_2)$ is positive-definite, where $\mathcal{K}$ is the Killing form. However, we will apply this decomposition also to the compact Lie group $\mathrm{SO}(2N,\mathbb{R})$ for fermions, in which case this condition will not hold.} The eigenspaces $\mathfrak{u}$ and $\mathfrak{p}$ of $\theta$ with eigenvalues $+1$ and $-1$, respectively, then satisfy the Lie bracket relations
\begin{align}
[\mathfrak{u}, \mathfrak{u}] \subseteq \mathfrak{u}\,,\quad [\mathfrak{u}, \mathfrak{p}] \subseteq \mathfrak{p},\,\quad[\mathfrak{p}, \mathfrak{p}] \subseteq \mathfrak{u}\,.
\end{align}
In particular, $\mathfrak{u}$ is a Lie subalgebra of $\mathfrak{g}$ that generates a Lie subgroup $\mathcal{U}\subset\mathcal{G}$, while $\mathfrak{p}$ is not. We further have $\mathfrak{g}=\mathfrak{u}\oplus\mathfrak{p}$, which are orthogonal complements with respect to the Killing form. On the group level, we can decompose a general group element $M=Tu$, where $T\in\exp(\mathfrak{p})$ and $u\in\mathcal{U}$.

In our case, for a fixed reference complex structure $J$, we define the involution\footnote{Note that the minus sign is crucial for $\theta$ to be a Lie algebra automorphism with $[\theta(K_1),\theta(K_2)]=\theta([K_1,K_2])$.} $\theta(K)=-JKJ$, where we identify $\mathfrak{u}(N)$ from~\eqref{eq:unitary-K} as the $+1$ eigenspace $\mathfrak{u}$ and $\mathfrak{u}_\perp(N)$ from~\eqref{eq:u-perp} as the $-1$ eigenspace $\mathfrak{p}$. Given an arbitrary Lie algebra element $K\in\mathfrak{g}$, we can decompose it as
\begin{align}
    K_\pm=\frac{K\mp\theta(K)}{2}=\frac{K\pm JKJ}{2}\,,\label{eq:Cartan-on_frakg}
\end{align}
such that $K_-\in \mathfrak{u}(N)$ and $K_+\in\mathfrak{u}_\perp(N)$.

Given a general group element $M$, we would like to find a $T=e^{K_+}$ with $K_+\in\mathfrak{u}_\perp(N)$ and $u\in\mathrm{U}(N)$, such that $M=Tu$. We can relate this to the relative complex structure $\Delta_M=-MJM^{-1}J$ from~\eqref{eq:RelativeComplexStructure} via
\begin{align}
    T=\sqrt{\Delta_M}\quad\text{and}\quad K_+=\frac{1}{2}\log{\Delta_M}\,,
\end{align}
where we plugged $M=Tu$ into~\eqref{eq:RelativeComplexStructure} and used $uJu^{-1}=J$ as well as $JT^{-1}=TJ$ due to $\{J,K_+\}=0$ to find $T^2=\Delta_M$.

The question is thus for which $\Delta_M$ there exists a canonical square root and logarithm. In~\cite{hackl2021bosonic}, it was shown that $\Delta_M$ is always diagonalizable and its spectrum was carefully studied.
\medskip

\noindent
\textbf{Bosons.} For bosons, the spectrum of $\Delta_M$ consists of pairs of positive real numbers $(e^{\lambda}, e^{-\lambda})$, so that $T=\sqrt{\Delta_M}$ and $K_+=\frac{1}{2}\log{\Delta_M}$ are uniquely defined by taking only positive square roots $(e^{\lambda/2},e^{-\lambda/2})$ and the real logarithm $(\lambda,-\lambda)$. We thus call $M=Tu$ the canonical Cartan decomposition of a bosonic group element with respect to $J$.\footnote{While we could also define an alternative $T$, where we take instead some or all of the square roots to be negative, such a $T$ could not be written as $T=e^{K_+}$, because $K_+$ has purely real eigenvalues for bosons, as it is symmetric with respect to $G$ associated to $J$.}
\medskip

\noindent
\textbf{Fermions.} For fermions, the spectrum of $\Delta_M$ consists of quadruples $(e^{\ii \lambda},e^{\ii \lambda},e^{-\ii \lambda},e^{-\ii \lambda})$ with $\lambda\in(0,\pi)$ and potentially quadruples $(-1,-1,-1,-1)$ and pairs $(1,1)$. While there is no way to make the square root and logarithm unique everywhere, if we exclude group elements $M$ where $\Delta_M$ has $-1$ as an eigenvalue, we can put conditions on $T=e^{K_+}$ to make their choice unique. Specifically, we then require that the eigenvalues of $T$ have positive real part and that the eigenvalues of $K_+$ have modulus less than $\frac{\pi}{2}$.\footnote{Note that $K_+$ has purely imaginary eigenvalues for fermions, as it is anti-symmetric with respect to $G$.} Then for every eigenvalue quadruple $(e^{\ii \lambda},e^{\ii \lambda},e^{-\ii \lambda},e^{-\ii \lambda})$ of $\Delta_M$, $T$ has eigenvalues $(e^{\ii\lambda/2},e^{\ii\lambda/2},e^{-\ii\lambda/2},e^{-\ii\lambda/2})$ and $K_+$ has eigenvalues $(\ii \lambda,\ii\lambda,-\ii\lambda,-\ii\lambda)$ with $\lambda\in[0,\frac{\pi}{2})$. Similarly, for pairs $(1,1)$, the eigenvalues of $T$ and $K_+$ are given by $(1,1)$ and $(0,0)$, respectively. Only in the special case of a $(-1,-1,-1,-1)$, there exist inequivalent ways to take the square root $(e^{\ii\pi/2},e^{\ii\pi/2},e^{-\ii\frac{\pi}{2}},e^{-\ii\pi/2})$, while ensuring that $T$ is real and a special orthogonal group element.\footnote{In particular, we must take the square root, such that there is an equal number of $e^{\pm\ii \pi/2}$ eigenvalues associated to the correct eigenspaces to ensure that $T$ and $K_+$ are real linear maps.} In summary and as proven in~\cite{hackl2021bosonic}, for fermions the Cartan decomposition $M=Tu$ with $T=e^{K_+}\in\exp(\mathfrak{u}_\perp(N))$ and $u\in\mathrm{U}(N)$ exists for all group elements $M\in\mathcal{G}$. If
\begin{align}
    \det(\id+\Delta_M)\neq 0\,,\label{eq:cartan-fermion-defined}
\end{align}
\ie $\Delta_M$ does not have $-1$ as eigenvalue, we can make this choice unique by requiring that $K_+$ is in the set
\begin{align}
    \mathcal{I}_{\mathfrak{u}_\perp(N)}=\{K_+\in\mathfrak{u}_\perp(N)\,|\,\lVert K_+\rVert<\tfrac{\pi}{2}\}\,,\label{eq:I-u-perp}
\end{align}
which restricts the eigenvalues to $\pm\ii \lambda$ with $\lambda\in[0,\frac{\pi}{2})$, such that the eigenvalues of $T$ have positive real part. We will see in section~\ref{sec:principal-fiber-bundles} how the Cartan decomposition with respect to a complex structure $J$ naturally gives rise to understand the group $\mathcal{G}$ as a $\mathrm{U}(N)$ principal fiber bundle over the base manifold $\mathcal{M}$ from~\eqref{eq:def-M-manifold}.

\subsection{Generators of Gaussian unitaries} \label{sec:gaussian-generators}
Let us consider quadratic Hamiltonians of the form
\begin{align}
  \label{eq:Hamiltonian}
    \hat{H}=\left\{\begin{array}{ll}
    \tfrac{1}{2}h_{ab}\hat{\xi}^a\hat{\xi}^b     &  \textbf{(bosons)}\\[2mm]
    \tfrac{\ii}{2}h_{ab}\hat{\xi}^a\hat{\xi}^b     & \textbf{(fermions)}
    \end{array}\right. \,,
\end{align}
where $h_{ab}$ is real and symmetric for bosons and real and antisymmetric for fermions. In other words, they correspond to the most general Hermitian operators that are quadratic in the observables $\hat{\xi}^a$.\footnote{For bosons, any antisymmetric part of $h_{ab}$ would need to be purely imaginary (due to Hermiticity) and gave rise to a constant energy offset given by $\Delta E=\tfrac{\ii}{4}h_{ab}\Omega^{ab}$. For fermions, any symmetric part would also need to be imaginary (due to Hermiticity) and gave rise to a constant energy offset given by $\Delta E=\frac{\ii}{4}h_{ab}G^{ab}$.} We stress that we do not allow for terms linear in the observables $\hat{\xi}^a$, for reasons that will become clear below.

The Hamiltonian is the generator of time-translations and hence of a specific family of unitary transformations. We will see in the next section that the unitaries generated by quadratic Hamiltonians form a group, which we refer to as Gaussian unitaries. However, from the perspective of this paper it is natural to also think about quadratic Hamiltonians~\eqref{eq:Hamiltonian} as being associated with elements of the Lie algebra $\mathfrak{sp}(2N,\mathbb{R})$ (for bosons) or $\mathfrak{so}(2N,\mathbb{R})$ (for fermions) introduced in~\eqref{eq:algebra-g} and they generate transformations that leave the symplectic form $\Omega^{ab}$ (for bosons) or the metric $G^{ab}$ (for fermions) invariant.

This is because the Lie algebra $\mathfrak{g}$ admits an embedding into the Weyl algebra (for bosons) or Clifford algebra (for fermions) generated by the observables $\hat{\xi}^a$. Indeed, for each $K\in\mathfrak{g}$ we can define an associated quadratic element
\begin{align}
  \label{eq:Embedding}
  \widehat{K}
  =\left\{\begin{array}{rl}
    \!\!-\frac{\ii}{2} \omega_{ac}K^c{}_b\hat{\xi}^a\hat{\xi}^b     &  \textbf{(bosons)}\\[2mm]
    \frac{1}{2}g_{ac}K^c{}_b\hat{\xi}^a\hat{\xi}^b     & \textbf{(fermions)}
    \end{array}\right.\;,
\end{align}
where $\omega$ and $g$ are the matrices inverse to $\Omega$ and $G$, respectively. This assignment then induces an embedding of $\mathfrak{g}$ into the Weyl/Clifford algebra and hence an action on Fock space. This embedding is an algebra homomorphism in the sense that
\begin{align}
    [\widehat{K}_1,\widehat{K}_2]=\widehat{[K_1,K_2]}\label{eq:Lie-algebra-representation}
\end{align}
for $K_1,K_2\in\mathfrak{g}$. Put simply, $\widehat{K}$ forms a representation of $\mathfrak{g}$ as anti-Hermitian operators acting on Hilbert space, where each element is symmetric (antisymmetric) in $\hat{\xi}^a$ for bosons (fermions).\footnote{Choosing a different ordering (instead of symmetrization/anti-symmetrization for bosons/fermions, respectively) would spoil the representation~\eqref{eq:Lie-algebra-representation}.\label{ft:spoil-rep}}

We note that $[\widehat{K},\hat{\xi}^a]$ is again linear in the observables $\hat{\xi}^b$. More specifically, $\mathfrak{g}$ has a well-defined action\footnote{The minus sign is a consequence of our choice of $\widehat{K}$, which was made to be consistent with~\eqref{eq:Equivariance}. 
This will allow us to define a group action, which on $\mathcal{U}(M)\ket{\Psi}$ and on operators is $\mathcal{U}^\dagger(M)\mathcal{O}\mathcal{U}(M)$, like time evolution in the Schrödinger or Heisenberg picture, respectively.}
\begin{align}
  \label{eq:EquivarianceInfinitesimal}
  [\widehat{K},\hat{\xi}^a]
  =-{K^a}_b\hat{\xi}^b\,.
\end{align}
This action preserves the Hermiticity of the $\hat{\xi}^a$.

We thus see that there exists a precise relation between quadratic Hamiltonians and elements of the Lie algebra $\mathfrak{g}$. Let us elaborate on this. Given a Hamiltonian $\hat{H}$ as defined in Eq.~\eqref{eq:Hamiltonian}, \ie in terms of the matrix $h_{ab}$, we can associate with it an element $K\in\mathfrak{g}$, defined by
\begin{align}
    K^a{}_b=\left\{\begin{array}{ll}
    \Omega^{ac}h_{cb}     &  \textbf{(bosons)}\\[2mm]
    G^{ac}h_{cb}     &  \textbf{(fermions)}
    \end{array}\right.\,.\label{eq:Kfromh}
\end{align}
We can then use the embedding~\eqref{eq:Embedding} to associate a quadratic operator $\widehat{K}$ with $K$. This operator turns out to satisfy $\widehat{K}=-\ii\hat{H}$.

The mathematical machinery introduced so far allows us to express many quantities, such as $\braket{0|e^{-\ii \hat{H}}|0}$ for some quadratic Hamiltonian $\hat{H}$, only in terms of $2N$-by-$2N$ matrices. 
For instance, the expectation value $\braket{0|e^{-\ii \hat{H}}|0}$ can be expressed as $\braket{J|e^{\widehat{K}}|J}$, which is only depends on a linear complex structure $J$ and a Lie algebra element $K$.
More precisely, any vacuum state $\ket{0}$ is always defined with respect to a set of annihilation operators $\hat{a}_i$ satisfying the standard commutation/anti-commutation relations, \ie $[\hat{a}_i,\hat{a}_j^\dagger]=\delta_{ij}$ and $\{\hat{a}_i,\hat{a}_j^\dagger\}=\delta_{ij}$ for bosons and fermions, respectively. Making such a choice is equivalent to choosing a complex structure $J$, such that its eigenspace $V^*_-$ with eigenvalue $-\ii$ corresponds exactly to the space of annihilation operators, \ie when writing $\hat{a}_i=v_{ia}\hat{\xi}^a$, we have $v_{ia}J^a{}_b=-\ii v_{i b}$. Therefore, we can describe a general vacuum $\ket{0}=\ket{J}$ (Gaussian state) fully in terms of a complex structure $J$.
Furthermore, a general quadratic Hamiltonian $\hat{H}$ of the form~\eqref{eq:Hamiltonian} is in one-to-one correspondence with a symplectic or orthogonal Lie algebra element $K$ from~\eqref{eq:Kfromh}, where we have the relation $-\ii \hat{H}=\widehat{K}$. In this language, the expectation value $\braket{0|e^{-\ii \hat{H}}|0}$ can be written as $\braket{J|e^{\widehat{K}}|J}$.

\subsection{Gaussian unitaries}\label{sec:GaussianUnitaries}
In the previous section, we have established that each purely quadratic Hamiltonian $\hat{H}$ can be linked to an element $K$ in $\mathfrak{g}$. However, we are frequently not so much interested in the Hamiltonian itself but rather in the time evolution generated by it.
  
To obtain the time evolution we need to exponentiate $\widehat{K}=-\ii\hat{H}$ and consider operators of the form $e^{-\ii\hat{H}}$ that we refer to as Gaussian unitaries. However, exponentiating the Fock space operator $\widehat{K}$ is very different from exponentiating the matrix $K$. Elements of the form $M=e^K$ where $K\in\mathfrak{g}$ generate the respective matrix group $\mathcal{G}$ introduced in~\eqref{eq:DefGroup}. In contrast, elements $e^{\widehat{K}}$ live in a different group $\widetilde{\mathcal{G}}$. For bosons, this group $\widetilde{\mathcal{G}}$ is the metaplectic group $\mathrm{Mp}(2N,\mathbb{R})$, a double cover of $\mathrm{Sp}(2N,\mathbb{R})$.\footnote{We note that the metaplectic group does not admit a faithful finite-dimensional linear representation. It hence cannot be presented as a matrix group.} Similarly, when considering fermions, the group is the spin group $\mathrm{Spin}(2N,\mathbb{R})$, a double cover of $\textrm{SO}(2N,\mathbb{R})$. We will now explain in more detail how this double cover group arises.

Using Baker-Campbell-Hausdorff for $e^{\widehat{K}}$, we find the relations
\begin{align}
    e^{-\widehat{K}}\hat{\xi}^ae^{\widehat{K}}&=(e^K)^a{}_b\hat{\xi}^b\,,\\
    e^{-\widehat{K}'}e^{-\widehat{K}}\hat{\xi}^a\,e^{\widehat{K}}e^{\widehat{K}'}&=(e^Ke^{K'})^a{}_b\hat{\xi}^b\,.
\end{align}
This may suggest that we could define $\mathcal{U}(e^K)=e^{\widehat{K}}$ and more generally $\mathcal{U}(M)$ for every group element $M$, such that the relation~\eqref{eq:EquivarianceInfinitesimal} is lifted to the group level as
\begin{align}
    \label{eq:Equivariance}
    \mathcal{U}^\dagger(M)\hat{\xi}^a\mathcal{U}(M)=M^a{}_b\hat{\xi}^b\,,
\end{align}
which seems to indicate that products of $e^{\widehat{K}}$ generate a representation of the group $\mathcal{G}$. However, it is easy to see that~\eqref{eq:Equivariance} characterizes $\mathcal{U}(M)$ only up to a complex phase, \ie for a given $\mathcal{U}(M)$ satisfying~\eqref{eq:Equivariance} also $e^{\ii \varphi}\mathcal{U}(M)$ will satisfy it.\footnote{Apart from this freedom in the complex phase, relation~\eqref{eq:Equivariance} does characterize $\mathcal{U}(M)$ uniquely, due to the representation of $\hat{\xi}^a$ being faithful, as explained in proposition 6 of~\cite{hackl2021bosonic}.} However, in the following we will argue that $\mathcal{U}(M)$ can only be defined as projective representation of $\mathcal{G}$, where $\mathcal{U}(M)$ is only specified up to an overall sign.

The space of anti-Hermitian quadratic operators $\widehat{K}$ forms a finite-dimensional Lie subalgebra of the space all anti-Hermitian operators, which itself forms the Lie algebra associated to the unitary group\footnote{For bosons, this is the unitary group of the unique infinite dimensional separable Hilbert space. For fermions, this is the unitary group $\mathrm{U}(2^N)$, as the fermionic Hilbert space has $\dim\mathcal{H}=2^N$.} $\mathrm{U}(\mathcal{H})$ of Hilbert space. Exponentiating elements $\widehat{K}$ of this Lie algebra representation will give rise to a representation of a Lie subgroup of the unitary group, whose Lie algebra is $\mathfrak{g}$. However, for a given Lie algebra $\mathfrak{g}$, there exist in general many associated Lie groups that recover $\mathfrak{g}$ from the infinitesimal group action on the tangent space at the identity. For bosons with $\mathfrak{g}=\mathfrak{sp}(2N,\mathbb{R})$, the Lie group candidates consist of the symplectic group $\mathrm{Sp}(2N,\mathbb{R})$ and its $n$-fold covers, including its universal cover with $n=\infty$. For fermions with $\mathfrak{g}=\mathfrak{so}(2N,\mathbb{R})$, the Lie group candidates are either the special orthogonal group $\mathrm{SO}(2N,\mathbb{R})$ or its universal cover $\mathrm{Spin}(2N,\mathbb{R})$.\footnote{Let us note that both $\mathrm{Sp}(2N,\mathbb{R})$ and $\mathrm{SO}(2N,\mathbb{R})$ have non-trivial centers consisting of $(\id,-\id)$ forming $\mathbb{Z}_2$, which implies that the ``smallest'' Lie group with Lie algebra $\mathfrak{g}$ is actually given by $\mathcal{G}/\mathbb{Z}_2$, \ie by $\mathrm{PSp}(2N,\mathbb{R})$ and $\mathrm{PSO}(2N,\mathbb{R})$, respectively. The only exception is $N=1$ for fermions, where $\mathcal{G}=\mathrm{SO}(2,\mathbb{R})$ is the unique compact Lie group with one-dimensional real Lie algebra.} These candidates are closely related to the fundamental groups, namely $\pi_1(\mathrm{Sp}(2N,\mathbb{R}))=\mathbb{Z}$ and $\pi_1(\mathrm{SO}(2N))=\mathbb{Z}_2$. Note that by construction, we will only get the group elements connected to the identity, as we are only multiplying exponentiated Lie algebra elements with each other.

We can determine which group the exponentials of the form $e^{\widehat{K}}$ generate by considering how loops in $\mathcal{G}$ are lifted. It turns out that for this, it suffices to consider a single bosonic or fermionic degree of freedom, as the same argument also applies to larger systems, as we can always embed such a single degree of freedom into a larger system. For both bosons and fermions, we consider
\begin{align}
    K=\begin{pmatrix}
        0 & 1\\         -1 & 0
    \end{pmatrix}\,,
\end{align}
with respect to the basis $(\hat{q},\hat{p})$. We compute the associated operator\footnote{The offset $\pm\frac{1}{2}$ in $\widehat{K}=-\ii(\hat{a}^\dagger\hat{a}\pm \frac{1}{2})$ is a consequence of the symmetrization/anti-symmetrization in $\hat{\xi}^a$. Trying to remove it by adding/subtracting an appropriate offset will spoil the commutation relations~\eqref{eq:Lie-algebra-representation} to close, as discussed in footnote~\ref{ft:spoil-rep}.}
\begin{align}
    \widehat{K}=\begin{cases}
    \ii(\hat{n}+\frac{1}{2})    & \textbf{(bosons)}\\
    \ii(\hat{n}-\frac{1}{2})    & \textbf{(fermions)}
    \end{cases}\,,
\end{align}
where $\hat{n}=\hat{a}^\dagger\hat{a}$ with $\hat{a}=\frac{1}{\sqrt{2}}(\hat{q}+\ii\hat{p})$. We now consider the trajectory $U(t)=e^{t\widehat{K}}$ with varying $t$, where we need to choose $t\in[0,4\pi]$ for the trajectory to form a closed loop starting and ending at the identity. In contrast, $M(t)=e^{tK}$ will already return at $t=2\pi$ to the identity $M(2\pi)=\id$ and we can use the group property $M(2\pi+t)=M(2\pi)M(t)=M(t)$ to see that $M(t)$ will run over the same loop twice, when $U(t)$ only performs a single loop. This implies that for both bosons and fermions, $e^{\widehat{K}}$ generates a representation of the double cover $\widetilde{\mathcal{G}}$ of $\mathcal{G}$, \ie the metaplectic group $\mathrm{Mp}(2N,\mathbb{R})$ for bosons and the spin group $\mathrm{Spin}(2N,\mathbb{R})$ for fermions. This argument sketches the idea of a general proof.\footnote{Loops in both $\mathrm{Sp}(2N,\mathbb{R})$ and $\mathrm{SO}(2N,\mathbb{R})$ can be characterized by a single winding number (either in $\mathbb{Z}$ or $\mathbb{Z}_2$, respectively). By establishing that a non-contractible loop in $\mathcal{G}$, such as $M(t)=e^{tK}$ with $t\in[0,2\pi]$, is covered twice by the respective loop in $U(t)=e^{t\widehat{K}}$ with $t\in[0,4\pi]$ in $\widetilde{\mathcal{G}}$, we establish that $\widetilde{\mathcal{G}}$ is the double cover of $\mathcal{G}$. For larger $N$, we can still use this argument where $\widehat{K}$ is just constructed from the first bosonic or fermionic degree of freedom in the system.}

Given a group element $M\in\mathcal{G}$, there thus exist two unitary operators $\pm\mathcal{U}(M)$ in the set generated by $e^{\widehat{K}}$ for arbitrary $\widehat{K}$. Slightly abusing notation\footnote{Indeed, when dealing with projective representations, $\mathcal{U}(M)$ does not refer to a single operator, but to the set of operators related by multiplication with the respective complex phases $e^{\ii \varphi}$, \ie $\pm 1$ in our case of the double cover.}, $\mathcal{U}(M)$ is thus only defined up to a sign\footnote{We conclude this from the fact that a double cover requires phases $e^{\ii \varphi}$ forming $\mathbb{Z}_2$, which yields $\pm 1$. However, this is also proven explicitly for bosons in chapter 4, Theorem (4.37) of~\cite{folland1989harmonic}. 
} and we have the relation
\begin{align}
    \mathcal{U}(M_1)\mathcal{U}(M_2)=\pm\mathcal{U}(M_1M_2)\label{eq:U(M)-product}
\end{align}
describing a projective representation of $\mathcal{G}$. In section~\ref{sec:rep-double-cover}, we will turn this into a proper representation $\mathcal{U}(M,\psi)$, where $(M,\psi)\in\widetilde{\mathcal{G}}$ encodes group elements in the respective double cover.

The name Gaussian unitaries stems from the fact that $\mathcal{U}(M)$ maps Gaussian states onto Gaussian states. This can be seen by plugging in $\mathcal{U}(M)\ket{J}$ into~\eqref{eq:z-displacement} and using~\eqref{eq:Equivariance} to find that $J$ is mapped to $MJM^{-1}$. We therefore have
\begin{align}
    \mathcal{U}(M)\ket{J}=e^{\ii \varphi}\ket{MJM^{-1}}\,,
\end{align}
\ie $\mathcal{U}(M)$ acting on $\ket{J}$ gives a Gaussian state vector $\ket{MJM^{-1}}$. The complex phase $\varphi$ depends on which state vector representatives $\ket{J}$ and $\ket{MJM^{-1}}$ we choose. For $u\in\mathrm{U}(N)$, we have
\begin{align}
    \mathcal{U}(u)\ket{J}=e^{\ii\varphi}\ket{J}\,,\label{eq:U(u)-relation}
\end{align}
\ie $\ket{J}$ is an eigenvector of $\mathcal{U}(u)$. Here, the complex phase is independent of our choice of state vector representative $\ket{J}$, as it really compared $\ket{J}$ with $\mathcal{U}(u)\ket{J}$. Note that we still have an overall sign ambiguity due to $\mathcal{U}(M)$ being only defined up to an overall sign. We can also determine how $\mathcal{U}(M)$ acts on the quadratic operator $\widehat{K}$ to find
\begin{align}
    \mathcal{U}(M)\widehat{K}\mathcal{U}^\dagger(M)=\widehat{MKM^{-1}}\,,\label{eq:unitary-K}
\end{align}
where one should note the opposite order of $\mathcal{U}(M)$ and $\mathcal{U}^\dagger(M)$, compared to~\eqref{eq:Equivariance}. In this case, the sign ambiguity is not present as~\eqref{eq:unitary-K} is invariant under changing the sign of $\mathcal{U}(M)$.

\subsection{Principal fiber bundle}\label{sec:principal-fiber-bundles}
We will use the Cartan decomposition, as introduced in section~\ref{sec:cartan}, to understand the Lie group $\mathcal{G}$ as a principal fiber bundle with right action by the subgroup $\mathrm{U}(N)$. This will allow us to get a more intricate understanding of the group and its structure, which is especially helpful when trying to build its double cover explicitly. Note that both the Cartan decomposition and the subgroup $\mathrm{U}(N)$ are defined with respect to a choice of reference complex structure $J$, so everything discussed here relies on this initial choice.

We consider $\mathcal{G}$ as principal fiber bundle with base manifold $\mathcal{M}$ from~\eqref{eq:def-M-manifold} and right action $\triangleleft: \mathcal{G}\times\mathrm{U}(N)\to\mathcal{G}$ of the stabilizer group $\mathrm{U}(N)$ from~\eqref{eq:U-group} with $M\triangleleft u=Mu$. Given a group element $M\in\mathcal{G}$, we have the projection $\pi(M)=J_M=MJM^{-1}\in\mathcal{M}$ onto the base manifold. The associated fiber over $J_M\in\mathcal{M}$ is then $\pi^{-1}(J_M)=\{Mu|u\in\mathrm{U}(N)\}$, \ie the orbit of $M$ under the right action $\mathrm{U}(N)$. This should not surprise, as $\mathcal{M}$ in~\eqref{eq:def-M-manifold} was constructed as the equivalence class of group elements related by right-multiplication of $u\in\mathrm{U}(N)$, so the equivalence classes that make up $\mathcal{M}$ are exactly the fibers of $\mathcal{G}$. This construction has been used to define and study circuit complexity of Gaussian unitaries in the context of bosonic and fermionic field theories~\cite{hackl2018circuit,chapman2019complexity,torres2024toward}, but it can also be used to optimize functions over the manifold of Gaussian states~\cite{windt2021local}.

We can recognize the Cartan decomposition as a way to select for each fiber $\{Mu|u\in\mathrm{U}(N)\}$ a unique element $T=\sqrt{-MJM^{-1}J}$, which works everywhere for bosons and almost everywhere for fermions, as discussed in section~\ref{sec:cartan}. From the fiber bundle perspective, we can understand the decomposition $\mathfrak{g}=\mathfrak{u}(N)\oplus\mathfrak{u}_\perp(N)$ as a natural orthogonal decomposition of the tangent space at $\id$, where $\mathfrak{u}(N)$ corresponds to the vertical direction parallel to $\mathrm{U}(N)$, while $\mathfrak{u}_\perp(N)$ represents the horizontal direction cutting through different fibers, which is orthogonal to $\mathfrak{u}(N)$ with respect to the Killing form on $\mathfrak{g}$.

Moving into the different directions $K_+\in\mathfrak{u}_\perp(N)$ will at least locally intersect with each fiber just once. We can consider all group elements on the trajectories $e^{t K_+}$ where we move into the direction $K_+$ as we increase $t$. This set is given by
\begin{align}
    \exp(\mathfrak{u}_\perp(N))=\{e^{K_+}|K_+\in\mathfrak{u}_\perp(N)\}
\end{align}
and we can consider its properties for bosons and fermions separately.

\textbf{Bosons.} For bosons, $\{K_+,J\}=0$ implies that $K_+G-GK_+^\intercal=0$, which means that $K_+$ is represented by a symmetric matrix in a basis, where the covariance matrix $G$ of the reference state $\ket{J}$ is given by the identity. Thus, $K_+$ is diagonalizable and its exponential $e^{K_+}$ will also be symmetric in such a basis and its eigenvalues will be positive. Vice versa, for any symplectic group element $T$ with purely positive eigenvalues that is symmetric with respect to $G$, there exists a unique logarithm $K_+=\log(T)\in \mathfrak{u}_\perp(N)$. Thus, the two sets $\exp(\mathfrak{u}_\perp(N))$ and $\mathfrak{u}_\perp(N)$ are diffeomorphic to each other and to $\mathbb{R}^{N(N+1)}$, as $\mathfrak{u}_\perp(N)$ is a real vector space of this dimension. Recall that the Cartan decomposition of a symplectic group element $M=Tu$ with respect to a complex structure $J$ is unique, we thus see that $\exp(\mathfrak{u}_\perp(N))$ intersects each fiber once and is thus itself diffeomorphic to the base manifold $\mathcal{M}$ for bosons.

\textbf{Fermions.} For fermions, $K_+$ is also diagonalizable, as all elements of $\mathfrak{g}$ are antisymmetric with respect to $G$, but in contrast to bosons $K_+$ will have purely imaginary eigenvalues. Therefore, the exponential map will fold $\mathfrak{u}_\perp(N)$ infinitely many times over itself when mapping to $\mathcal{G}$, which should not come as a surprise as $\mathcal{G}$ is compact for fermions.
From our previous analysis in section~\ref{sec:cartan}, we already know that there exists a unique Cartan decomposition of a group element $M\in\mathcal{G}$ if and only if $\Delta_M$ does not have $-1$ as eigenvalue, which is equivalent to requiring $\det(\id+\Delta_M)\neq 0$. For those, we can enforce uniqueness of $T$ by requiring that its eigenvalues lie in the set $e^{\pm \ii\theta}$ with $\theta\in[0,\frac{\pi}{2})$, which corresponds to the unique $K_+=\log(T)=\frac{1}{2}\log(\Delta_M)$ that lies in the set $\mathcal{I}_{\mathfrak{u}_\perp(N)}$ from~\eqref{eq:I-u-perp}. Vice versa, for each choice of $K_+\in \mathcal{I}_{\mathfrak{u}_\perp(N)}$, we will have $T=e^{K_+}$ that will appear in the Cartan decomposition of $M=Tu$ for arbitrary $u\in\mathrm{U}(N)$. Therefore, $\exp(\mathcal{I}_{\mathfrak{u}_\perp(N)})$ describes the set of all $T$, in Cartan decompositions of $M$ with $\det(\id+\Delta_M)\neq 0$.

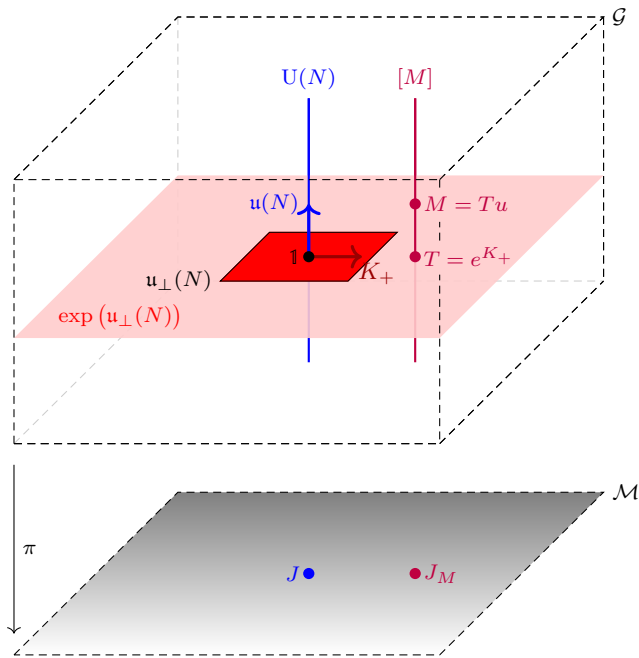
\begin{figure}[t]
    \centering
    \tdplotsetmaincoords{90}{90}
		\tdplotsetrotatedcoords{0}{20}{70}
            \tikzset{zxplane/.style={canvas is zx plane at y=#1,very thin}}
            \tikzset{yxplane/.style={canvas is yx plane at z=#1,very thin}}
		\begin{tikzpicture}[scale=1.4]
		\tikzset{->-/.style={decoration={
					markings,
					mark=at position #1 with {\arrow{>}}},postaction={decorate}}}
		\coordinate (A) at (-2,0,2);
		\coordinate (B) at (-2,0,-2);
		\coordinate (C) at (2,0,-2);
		\coordinate (D) at (2,0,2);
		\coordinate (S) at (0,-2.5,0);
		\coordinate (E) at ($(A)+(S)$);
		\coordinate (F) at ($(B)+(S)$);
		\coordinate (G) at ($(C)+(S)$);
		\coordinate (H) at ($(D)+(S)$);
		\coordinate (S2) at (0,-2,0);
		\coordinate (E2) at ($(A)+(S)+(S2)$);
		\coordinate (F2) at ($(B)+(S)+(S2)$);
		\coordinate (G2) at ($(C)+(S)+(S2)$);
		\coordinate (H2) at ($(D)+(S)+(S2)$);
		\coordinate (S3) at (0,-1.5,0);
            \coordinate (V) at (0,-1.5,0);
		\coordinate (E3) at ($(A)+(S3)$);
		\coordinate (F3) at ($(B)+(S3)$);
		\coordinate (G3) at ($(C)+(S3)$);
		\coordinate (H3) at ($(D)+(S3)$);
            \coordinate (P1) at (0,-.2,0); 
            \coordinate (P2) at (0,.2,0); 
            \draw[->] ($(E)+(P1)$) -- node[right,font=\footnotesize]{$\pi$} ($(E)+ (S2)+(P2)$);
		\draw[densely dashed] (A) -- (B) -- (C) node[right,font=\footnotesize]{$\mathcal{G}$} -- (D);
		\draw[densely dashed] (E) -- (H);
		\draw[densely dashed] (C) -- (G);
		\draw[densely dashed] (H) -- (G);
		\draw[densely dashed,opacity=0.2] (E) -- (F) -- (G);
		\draw[densely dashed,opacity=0.2] (B) -- (F);
		\shadedraw[densely dashed] (E2) -- (F2) -- (G2) node[right,font=\footnotesize]{$\mathcal{M}$} -- (H2) -- cycle;
            
		\draw[blue, thick] (0,-2.5,0) -- (0,0,0) node[above,blue,font=\footnotesize]{$\mathrm{U}(N)$};
            \draw[purple, thick] (1,-2.5,0) -- (1,0,0) node[above,font=\footnotesize]{$[M]$};
		\fill[lred,opacity=.9] (E3) node[above,xshift=14mm,red,opacity=1,font=\footnotesize]{$\exp\big(\mathfrak{u}_\perp(N)\big)$} -- (F3) -- (G3) -- (H3) -- cycle;
            \begin{scope}[zxplane=-1.5]
                \filldraw[red] (.6,-.6) node[left,font=\footnotesize,black]{$\mathfrak{u}_\perp(N)$} rectangle (-.6,.6);
                \draw[black] (.6,-.6) rectangle (-.6,.6);
                \draw[very thick, dred,->] (0,0) -- (0,.5) node[right=-1.8mm,yshift=-2.2mm]{$K_+$};
            \end{scope}
		\draw[blue, thick] (0,0,0) -- (0,-1.5,0);
            \draw[purple, thick] (1,0,0) -- (1,-1.5,0);
            \draw[blue, very thick,<-] (0,-1,0) node[left,font=\footnotesize]{$\mathfrak{u}(N)$} -- (0,-1.5,0);
            \fill [blue] ($(0,0,0)+(S)+(S2)$) node[left,font=\footnotesize]{$J$} circle (1.5pt);
            \fill [purple] ($(1,0,0)+(S)+(S2)$) node[right,font=\footnotesize]{$J_M$} circle (1.5pt);
            \fill [black] ($(0,0,0)+(S3)$) node[left,font=\footnotesize]{$\id$} circle (1.5pt);
            \fill [purple] ($(1,0,0)+(S3)$) node[right,font=\footnotesize]{$T=e^{K_+}$} circle (1.5pt);
            \fill [purple] ($(1,-1,0)$) node[right,font=\footnotesize]{$M=Tu$} circle (1.5pt);
		\draw[densely dashed] (A) -- (B) -- (C) -- (D);
		\draw[densely dashed] (E) -- (H);
		\draw[densely dashed] (C) -- (G);
		\draw[densely dashed] (H) -- (G);
            \draw[densely dashed] (E) -- (A) -- (D);
            \begin{scope}
            \clip [] ($(D)-(-.1,0,0)$) rectangle ($(D)-(.1,.1,0)$) ($(D)-(-.1,.4,0)$) rectangle ($(D)-(.1,.62,0)$) ($(D)-(-.1,.9,0)$) rectangle ($(H)-(.1,0,0)$);
            \draw[densely dashed] (D) -- (H);
            \end{scope}
		\end{tikzpicture}
    \caption{\emph{Bosonic fiber bundle.} We illustrate the group $\mathcal{G}=\mathrm{Sp}(2N,\mathbb{R})$ as principal fiber bundle with group action of $\mathrm{U}(N)$ and base manifold $\mathcal{M}=\mathcal{G}/\mathrm{U}(N)$. The fibers are the equivalence classes $[M]$ by the relation~\eqref{eq:equivalence-relation}, which we can uniquely characterized by their complex structure $J_M=MJM^{-1}$. $J$ is an a priori chosen reference complex structure. We illustrate the Lie algebra as the tangent space at the identity, which we can decompose as $\mathfrak{g}=\mathfrak{u}(N)\oplus\mathfrak{u}_\perp(N)$ into a direct sum of the ``vertical'' $\mathfrak{u}(N)$ subalgebra tangential to the fiber $\mathrm{U}(N)$ and its orthogonal complement $\mathfrak{u}_\perp(N)$ being ``horizontal''. The exponential map provides a diffeomorphsim from $\mathfrak{u}_\perp(N)$ to the set $\exp(\mathfrak{u}_\perp(N))$, which intersects each fiber once. For a fiber $[M]=\{Mu|u\in\mathrm{U}(N)\}$, this intersection point is $T=\sqrt{-MJM^{-1}J}$ appearing in the Cartan decomposition $M=Tu$. While the vertical direction of $\mathrm{U}(N)$ is compact, both the base manifold $\mathcal{M}$ and $\exp(\mathfrak{u}_\perp(N))$ are $N(N+1)$-dimensional and non-compact (dashed lines indicate that the direction continues).}
    \label{fig:bosonic-fiber-bundle}
\end{figure}

\begin{figure}[t]
    \centering
    \tdplotsetmaincoords{90}{90}
		\tdplotsetrotatedcoords{0}{20}{70}
            \tikzset{zxplane/.style={canvas is zx plane at y=#1,very thin}}
            \tikzset{yxplane/.style={canvas is yx plane at z=#1,very thin}}
		\begin{tikzpicture}[scale=1.4]
		\tikzset{->-/.style={decoration={
					markings,
					mark=at position #1 with {\arrow{>}}},postaction={decorate}}}
            \coordinate (O) at (0,0,0);
            \coordinate (s) at ({2*cos(22)},0,{-2*sin(22)});
		\coordinate (A) at (-2,0,2);
		\coordinate (B) at (-2,0,-2);
		\coordinate (C) at (2,0,-2);
		\coordinate (D) at (2,0,2);
		\coordinate (S) at (0,-2.5,0);
		\coordinate (E) at ($(A)+(S)$);
		\coordinate (F) at ($(B)+(S)$);
		\coordinate (G) at ($(C)+(S)$);
		\coordinate (H) at ($(D)+(S)$);
		\coordinate (S2) at (0,-2,0);
		\coordinate (E2) at ($(A)+(S)+(S2)$);
		\coordinate (F2) at ($(B)+(S)+(S2)$);
		\coordinate (G2) at ($(C)+(S)+(S2)$);
		\coordinate (H2) at ($(D)+(S)+(S2)$);
		\coordinate (S3) at (0,-1.5,0);
		\coordinate (E3) at ($(A)+(S3)$);
		\coordinate (F3) at ($(B)+(S3)$);
		\coordinate (G3) at ($(C)+(S3)$);
		\coordinate (H3) at ($(D)+(S3)$);
            \coordinate (P1) at (0,-.3,0); 
            \coordinate (P2) at (0,.3,0); 
            \draw[->] ($(S)-(s)+(P1)$) -- node[right,font=\footnotesize]{$\pi$} ($(S)-(s)+ (S2)+(P2)$);
                \phantom{
    		\draw (A) -- (B) -- (C) node[right,font=\footnotesize]{$\mathcal{G}$} -- (D) -- cycle;
    		\draw (E) -- (H);
    		\draw (C) -- (G);
    		\draw (H) -- (G);
                \draw (A) rectangle (H);
    		\draw[opacity=0.2] (E) -- (F) -- (G);
    		\draw[opacity=0.2] (B) -- (F);
    		\draw (E2) -- (F2) -- (G2) node[right,font=\footnotesize]{$\mathcal{M}$} -- (H2) -- cycle;
                }
                \begin{scope}[zxplane=0]
                   \draw[orange,thick] (0,0) circle[radius=2cm] ;
                \end{scope}
                \begin{scope}[zxplane=-2.5]
                   \draw[orange,thick] (0,0) circle[radius=2cm] ;
                \end{scope}
                \begin{scope}[zxplane=-4.5]
                   \shadedraw[orange,thick] (0,0) circle[radius=2cm] ;
                \end{scope}
                \draw[orange,thick] (s) -- ($(s)+(S)$);
                \draw[orange,thick] ($(O)-(s)$) -- ($(S)-(s)$);
		\draw[blue, thick] (0,-2.5,0) -- (0,0,0) node[above,blue,font=\footnotesize]{$\mathrm{U}(N)$};
            \draw[purple, thick] (1,-2.5,0) -- (1,0,0) node[above,font=\footnotesize]{$[M]$};
            \begin{scope}[zxplane=-1.5]
               \fill[lred,opacity=.9] (0,0) circle[radius=2cm];
               \draw[dgreen,thick] (0,0) circle[radius=2cm];
            \end{scope}
            \begin{scope}[zxplane=-1.5]
                \draw[red] (.6,-.6) node[left,font=\footnotesize,black]{$\mathfrak{u}_\perp(N)$} rectangle (-.6,.6);
                \fill[red] (0,0) circle[radius=.6cm];
                \draw[dgreen,thick] (0,0) circle[radius=.6cm];
                \draw (-.6,0) node[above,dgreen,xshift=4mm,yshift=0mm,font=\footnotesize]{$\mathfrak{B}_{\mathfrak{u}_\perp(N)}$};
                \draw[black] (.6,-.6) rectangle (-.6,.6);
                \draw[very thick, dred,->] (0,0) -- (0,.5) node[right=-1.8mm,yshift=-2.2mm]{$K_+$};
                \draw[black] (.2,-.2) -- (-.3,-.8) node[left,red]{$\mathcal{I}_{\mathfrak{u}_\perp(N)}$};
            \end{scope}
		\draw[blue, thick] (0,0,0) -- (0,-1.5,0);
            \draw[purple, thick] (1,0,0) -- (1,-1.5,0);
            \draw[blue, very thick,<-] (0,-1,0) node[left,font=\footnotesize]{$\mathfrak{u}(N)$} -- (0,-1.5,0);
            \fill [blue] ($(0,0,0)+(S)+(S2)$) node[left,font=\footnotesize]{$J$} circle (1.5pt);
            \draw[red] ($(0,0,1.5)+(S3)$) node[font=\footnotesize]{$\exp(\mathcal{I}_{\mathfrak{u}_\perp(N)})$};
            \draw ($(-1,0,0)+(S)+(S2)$) node[font=\footnotesize]{$\mathcal{I}_{\mathcal{M}}$} ($(-2,0,-1.6)+(S)+(S2)$) node[font=\footnotesize,orange]{$\mathcal{B}_{\mathcal{M}}$};
            \draw ($(-1,0,0)$) node[font=\footnotesize]{$\mathcal{I}_{\mathcal{G}}$} ($(-2,0,-1.6)$) node[font=\footnotesize,orange]{$\mathcal{B}_{\mathcal{G}}$};
            \draw[dgreen] ($(0,0,1.5)+(S3)$) node[font=\footnotesize,below,yshift=-2mm,xshift=-3mm]{$\exp(\mathcal{B}_{\mathfrak{u}_\perp(N)})$};
            \fill [purple] ($(1,0,0)+(S)+(S2)$) node[right,font=\footnotesize]{$J_M$} circle (1.5pt);
            \fill [black] ($(0,0,0)+(S3)$) node[left,font=\footnotesize]{$\id$} circle (1.5pt);
            \fill [purple] ($(1,0,0)+(S3)$) node[right,font=\footnotesize]{$T=e^{K_+}$} circle (1.5pt);
            \fill [purple] ($(1,-1,0)$) node[right,font=\footnotesize]{$M=Tu$} circle (1.5pt);
            \begin{scope}[zxplane=0]
                   \draw[orange,thick] (0,0) circle[radius=2cm] ;
                \end{scope}
            \draw (s) node[right,font=\footnotesize]{$\mathcal{G}$};
            \draw ($(s)+(S)+(S2)$) node[right,font=\footnotesize]{$\mathcal{M}$};
		\end{tikzpicture}
    \caption{\emph{Fermionic fiber bundle.} The structure of the fermionic fiber bundle is almost identical to the bosonic one, but while the bosonic base manifold (and group) are non-compact in the horizontal direction, the fermionic one is compact. We illustrate this by drawing quasi-boundaries $\mathcal{B}_\mathcal{M}$ and $\mathcal{B}_{\mathcal{G}}$, which represent points where different parts of the manifolds are glued together leaving the overall manifolds without actual boundary (thus quasi-boundaries). In particular, we have $\dim\mathcal{B}_{\mathcal{M}}=\dim\mathcal{I}_{\mathcal{M}}-2$, so in our picture where $\mathcal{I}_{\mathcal{M}}$ is illustrated as a disk, its quasi boundary $\mathcal{B}_{\mathcal{M}}$ would be zero-dimensional, so we should think of the surrounding circle as being identified to a single point in the picture. It then does not come as a surprise that $\mathcal{B}_{\mathfrak{u}_\perp(N)}$, illustrated as a circle, actually runs parallel to the fiber over this point. We also show the action of the exponential map to $\mathcal{I}_{\mathfrak{u}_\perp(N)}$ its boundary to get the respective sets in the group. Both, the base manifold $\mathcal{M}$ and $\mathfrak{u}_\perp(N)$ are $N(N-1)$-dimensional and compact (drawn boundaries indicate where the manifold is glued together).}
    \label{fig:fermionic-fiber-bundle}
\end{figure}
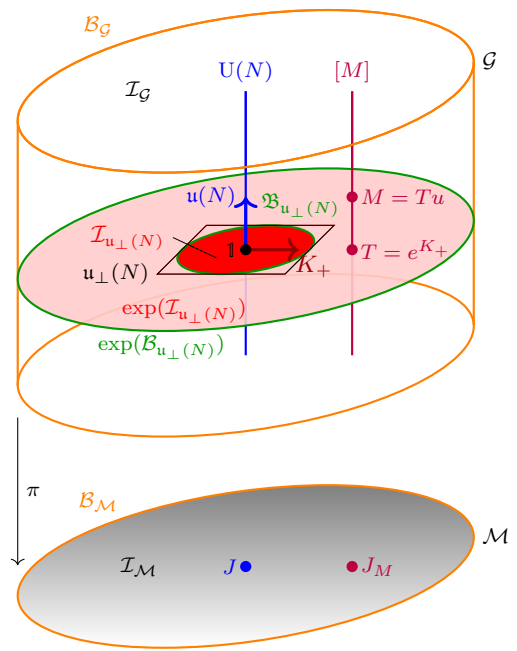

For fermions, we can thus decompose both $\mathcal{G}$ and its base manifold $\mathcal{M}$ into inner parts
\begin{align}
    \mathcal{I}_{\mathcal{G}}&=\{M\in\mathcal{G}|\det(\id+\Delta_M)\neq0\}\,,\\
    \mathcal{I}_{\mathcal{M}}&=\{\tilde{J}\in\mathcal{M}|\det(\id-\tilde{J}J)\neq0\}\,,
\end{align}
which are directly related to $\mathcal{I}_{\mathfrak{u}_\perp(N)}$ via the exponential map, \ie writing $\tilde{J}=e^{K_+}Je^{-K_+}=e^{2K_+}J$ provides a diffeomorphism from $\mathcal{I}_{\mathfrak{u}_\perp(N)}$ to $\mathcal{I}_{\mathcal{M}}$ and $\exp(\mathcal{I}_{\mathfrak{u}_\perp(N)})$ intersects with each fiber in $\mathcal{I}_{\mathcal{G}}$ exactly once. We can also introduce the respective quasi-boundaries\footnote{We refer to these sets as quasi-boundaries, as we glue them non-trivially onto the inner parts to make up the full manifold, \ie $\mathcal{G}$ or $\mathcal{M}$, respectively. For example, we could decompose the two-sphere into an open disk and a single point that can be glued back together to give $S^2$.}
\begin{align}
    \mathcal{B}_{\mathcal{G}}&=\{M\in\mathcal{G}|\det(\id+\Delta_M)=0\}\,,\label{eq:B_G}\\
    \mathcal{B}_{\mathcal{M}}&=\{\tilde{J}\in\mathcal{M}|\det(\id-\tilde{J}J)=0\}\,,
\end{align}
such that $\mathcal{G}=\mathcal{I}_{\mathcal{G}}\cup\mathcal{B}_{\mathcal{G}}$ and $\mathcal{M}=\mathcal{I}_{\mathcal{M}}\cup\mathcal{B}_{\mathcal{M}}$. While $\mathcal{I}_{\mathcal{M}}$ is contractible, the non-trivial topology of $\mathcal{M}$ is arises from the way we attach the quasi-boundaries. Consequently, $\mathcal{I}_{\mathcal{G}}$ is a trivial $\mathrm{U}(N)$ fiber bundle over $\mathcal{I}_{\mathcal{M}}$, \ie $\mathcal{I}_{\mathcal{G}}=\mathcal{I}_{\mathcal{M}}\times\mathrm{U}(N)$.

It is instructive to consider what happens at the quasi boundaries. From the perspective of $\mathfrak{u}_\perp(N)$, we can define the boundary of the inner part $\mathcal{I}_{\mathfrak{u}_\perp(N)}$ from~\eqref{eq:I-u-perp} as
\begin{align}
    \mathcal{B}_{\mathfrak{u}_\perp(N)}=\{K_+\in\mathfrak{u}_\perp(N)\,|\,\lVert K_+\rVert_{\infty}=\tfrac{\pi}{2}\}\,,
\end{align}
which consists of those elements of $\mathfrak{u}_\perp(N)$ that have $\pm\ii \frac{\pi}{2}$ as largest eigenvalues. The eigenvalues of $K_+$ are purely imaginary and come in conjugate pairs of multiplicity two.\footnote{A quick way to see this is to recognize that $K_+\in\mathfrak{u}_\perp(N)$ can be written as $K_+=\begin{pmatrix}
    A & B\\
    B & -A
\end{pmatrix}=\begin{pmatrix}
    0 & \id\\
    -\id & 0
\end{pmatrix}\begin{pmatrix}
    B & -A\\
    -A & -B
\end{pmatrix}$ in the basis of~\eqref{eq:xi-basis}, where $A$ and $B$ are antisymmetric matrices. By the so-called Stenzel condition~\cite{ikramov2009product} all non-zero eigenvalues then appear with even multiplicity. As $K_+$ is real and antisymmetric, all eigenvalues are also imaginary and come in conjugate pairs.} Consequently, $\mathcal{B}_{\mathfrak{u}_\perp(N)}$ consists of $K_+$ that have at least one eigenvalue quadruple $(\ii\frac{\pi}{2},\ii\frac{\pi}{2},-\ii\frac{\pi}{2},-\ii\frac{\pi}{2})$. Assuming the generic case of just a single quadruple, the associated four eigenvectors will span the four-dimensional eigenspace of $\Delta=e^{2K_+}$ with eigenvalue $-1$. We can ask how many $K_+$ give rise to the same $\Delta$, such that $T=e^{K_+}$ lies in the same fiber. Focusing just on the respective four-dimensional eigenspace of $\Delta$ and expressing everything in a standard basis where $J$ takes the form~\eqref{eq:J-standard-form}, the respective block in $K_+\in\mathcal{B}_{\mathcal{G}}$ must have the form
\begin{align}
    \frac{\pi}{2}\left(\begin{array}{cccc}
 0 & \cos (\phi ) & 0 & \sin (\phi ) \\
 -\cos (\phi ) & 0 & -\sin (\phi ) & 0 \\
 0 & \sin (\phi ) & 0 & -\cos (\phi ) \\
 -\sin (\phi ) & 0 & \cos (\phi ) & 0
\end{array}\right)\,,
\end{align}
which is a one-dimensional family parametrized by $\phi\in[0,2\pi)$. As all of these $K_+$ give rise to the same $\Delta$, this family is tangential to the fiber over the base point $J_{e^{K_+}}=e^{K_+}Je^{-K_+}=e^{2K_+}J$. This means the projection $\mathcal{B}_{\mathcal{M}}$ of the set $\exp(\mathcal{B}_{\mathfrak{u}_\perp(N)})$ will be of lower dimension, as the aforementioned vertical direction will be projected out when going to the base manifold. At a generic point of $\dim\mathcal{B}_{\mathcal{M}}$, we will have $\dim\mathcal{B}_{\mathcal{M}}=\dim\mathcal{I}_{\mathcal{M}}-2$, so it is indeed a quasi-boundary that tells us how to glue $\mathcal{I}_{\mathcal{M}}$ together to form the compact manifold $\mathcal{M}$ without boundary. Please consider figure~\ref{fig:fermionic-fiber-bundle} for a general illustration and section~\ref{sec:case-two-fermionic-modes} for a concrete example where we work out explicitly how the boundary $\mathcal{B}_{\mathfrak{u}_\perp(N)}$ lives inside the fiber over a single base point. For $N\geq 4$, there can be two or more eigenvalue quadruples $(\ii\frac{\pi}{2},\ii\frac{\pi}{2},-\ii\frac{\pi}{2},-\ii\frac{\pi}{2})$ for $K_+\in\dim\mathcal{B}_{\mathfrak{u}_\perp(N)}$, which leads to more directions being tangential to the fiber, which means such points project to lower dimensional edges of $\mathcal{B}_{\mathcal{M}}$ on the base manifold.

\section{Towards constructing the double cover}\label{sec:DoubleCover}
In the previous section, we learned that the Gaussian unitaries form a representation of the group $\widetilde{\mathcal{G}}$, which is the double cover of the group $\mathcal{G}$, \ie the symplectic group for bosons and the special orthogonal group for fermions. The goal of this section is to construct an explicit parametrization of this double cover $\widetilde{\mathcal{G}}$, \ie the metaplectic group $\mathrm{Mp}(2N,\mathbb{R})=\widetilde{\mathrm{Sp}}(2N,\mathbb{R})$ for bosons and the spin group $\mathrm{Spin}(2N,\mathbb{R})=\widetilde{\mathrm{SO}}(2N,\mathbb{R})$ for fermions.\footnote{While the spin group is not only the double cover, but also the universal cover of $\mathrm{SO}(2N,\mathbb{R})$ for $N>1$, the metaplectic group is not the universal cover of $\mathrm{Sp}(2N,\mathbb{R})$.}
Our construction is based on~\cite{robinson1989metaplectic,rawnsley2012universal}. The idea is that we construct the double cover as the subset
\begin{align}
    \widetilde{\mathcal{G}}=\left\{(M,\psi)\,\big|\,\varphi(M)=\psi^2\right\}\subset \mathcal{G}\times \mathrm{U}(1)\,,\label{eq:double-cover}
\end{align}
albeit with a non-trivial group multiplication that mixes the two factors. The constraint $\varphi(M)=\psi^2$ involves a so-called circle function $\varphi: \mathcal{G}\to \mathrm{U}(1)$ and ensures that for every group element $M$, there will be two solutions $\psi$ and $-\psi$. Here, we represent $\mathrm{U}(1)$ as the unit circle $\mathrm{U}(1)=\{\psi\in\mathbb{C}\,|\,\psi^*\psi=1\}$ embedded in the complex plane. Since $\widetilde{\mathcal{G}}$ is a double cover, it locally looks like $\mathcal{G}\times\mathbb{Z}_2$ as manifold and will have exactly two group elements $(M,\psi)$ and $(M,-\psi)$ that are associated with each element $M\in\mathcal{G}$ under the identification $\mathcal{G}=\widetilde{\mathcal{G}}/\mathbb{Z}_2$.

While we will fully succeed in reaching our goal for bosons (just as Rawnsley~\cite{rawnsley2012universal} did before us), the case of fermions will turn out to be more difficult, as our circle functions $\varphi: \mathrm{SO}(2N,\mathbb{R})\to\mathrm{U}(1)$ will only be defined \emph{almost} everywhere, \ie there will be a set of measure zero, for which $\varphi$ is undefined. Luckily, those points will be exactly the group elements $M$, for which $\braket{J|\mathcal{U}(M)|J}=0$, so that for the purpose of computing $\braket{J|\mathcal{U}(M)|J}=0$ it is irrelevant whether we multiply the expectation value by a complex phase. However, we will still discuss potential solutions to make the parametrization work at these degenerate points.

\subsection{Complex linear group (real representation)}\label{sec:complex-linear-group}
The whole construction will be based on a choice of a fixed reference linear complex structure $J$ that is compatible with $\Omega$ for bosons and with $G$ for fermions, as discussed in section~\ref{sec:complex-Gaussian}. Any such choice defines the subgroup $\mathrm{GL}(N,\mathbb{C})\subset\mathrm{GL}(2N,\mathbb{R})$ introduced in~\eqref{eq:GLNC}. There, we defined elements of $\mathrm{GL}(N,\mathbb{C})$ as real $2N$-by-$2N$ matrices $M$ that commute with $J$, which itself can be interpreted as the real matrix form of the imaginary unit $\ii$. We will later show how to map these matrices to complex $N$-by-$N$ matrices. In the present section, we show how any group or algebra element can be uniquely decomposed into a complex-linear element and a complex-antilinear element, which allows us to prove four important lemmas for later use. From this perspective, the subgroup $\mathrm{GL}(N,\mathbb{C})\subset\mathrm{GL}(2N,\mathbb{R})$ constitutes the set of invertible \emph{complex linear} (as opposed to antilinear) transformations inside the larger group of \emph{real} linear transformations.

Given a $2N$-dimensional vector space $V$ equipped with a complex structure $J$, we can decompose every linear map $M: V\to V$ into its linear and anti-linear part $C_M$ and $D_M$ with respect to $J$, \ie into parts that commute or anti-commute with $J$. In terms of explicit projections we have\footnote{We recall that $J^{-1}=-J$ due to $J^2=-\id$.}
\begin{align}
    C_M=\frac{1}{2}(M-JMJ)\,,\,\,D_K=\frac{1}{2}(M+JMJ)\,,\label{eq:CMDM}
\end{align}
such that indeed $M=C_M+D_M$ and we have
\begin{align}
    [C_M,J]=0\quad\text{and}\quad \{D_M,J\}=0\,.
\end{align}
We will apply this decomposition specifically to group elements $M\in\mathcal{G}$ represented as linear maps $M: V\to V$, but it is important to note that $C_M$ and $D_M$ themselves will in general not represent group elements. The same decomposition can also be applied to Lie algebra elements $K$ and coincides with the Lie algebra Cartan decomposition introduced in~\eqref{eq:Cartan-on_frakg}.

\begin{lemma}\label{lem:invertible-CM}
Let $M$ be invertible. The linear map $C_M$ is always invertible in bosonic systems. In fermionic systems, it is invertible if and only if $\Delta_{M^{-1}}=-M^{-1}JMJ$ has eigenvalues that are all distinct from $-1$.
\end{lemma}
\begin{proof}
We can multiply $C_M$ by $M^{-1}$ from the left to find
\begin{align}
    M^{-1}C_M=\frac{1}{2}(\id-M^{-1}JMJ)=\frac{1}{2}(\id+\Delta_{M^{-1}})\,.
\end{align}
The spectrum of $\Delta_M$ was carefully studied in~\cite{hackl2021bosonic}. For bosons, the spectrum consists of pairs of positive real numbers $(e^{\lambda}, e^{-\lambda})$, so that $C_M$ is always invertible. For fermions, the spectrum consists of quadruples $(e^{\ii \lambda},e^{\ii \lambda},e^{-\ii \lambda},e^{-\ii \lambda})$ and potentially pairs $(1,1)$.\footnote{For $M\in\mathrm{O}(2N,\mathbb{R})$ with $\det(M)=-1$, there will be an odd number of eigenvalue pairs $(-1,1)$, as discussed in~\cite{hackl2021bosonic}. As we restrict to $M\in\mathrm{SO}(2N,\mathbb{R})$, this does not apply to our case.} Only if a quadruple $(-1,-1,-1,-1)$ appears, $C_M$ cannot be inverted.
\end{proof}

For later calculations, we will also introduce
\begin{align}
    Z_M=C_M^{-1}D_M\,,
\end{align}
which is well-defined whenever $C_M$ is invertible.

\begin{lemma}\label{LM:ZProperties}
$Z_M$ satisfies the following properties:
\begin{align}
    Z_M&=\tfrac{\id-\Delta_{M^{-1}}}{\id+\Delta_{M^{-1}}} \quad\text{with}\quad \Delta_M=-MJM^{-1}J\,,\\
    Z_MJ&=-JZ_M\,,\\
    M&=C_M(\id+Z_M)\,,\label{eq:ZProperty3}\\
    \id&=C_{M^{-1}}C_M(1-Z_M^2)\,,\\
    Z_{M^{-1}}C_M&=-C_MZ_M\label{eq:num5}\,.
\end{align}
\end{lemma}
\begin{proof}
For the first equation, we compute
\begin{align}
\begin{split}
    Z_M&=C_M^{-1}D_M=(M-JMJ)^{-1}(M+JMJ)\\
    &=(M-JMJ)^{-1}(MM^{-1})(M+JMJ)\\
    &=(\id-M^{-1}JMJ)^{-1}(\id+M^{-1}JMJ)\\
    &=(\id+\Delta_{M^{-1}})^{-1}(\id-\Delta_{M^{-1}})\\
    &=\frac{\id-\Delta_{M^{-1}}}{\id+\Delta_{M^{-1}}}\,.
\end{split}
\end{align}
For the second equation, we use that $[J,C_M]=0$ implies $[J,C^{-1}_M]=0$, which together with $\{D_M,J\}=0$ implies the relation. For the third equation, we calculate
\begin{align}
    M=C_M+D_M=C_M(\id+C_M^{-1}D_M)=C_M(\id+Z_M)\,.
\end{align}
For the fourth and fifth equation, we follow~\cite{rawnsley2012universal} and write
\begin{align}
    \id=M^{-1}M=C_{M^{-1}}(1+Z_{M^{-1}})C_M(1+Z_M)\,.
\end{align}
We can now multiply out this equation and split into the complex linear and complex anti-linear parts to find
\begin{align}
    \id&=C_{M^{-1}}(C_M+Z_{M^{-1}}C_MZ_M)\,,\\
    0&=C_{M^{-1}}(Z_{M^{-1}}C_M+C_MZ_M)\,,
\end{align}
from which the fourth and fifth equation follow.
\end{proof}

Given a Kähler triple $(G,\Omega,J)$ with inverses $g=G^{-1}$ and $\omega=\Omega^{-1}$, we can turn the classical phase space $V\simeq \mathbb{R}^{2N}$ into a complex Hilbert space $V\simeq \mathbb{C}^{N}$ with complex multiplication $\cdot: \mathbb{C}\times V\to V; (z,v)\mapsto \mathrm{Re}(z)v+\mathrm{Im}(z)Jv$ and inner product
\begin{align}
    \braket{v,w}\propto g(v,w)+\ii \omega(v,w)\,,
\end{align}
where we can choose an arbitrary positive real normalization factor. We will now prove various relations between $D_M$ and $C_M$ with respect to this inner product, which will be relevant to proposition~\ref{prop:eta-continuity}.

\begin{lemma}\label{LM:Adjoint}
For the complex inner product $\braket{v,w}=g(v,w)+\ii \omega(v,w)$, we have for $M\in\mathcal{G}$
\begin{align}
    \braket{C_Mv,w}&=\braket{v,C_{M^{-1}}w}\,,\label{eq:CM-rel}\\[2mm]
    \braket{D_Mv,w}&=\begin{cases}
    -\braket{D_{M^{-1}}w,v} & \normalfont{\textbf{(bosons)}}\\[1mm]
    +\braket{D_{M^{-1}}w,v} & \normalfont{\textbf{(fermions)}}
    \end{cases}\,,\label{eq:DM-rel}\\
    \braket{Z_Mv,w}&=\begin{cases}
    +\braket{Z_Mw,v} & \normalfont{\textbf{(bosons)}}\\
    -\braket{Z_Mw,v} & \normalfont{\textbf{(fermions)}}
    \end{cases}\,.\label{eq:ZM-symmetry}
\end{align}
where the inverse metric $g$ and symplectic form $\omega$ were introduced in~\eqref{eq:gomega}.
\end{lemma}
\begin{proof}
For~\eqref{eq:CM-rel} and~\eqref{eq:DM-rel}, we treat bosons and fermions separately. Note that we use $J\in\mathcal{G}$ due to~\eqref{eq:Compatibility}.\\
\textbf{Properties of $C_M$ and $D_M$ for bosons.} We have $\omega(Mv,w)=\omega(v,M^{-1}w)$ for symplectic $M$ which, in combination with $J^2=-\id$, yields
\begin{align}
  \label{eq:omegaequation}
  \hspace{-1mm}\omega((M\pm JMJ)v,Jw)
  &=\omega(v,(M^{-1}\pm JM^{-1}J)Jw)\\
  &=\mp\omega(v,J(M^{-1}\pm JM^{-1}J)w).\nonumber
\end{align}
  Using $g(v,w)=\omega(v,Jw)$ and the symmetry of $g$, the last equality translates into
\begin{align}
  g(C_Mv,w)
  &=g(v,C_{M^{-1}}w)\\
  g(D_Mv,w)
  &=-g(v,D_{M^{-1}}w)
   =-g(D_{M^{-1}}w,v).
\end{align}
  If we evaluate the first line of \eqref{eq:omegaequation} and perform the substitution $w\to-Jw$ we similarly obtain
\begin{align}
  \omega(C_Mv,w)
  &=\omega(v,C_{M^{-1}}w)\\
  \omega(D_Mv,w)
  &=\omega(v,D_{M^{-1}}w)
   =-\omega(D_{M^{-1}}w,v).
\end{align}
Combining all these equalities proves~\eqref{eq:CM-rel} and~\eqref{eq:DM-rel} for bosons.\\
\textbf{Properties of $C_M$ and $D_M$ for fermions.} We have $g(Mv,w)=g(v,M^{-1}w)$ for orthogonal $M$ and $J^2=-\id$ yielding
\begin{align}
  \label{eq:gequation}
  g((M\pm JMJ)v,Jw)
  &=g(v,(M^{-1}\pm JM^{-1}J)Jw)\\
  &=\mp g(v,J(M^{-1}\pm JM^{-1}J)w).\nonumber
\end{align}
Using $\omega(v,w)=-g(v,Jw)$ and the anti-symmetry of $\omega$, this implies
\begin{align}
  \omega(C_Mv,w)
  &=\omega(v,C_{M^{-1}}w)\\
  \omega(D_Mv,w)
  &=-\omega(v,D_{M^{-1}}w)
   =\omega(D_{M^{-1}}w,v).
\end{align}
On the other hand, evaluating the first line of \eqref{eq:gequation} for $w\to-Jw$, combined with the symmetry of $g$ leads to
\begin{align}
  g(C_Mv,w)
  &=g(v,C_{M^{-1}}w)\\
  g(D_Mv,w)
  &=g(D_{M^{-1}}w,v).
\end{align}
Combining all these equalities proves~\eqref{eq:CM-rel} and~\eqref{eq:DM-rel} for fermions.\\
\textbf{Properties of $Z_M$.} We now use the previous findings to prove~\eqref{eq:ZM-symmetry} for bosons and fermions simultaneously. Defining $w=C_{M^{-1}}w'$, we compute
\begin{align}
\begin{split}
    \braket{Z_Mv,w}&=\braket{Z_Mv,C_{M^{-1}}w'}=\braket{C_MZ_Mv,w'}\\
    &=\braket{D_Mv,w'}=\mp \braket{D_{M^{-1}}w',v}\\
    &=\mp \braket{C_{M^{-1}}Z_{M^{-1}}w',v}\\
    &=\pm \braket{Z_MC_{M^{-1}}w',v}=\pm \braket{Z_{M}w,v}\,,
\end{split}
\end{align}
where we used~\eqref{eq:num5} to get in the second-last line and the upper sign refers to bosons, while the lower sign refers to fermions.
\end{proof}

\subsection{Complex linear group (complex representation)}
In the previous section, we discussed the group $\mathrm{GL}(N,\mathbb{C})$ and Lie algebra $\mathfrak{gl}(N,\mathbb{C})$ as the set of invertible real $2N$-by-$2N$ matrices that commute with a given complex structure $J$. Therefore, computing traces or determinants in this representation will yield real numbers. The goal of the present section is to establish a structure-preserving bijection between such real $2N$-by-$2N$ matrices and complex $N$-by-$N$ matrices, allowing us to define complex traces and determinants.

Given a $2N$-dimensional real vector space $V$ with a complex structure $J: V\to V$, we consider linear maps $K: V\to V$ that commute with $J$, \ie $[K,J]=0$, which also implies $C_K=K$. The matrix representation of $K$, with respect to a basis where $J$ takes the standard form~\eqref{eq:J-standard-form}, has the block structure
\begin{align}
    K=C_K=\left(\begin{array}{c|c}
    K_1     &  K_2\\
    \hline
    -K_2    & K_1
    \end{array}\right)\,.\label{eq:dec-form}
\end{align}
Noting that we would like to think of $J$ as corresponding to the imaginary unit $\ii$, this allows us to identify the real $2N$-by-$2N$ matrix $K$ with the complex $N$-by-$N$ matrix
\begin{align}
    \label{eq:ComplexDecomposition}
    \overline{K}=K_1+\ii K_2\,.
\end{align}
This identification is an isomorphism between the respective matrix spaces, \ie between real $2N$-by-$2N$ matrices commuting with $J$ and complex $N$-by-$N$ matrices. Indeed, a straightforward computation yields $\overline{\id}=\id$ and $\overline{J}=\ii\id$ as well as
\begin{align}
    \label{eq:ComplexDecompositionProperties}
    \overline{\alpha K}=\alpha \overline{K}\,,\,\overline{K+K'}=\overline{K}+\overline{K'}\,,\, \overline{KK'}=\overline{K}\,\overline{K'}\,,
\end{align}
where $\alpha\in\mathbb{R}$. We note that the correspondence also preserves the invertibility of a matrix and that $\overline{K^{-1}}=\overline{K}^{-1}$ (see also Corollary~\ref{eq:DetSquared}).

The exact form of the complex matrix $\overline{K}$ for a given $K$ is not unique, but rather depends on the specific basis we use, as long as $J$ takes the standard form~\eqref{eq:J-standard-form}. The form of $J$ is preserved when changing basis by applying any invertible map $M: V\to V$, such that $MJM^{-1}=J$ which is equivalent to $[M,J]=0$. Therefore, $\overline{K}$ is only defined up to arbitrary basis changes $\overline{K}\to \overline{M}\,\overline{K}\,\overline{M}^{-1}$ for an invertible complex $N$-by-$N$ matrix $\overline{M}$.

The identifications \eqref{eq:dec-form} and \eqref{eq:ComplexDecomposition} allow us to define the complex determinant and trace
\begin{align}
    \overline{\det}(K)=\det(\overline{K})\quad\text{and}\quad \overline{\mathrm{Tr}}(K)=\mathrm{Tr}(\overline{K})\,,\label{eq:detbar}
\end{align}
which we now show to be basis-independent.

\begin{lemma}
For any linear map $K: V\to V$, the above defined determinant $\overline{\det}(K)$ and trace $\overline{\mathrm{Tr}}(K)$ are uniquely defined.
\end{lemma}
\begin{proof}
The block form of $K$ follows directly from writing out $KJ=JK$ in blocks. Having chosen a basis, where $J$ takes the standard form~\eqref{eq:J-standard-form}, we can move to another such basis by applying an invertible linear map $M: V\to V$ that preserves $J$, \ie $MJM^{-1}=J$. This is equivalent to $[M,J]=0$, so that the same block decomposition applies to $M$, which we can now explicitly identify with a group element $\overline{M}\in\mathrm{GL}(N,\mathbb{C})$. The general properties \eqref{eq:ComplexDecompositionProperties} imply $\overline{MKM^ {-1}}=\overline{M}\,\overline{K}\,\overline{M}^{-1}$, from which it follows that $\overline{\det}$ and $\overline{\mathrm{Tr}}$ are basis-independent.
\end{proof}

We also have the following relation between the eigenvectors of $K$ and those of $\overline{K}$.
\begin{lemma}\label{lem:eigenvaluesMbar}
    Let $K$ be a complex linear map, i.e.\ $[K,J]=0$. If $\overline{K}$ has eigenvalues $\{\lambda_i\}$ with algebraic multiplicities $\mu_i$, then $K$ will have as eigenvalues the pairs $\{\lambda_i,\lambda_i^*\}$, where the eigenvalues $\lambda_i$ and $\lambda_i^*$ each have algebraic multiplicity $\mu_i$. In particular if $\lambda_i\in \mathbb{R}$ then $K$ will have an eigenvalue $\lambda_i$ with even multiplicity $2\mu_i$.
\end{lemma}
\begin{proof}
    Consider a eigenvalue $\lambda$ of $\overline{K}$ with algebraic multiplicity $\mu$. For sufficiently large $n$, $\ker(\overline{K}-\lambda\id)^n$ will have dimension $\mu$ and we denote a basis of this kernel (\ie the generalized eigenvectors) as $v_l$ with $l=1,\dots,\mu$. We now observe that the vectors $(v_l,\ii v_l)$ are in $\ker(K-\lambda\id)^n$ and the vectors $(v_l^*,-\ii v_l^*)$ are in $\ker(K-\lambda^*\id)^n$. To see this consider 
    \begin{align}
        &\left[\left(\begin{array}{c|c}
        K_1     &  K_2\\
        \hline
        -K_2    & K_1
        \end{array}\right) -\lambda \id \right]^n
        \left(\begin{array}{c}
         v_l\\
         \ii v_l
         \end{array}\right) \nonumber\\
        & \hspace{20mm}=\left(\begin{array}{c}
         (\overline{K}-\lambda \id)^n \, v_l\\
         \ii\, (\overline{K}-\lambda \id)^n \, v_l
         \end{array}\right) =0 \,,
    \end{align}
    where we have repeatedly used $K_1 v+\ii K_2 v=\overline{K}v$. Similarly, using $K_1 v-\ii K_2 v=\overline{K}^* v$ we have
    \begin{align}
        &\left[\left(\begin{array}{c|c}
        K_1     &  K_2\\
        \hline
        -K_2    & K_1
        \end{array}\right) -\lambda^* \id \right]^n
        \left(\begin{array}{c}
         v^*\\
         -\ii v^*
         \end{array}\right) \nonumber\\
         &\hspace{20mm}=\left(\begin{array}{c}
         (\overline{K}^*-\lambda^* \id)^n \, v^*\\
         \ii\, (\overline{K}^*-\lambda^* \id)^n \, v^*
         \end{array}\right) \nonumber\\
         &\hspace{20mm}= \left(\begin{array}{c}
         (\overline{K}-\lambda \id)^n \, v\\
         \ii\, (\overline{K}-\lambda \id)^n \, v
         \end{array}\right)^* =0 \,. 
    \end{align}
    It is straightforward to see that the vectors $(v_l,\ii v_l)$ and $(v_l^*,-\ii v_l^*)$ are linearly independent. This means that $K$ has $\mu$ independent generalized eigenvectors for each of the eigenvalues $\lambda$ and $\lambda^*$. On dimensional grounds it is clear that $K$ cannot have further eigenvalues and eigenvectors that are not of this form.
\end{proof}

As an immediate consequence we have the following.
\begin{corollary}\label{lem:determinant}
  \label{eq:DetSquared}
  $\det K=\bigl|\overline{\det}(K)\bigr|^2$ and for $K$ with purely non-negative eigenvalues $\overline{\det}(K)=\sqrt{\det{(K)}}$.
\end{corollary}
\begin{proof}
    From lemma~\ref{lem:eigenvaluesMbar} it follows that $\det K = (\lambda_1 \lambda_1^*)^{\mu_1} (\lambda_2 \lambda_2^*)^{\mu_2} \cdots (\lambda_d \lambda_d^*)^{\mu_d} = |\lambda_1 ^{\mu_1} \lambda_2 ^{\mu_2} \cdots \lambda_d^{\mu_d}|^2=|\det \overline{K}|^2$, where $\lambda_i$ for $i=1,\dots,d$ are the eigenvalues of $\overline{K}$ with multiplicities $\mu_i$. If these eigenvalues are all non-negative, then the previous expression reduces to $\det K = (\lambda_1 ^{\mu_1} \lambda_2 ^{\mu_2} \cdots \lambda_d^{\mu_d})^2=(\overline{\det} K)^2$ where we can take the square root because $\overline{\det} K = \lambda_1 ^{\mu_1} \lambda_2 ^{\mu_2} \cdots \lambda_d^{\mu_d} \geq 0$.
\end{proof}

\begin{lemma}\label{lem:coshKplus}
Given $K_+\in\mathfrak{u}_\perp(N)$ for bosons and $K_+\in\mathcal{I}_{\mathfrak{u}_\perp(N)}$ for fermions, $\cosh{K_+}$ commutes with $J$ and we have $\overline{\det}(\cosh{K_+})=\sqrt{\det(\cosh{K_+})}>0$.
\end{lemma}
\begin{proof}
We recall that $\cosh{K_+}=\sum^\infty_{l=0}\frac{K_+^{2l}}{(2l)!}$ only consists of even powers of $K_+$. As $K_+$ itself anti-commutes with $J$, the even powers ensure that $(\cosh{K_+})J=J(\cosh{K_+})$. This implies that $\overline{\det}(\cosh{K_+})$ is well-defined. For bosons, all eigenvalues of $K_+$ are real, which implies that all eigenvalues $\cosh{K_+}$ are positive. For fermions, we know that the eigenvalues $\pm\ii\lambda$ of $K_+\in \mathcal{I}_{\mathfrak{u}_\perp(N)}$ are imaginary with modulus $\lambda<\frac{\pi}{2}$. The eigenvalues of $\cosh{K_+}$ are then given by $\cos{\lambda}$ which are necessarily also all positive (as $\cos{\lambda}>0$ for $\lambda\in[0,\frac{\pi}{2})$). With all eigenvalues of $\cosh{K_+}$ being positive and real, we have $\overline{\det}(\cosh{K_+})=\sqrt{\det(\cosh{K_+})}>0$ due to Corollary~\ref{lem:determinant}.
\end{proof}

We note that $C_M$ by construction commutes with $J$, so that $\overline{\mathrm{det}}(C_M)$ and $\overline{\mathrm{tr}}(C_M)$ are always well-defined, regardless of the choice of $M$.

\subsection{Circle and cocycle function}\label{sec:circle-and-cocycle}
The goal of this subsection is to introduce a so-called circle function $\varphi: \mathcal{G}\to\mathrm{U}(1)$, which is a mapping from $\mathcal{G}$ to $\mathrm{U}(1)$. In the case of fermions, we will see that we can define $\varphi$ only \emph{almost} everywhere on $\mathcal{G}$. In a second step, we will then introduce the so-called cocycle function $\eta: \mathcal{G}\times\mathcal{G}\to\mathbb{R}$, which encodes the difference between $\varphi(M_1M_2)$ and $\varphi(M_1)\varphi(M_2)$ and plays a crucial role for modelling the double cover of $\mathcal{G}$.

For a reference complex structure $J$ that is compatible with $\Omega$ or $G$, respectively, we define $\varphi: \mathcal{G}\to\mathrm{U}(1)$ as
\begin{align}\label{eq:circle-function}
    \varphi(M)=\frac{\overline{\det}(C_M)}{|\overline{\det}(C_M)|}\,,
\end{align}
where $\overline{\mathrm{det}}(C_M)$ was introduced in~\eqref{eq:detbar}. This choice is motivated by the idea that gluing two copies of $\mathcal{G}$ to get its double cover $\widetilde{\mathcal{G}}$ requires some understanding of the winding number, as encoded by the fundamental group $\pi_1(\mathcal{G})$. For bosons, the circle function was introduced in~\cite{rawnsley2012universal}, because it induces an isomorphism of the fundamental groups, such that a closed curve in $\mathrm{Sp}(2N,\mathbb{R})$ is mapped to a closed curve in $\mathrm{U}(1)$ with the same winding number. For fermions,~\eqref{eq:circle-function} is not quite a circle function in the sense of~\cite{rawnsley2012universal}, because it does not induce an isomorphism between fundamental groups of $\mathrm{SO}(2N,\mathbb{R})$ and $\mathrm{U}(1)$ for $N>1$, as the two groups are different, \ie $\pi_1(\mathrm{SO}(2N,\mathbb{R}))=\mathbb{Z}_2$ and $\pi_1(\mathrm{U}(1))=\mathbb{Z}$. Moreover, lemma~\ref{lem:invertible-CM} implies that there exist group elements $M$ for fermions, for which $C_M$ is not invertible, such that the right hand side of~\eqref{eq:circle-function} is ill-defined due to $\overline{\det}(C_M)=0$. Despite these limitations, the fermionic circle function can still be used to determine if a loop $\gamma$ in $\mathrm{SO}(2N,\mathbb{R})$ can be contracted or not\footnote{We assume here that $\varphi$ is defined for all points on the loop $\gamma$.}, depending on whether the winding number of its image $\varphi(\gamma)$ has even or odd winding number in $\mathrm{U}(1)$. In summary, the circle function is designed to probe non-contractable directions in $\mathcal{G}$, which will later allow us to glue two copies of $\mathcal{G}$ together to get its double cover $\widetilde{\mathcal{G}}$.

\begin{lemma}
The circle function satisfies the following properties:
\begin{itemize}
    \item[(i)]  Given a complex $J$ and a group element $M\in\mathcal{G}$, such that the respective Cartan decomposition is given by $M=Tu$, we have
\begin{align}
    \varphi(M)=\overline{\det}(u)\,,
\end{align}
including $\varphi(\id)=1$. For fermions, $\varphi$ is only defined for $M\in\mathcal{G}$ with $\det(\id+\Delta_M)\neq 0$.
    \item[(ii)] For $u_1,u_2\in\mathrm{U}(N)$, we have
\begin{align}
    \varphi(u_1Mu_2)=\varphi(u_1)\varphi(M)\varphi(u_2)\,.\label{eq:phi-equivariance}
\end{align}
    \item[(iii)] The inverse $M^{-1}$ satisfies
    \begin{align}
    \varphi(M^{-1})=\varphi(M)^*\,.\label{eq:phi-inverse}
\end{align}
\end{itemize}
\end{lemma}
\begin{proof}
We prove each statement:\\
(i) We plug $M=Tu$ into $C_M=\frac{M-JMJ}{2}$ and find
\begin{align}
    C_M&=\frac{Tu-JTuJ}{2}=\frac{T+T^{-1}}{2}u\,,
\end{align}
where we used $T^{-1}J=JT$. This implies
\begin{align}
    \overline{\det}(C_M)=\overline{\det}(\tfrac{T+T^{-1}}{2})\overline{\det}(u)\,.
\end{align}
Using $T=e^{K_+}$ allows us to evaluate $\overline{\det}(\frac{T+T^{-1}}{2})=\sqrt{\det(\cosh{K_+})}\geq 0$ using lemma~\eqref{lem:coshKplus}, which implies $\varphi(M)=\overline{\det}(u)$ whenever $K_+$ does not have any eigenvalues $\pm\ii \frac{\pi}{2}$, which is equivalent to requiring $\det(\id+\Delta_M)\neq 0$, as noticed in the context of~\eqref{eq:cartan-fermion-defined}.\\
(ii) Using the above decomposition also yields
\begin{align}
    \overline{\det}{C_{u_1Mu_2}}=\overline{\det}(u_1)\overline{\det}(\tfrac{T+T^{-1}}{2})\overline{\det}(u)\overline{\det}(u_2)\,,
\end{align}
from which $\varphi(u_1Mu_2)=\varphi(u_1)\varphi(M)\varphi(u_2)$ follows.\\
(iii) Using the arguments from (i) with $M^{-1}=u^{-1}T^{-1}$, we can compute
\begin{align}
    \overline{\det}(C_{M^{-1}})=\overline{det}(u^{-1})\overline{\det}(\tfrac{T+T^{-1}}{2})\,,
\end{align}
which implies $\varphi(M^{-1})=\overline{\det}(u^{-1})=\varphi(M)^*$, whenever the Cartan decomposition is unique.
\end{proof}

Whenever the relevant circle functions exist it can be shown that there is a unique continuous map $e^{\ii\eta(M_1,M_2)}$, such that
\begin{align}
    \varphi(M_1M_2)=\varphi(M_1)\varphi(M_2)e^{\ii\eta(M_1,M_2)}\,.\label{eq:cocycle-condition}
\end{align}
As is explained in~\cite{rawnsley2012universal}, actually even the map $\eta(M_1,M_2)$ itself is continuous for bosons (when $\mathcal{G}$ is the symplectic group). In contrast, when $N>1$ this is not true for the group $\mathrm{SO}(2N,\mathbb{R})$ relevant to fermions, where $\eta(M_1,M_2)$ may have jumps of $4\pi$.\footnote{The jumps of $4\pi$ are a symptom of the fact that the fundamental group $\mathrm{SO}(2N,\mathbb{R})$ with $N>1$ is $\mathbb{Z}_2$ rather than $\mathbb{Z}$. If there were no jumps and we could define a continuous function $\eta(M_1,M_2)$ without any jumps, we could use it to construct triple and higher covers of $\mathrm{SO}(2N,\mathbb{R})$, which cannot exist as the double cover is already universal. This is in contrast to $\mathrm{SO}(2)$ and $\mathrm{Sp}(2N,\mathbb{R})$, where the universal cover wraps around an (countably) infinite number of times (the universal cover of $\mathrm{U}(1)$ is the real line).} However, this particular nature of the jumps is sufficient to ensure that both $e^{\ii\eta(M_1,M_2)}$ and also $e^{\ii\eta(M_1,M_2)/2}$ are continuous on all of $\mathcal{G}\times\mathcal{G}$.

\begin{lemma}\label{lem:cocycle-properties}
The cocycle function satisfies:
\begin{align}
    \eta(M_1,M_2)&=\eta(u_1M_1,M_2u_2)\text{ for }u_1,u_2\in\mathrm{U}(N)\,,\label{eq:eta1}\\
    \eta(u,M)&=\eta(M,u)=0\text{ for }u\in\mathrm{U}(N)\,,\label{eq:eta2}\\
    \eta(M,M^{-1})&=0\,,\label{eq:eta3}\\
    \eta(M_1,M_2)&=-\eta(M_2^{-1},M_1^{-1})\,.\label{eq:eta4}
\end{align}
\end{lemma}
\begin{proof}
We can prove all these properties by considering~\eqref{eq:cocycle-condition} as a continuous function of $M_1,M_2\in\mathcal{G}$. \eqref{eq:eta1} follows from~\eqref{eq:phi-equivariance} applied to both sides of~\eqref{eq:cocycle-condition}. \eqref{eq:eta2} follows from first computing $\eta(\id,M)=0$ from~\eqref{eq:cocycle-condition} using $\varphi(\id)=1$ and then applying~\eqref{eq:eta1}. \eqref{eq:eta3} and~\eqref{eq:eta4} follow both from applying~\eqref{eq:phi-inverse} to~\eqref{eq:cocycle-condition}.
\end{proof}

\subsection{Computing the cocycle function}
The goal of this section is to develop a simple method to evaluate $\eta(M_1,M_2)$ introduced in~\eqref{eq:cocycle-condition} and eventually arrive at~\eqref{eq:etaDefinitionLog}, which expresses $\eta(M_1,M_2)$ as a simple trace. We can express~\eqref{eq:cocycle-condition} as
\begin{align}
    \hspace{-2mm}e^{\ii\eta(M_1,M_2)}&=\frac{\overline{\det}(C_{M_1M_2})}{\overline{\det}(C_{M_1})\overline{\det}(C_{M_2})}\left|\frac{\overline{\det}(C_{M_1})\overline{\det}(C_{M_2})}{\overline{\det}(C_{M_1M_2})}\right|\nonumber\\
    &=\frac{\overline{\det}(C_{M_1}^{-1}C_{M_1M_2}C_{M_2}^{-1})}{|\overline{\det}(C_{M_1}^{-1}C_{M_1M_2}C_{M_2}^{-1})|}\,,
    \label{eq:eta-def2}
\end{align}
where we used standard properties of the determinant.

Let us first simplify the argument of the determinant.
\begin{proposition}
We have
\begin{align}
    C_{M_1}^{-1}C_{M_1M_2}C_{M_2}^{-1}=\id-Z_{M_1}Z_{M_2^{-1}}\,.\label{eq:reduction}
\end{align}
\end{proposition}
\begin{proof}
We follow the strategy of~\cite{rawnsley2012universal}. From~\eqref{eq:ZProperty3} we have $M=C_M(\id+Z_M)$. Using this, we write the product $M_1M_2$ in the two different ways
\begin{align}
    M_1M_2&=C_{M_1M_2}(\id+Z_{M_1M_2})\,,\\
    M_1M_2&=C_{M_1}(\id+Z_{M_1})C_{M_2}(\id+Z_{M_2})\,.
\end{align}
Equating them and taking the complex linear piece yields
\begin{align}
    C_{M_1M_2}&=C_{M_1}(C_{M_2}+Z_{M_1}C_{M_2}Z_{M_2})\\
    &=C_{M_1}(\id-Z_{M_1}Z_{M_2^{-1}})C_{M_2}\,,
\end{align}
where we used that $C_M$ and $\id$ are complex linear, while $Z_M$ is complex antilinear, so $Z_M$ needs to appear in even powers. We used~\eqref{eq:num5} to reach the last equation, which is equivalent to~\eqref{eq:reduction}, which we wanted to prove.
\end{proof}

Formula \eqref{eq:eta-def2} in conjunction with the previous proposition states that the function $e^{\ii\eta(M_1,M_2)}$ can be expressed in terms of the (potentially generalized) eigenvalues $\lambda_i$ of the {\em complex} matrix $\overline{\id-Z_{M_1}Z_{M_2^{-1}}}$, at least as long as none of them vanishes. Under this assumption, one choice which clearly is consistent with \eqref{eq:eta-def2} is to define
\begin{align}
  \label{eq:eta-explicit}
  \eta(M_1,M_2)=\sum_{i=1}^N \arg(\lambda_i)\;,
\end{align}
where we assume $\arg(\lambda)\in(-\pi,\pi]$, \ie a cut along the negative real axis. We would like to understand how the right hand side varies as a function of $M_1$ and $M_2$. While possibly having jumps by integer multiples of $2\pi$ if one or more eigenvalues cross the negative real axis, the choice~\eqref{eq:eta-explicit} is sufficient to render the function $e^{\ii\eta(M_1,M_2)}$ continuous. The following proposition establishes that we even have a stronger result, as eigenvalues can only cross the negative real axis in pairs, leading to jumps of $\eta(M_1,M_2)$ by integer multiples of $4\pi$. 

\begin{proposition}\label{prop:eta-continuity}
  Formula~\eqref{eq:eta-explicit} gives rise to a continuous function $e^{\ii \eta(M_1,M_2)/2}$ of $M_1$ and $M_2$, whenever the right hand side exists. It exists everywhere for bosons, but for fermions it only exists for those $(M_1,M_2)$, for which $C_{M_1}$, $C_{M_2^{-1}}$ and $\id-Z_{M_1}Z_{M_2^{-1}}$ are all invertible.
\end{proposition}
\begin{proof}
We prove this for bosons and fermions separately:\\[2mm]
\textbf{Bosons.} We show that the eigenvalues of $\id-Z_{M_1}Z_{M_2^{-1}}$ have a positive real part~\cite{rawnsley2012universal}, which follows from
\begin{align}
\begin{split}
    &\hspace{-.4cm}\braket{[(\id-Z_1Z_2)+(\id-Z_1Z_2)^*]v,v}\\
    &=\braket{(\id-Z_1Z_2)v,v}+\braket{v,(\id-Z_1Z_2)v}\\
    &=\braket{(\id-Z_1^2)v,v}+\braket{(\id-Z_2^2)v,v}+\braket{Z_1^2v,v}\\
    &\quad+\braket{Z_2^2v,v}-\braket{Z_1Z_2v,v}-\braket{v,Z_1Z_2v}\\
    &=\braket{(\id-Z_1^2)v,v}+\braket{(\id-Z_2^2)v,v}+\braket{Z_1v,Z_1v}\\
    &\quad+\braket{Z_2v,Z_2v}-\braket{Z_1v,Z_2v}-\braket{Z_2v,Z_1v}\\
    &=\braket{(\id-Z_1^2)v,v}+\braket{(\id-Z_2^2)v,v}+\lVert(Z_1-Z_2)v\rVert^2\,,\hspace{-7mm}
\end{split}
\end{align}
where we used $\braket{Zv,w}=\braket{Zw,v}$ for bosons from~\eqref{eq:ZM-symmetry}. Note that $C_M^*$ refers to the adjoint with respect to the inner product $\braket{\cdot,\cdot}$. The first two terms of the last equation are positive, as due to Lemmas \ref{LM:ZProperties} and \ref{LM:Adjoint} the matrix $\id-Z^2=(C_{M^{-1}}C_M)^{-1}=(C^*_MC_M)^{-1}$ is clearly positive-definite ($C_{M^{-1}}=C_M^*$ follows from~\eqref{eq:CM-rel}) and the last term is manifestly non-negative. If we now consider a vector $v$ with $(\id-Z_1Z_2)v=\lambda v$, above equation implies $\mathrm{Re}(\lambda_i)>0$ for all eigenvalues $\lambda_i$ of $\id-Z_1Z_2$. By lemma~\ref{lem:eigenvaluesMbar} also all the eigenvalues of $\overline{\id-Z_{M_1}Z_{M_2^{-1}}}$ must then stay in the right half of the complex plane, so $\eta(M_1,M_2)$ and thus also $e^{\ii c\eta(M_1,M_2)}$ will be continuous functions of $M_1,M_2\in\mathcal{G}$ regardless of the choice $c\in\mathbb{R}$, including $c=\frac{1}{2}$.\\
\textbf{Fermions.} We need to calculate the eigenvalues of
\begin{align}
    \overline{\id-Z_{M_1}Z_{M_2^{-1}}}\,,
\end{align}
which is well-defined if the spectra of $\Delta_{M_1^{-1}}$ and $\Delta_{M_2}$ do not contain $-1$. Under this assumption, we recall from lemma~\ref{LM:Adjoint} that $Z_M$ is antisymmetric with respect to the inner product $g$. Thus, there exists a basis where both $Z_{M_1}$ and $Z_{M_2^{-1}}$ are represented by anti-symmetric matrices, so they are complex diagonalizable. Based on the so-called Stenzel condition~\cite{ikramov2009product}, each non-zero eigenvalues of the product $Z_{M_1}Z_{M_2^{-1}}$ appear then an even number of times. Of course, due to the vector space being even dimensional, also zero eigenvalues must then appear an even number of times. Moreover, as the matrices are real-valued, we also know that every non-real eigenvalue appears in conjugate pairs, which implies that non-real eigenvalues must appear in quadruples $(\lambda,\lambda,\lambda^*,\lambda^*)$.\\
As a consequence of lemma~\ref{lem:eigenvaluesMbar} and on dimensional grounds, we have the following statement: If a real matrix $M$ with $[M,J]=0$ has eigenvalues $(\lambda,\lambda^*)$, the complex matrix $\overline{M}$ needs to have one of these eigenvalues as well, say $\lambda$. Applied to our situation, this implies that $\overline{\id-Z_{M_1}Z_{M_2^{-1}}}$ has either real eigenvalues or complex eigenvalues come as pairs, namely either as $(\lambda,\lambda^*)$ or as  
The $\mathrm{arg}$ function appearing in the definition \eqref{eq:eta-explicit} of $\eta(M_1,M_2)$ is well-defined on the complex plane with a branch cut along $(-\infty,0]$. Apart from points $(M_1,M_2)$, where one or more eigenvalues vanish, $\eta(M_1,M_2)$ only has discontinuities associated to a jump by a multiple of $4\pi$. The reason is that complex eigenvalues either appear in conjugate pairs $(\lambda,\lambda^*)$, in which case there can be no jump (as the eigenvalues cross the branch cut in opposite directions), or in equal pairs $(\lambda,\lambda)$, in which case $\eta(M_1,M_2)$ will jump by $4\pi$ (as both $\mathrm{arg}(\lambda)$ jump from $\pi$ to $-\pi$ or vice versa). Therefore, $e^{\ii\eta(M_1,M_2)/2}$ is continuous outside of the critical points $(M_1,M_2)$ where $\id-Z_{M_1}Z_{M_2^{-1}}$ has vanishing eigenvalues or where $Z_{M_1}$ or $Z_{M_2^{-1}}$ are ill-defined.
\end{proof}

We thus have $|\eta(M_1,M_2)|<\frac{N\pi}{2}$ for bosons and $|\eta(M_1,M_2)|<N\pi$ for fermions. It is a continuous function for bosons, while there may be jumps of $4\pi$ for fermions. Consequently, $e^{\ii\eta(M_1,M_2)/2}$ is a well-defined and continuous function on all of $\mathcal{G}\times\mathcal{G}$, which we will need in the next step to define the multiplication on $\widetilde{\mathcal{G}}$.

In summary, we can compute $\eta(M_1,M_2)$ in three steps:
\begin{enumerate}
    \item Evaluate the real $2N$-by-$2N$ matrix $\id-Z_{M_1}Z_{M_2^{-1}}$ that commutes with $J$.
    \item Convert it to a complex $N$-by-$N$ matrix using~\eqref{eq:dec-form}.
    \item Compute its $N$ complex eigenvalues and add their complex phases. If the eigenvalue is a negative real number, we take $\pi$ as its complex phase (fixing the logarithm at the branch cut) for definiteness.
\end{enumerate}
We can write this equivalently as
\begin{align}
  \label{eq:etaDefinitionLog}
    \begin{split}
    \hspace{-2mm}\eta(M_1,M_2)&=\mathrm{Im}\,\overline{\mathrm{Tr}}\log(\id-Z_{M_1}Z_{M_2^{-1}})\,.\\
    &= \mathrm{Im}\,\overline{\mathrm{Tr}}\log\left(\id-\frac{\id-\Delta_{M_1^{-1}}}{\id+\Delta_{M_1^{-1}}}\frac{\id-\Delta_{M_2}}{\id+\Delta_{M_2}}\right)
    \end{split}
\end{align}
Our previous analysis ensures that $e^{\ii \eta(M_1,M_2)/2}$ is continuous, whenever it is defined.

\subsection{Constructing the double cover}\label{sec:construction-double-cover}
For bosons, the normalized circle function $\varphi: \mathcal{G}\to \mathrm{U}(1)$ is a homomorphism of the fundamental groups $\pi_1(\mathcal{G})$ and $\pi_1(\mathrm{U}(1))$, \ie a non-contractible loop with a given winding number in $\mathcal{G}$ is mapped to a loop in $\mathrm{U}(1)$ of the same winding number. In order to construct the double cover, we want to assign to each group element $M\in\mathcal{G}$ an additional two-fold choice, \ie $M\to (M,\pm \psi_M)$, such that a loop with winding number $1$ in $\mathcal{G}$ connects exactly the two elements $(M,\psi_M)$ and $(M,-\psi_M)$, as forshadowed in~\eqref{eq:double-cover}. A natural choice is the requirement $\psi_M^2=\varphi(M)$, which uses the natural double covering map $z\mapsto z^2$ of the complex unit circle onto itself.\footnote{Note that the naive choice $\psi_M=\pm 1$ yielding the set $\mathcal{G}\times \mathbb{Z}_2$ would not accurately describe the topology of $\widetilde{\mathcal{G}}$.} For fermions, the picture is similar, though with two caveats: First, the fundamental group of $\mathrm{SO}(2N,\mathbb{R})$ for $N>1$ is $\mathbb{Z}_2$, so the winding number of a loop in $\mathcal{G}$ can only be $0$ or $1$, which means after mapping the loop to $\mathrm{U}(1)$ we only need to distinguish even or odd winding numbers. Second, for a chosen $J$ the fermionic circle function $\varphi(M)$ is ill-defined for some $M$, so to compute the winding number of a loop may require a deformation of the loop to avoid such critical points or a change of $J$.

We can use the normalized circle function $\varphi$ from~\eqref{eq:circle-function} to define the double cover $\widetilde{\mathcal{G}}$ of the group $\mathcal{G}$ as the set
\begin{align}
    \widetilde{\mathcal{G}}=\left\{(M,\psi)\in\mathcal{G}\times\mathrm{U}(1)\,\big|\,\psi^2=\varphi(M)\right\}\,, \label{eq:double-cover-def}
\end{align}
which means that for every $M\in\mathcal{G}$, the double cover contains the two group elements $(M,\pm\sqrt{\varphi(M)})$, where we define the square root $\sqrt{e^{\ii\vartheta}}=e^{\ii\vartheta/2}$ for $\vartheta\in(-\pi,\pi]$.

On this set, we define the group multiplication
\begin{align}
   \hspace{-2mm} (M_1,\psi_1)\cdot(M_2,\psi_2)=(M_1\cdot M_2,\psi_1\psi_2e^{\ii\eta(M_1,M_2)/2}),
    \label{eq:group-multiplication-double-cover}
\end{align}
where $\eta(M_1,M_2)$ is the cocycle associated to $\varphi$ and defined by equation~\eqref{eq:cocycle-condition} and the requirement to be continuous up to jumps of $4\pi$. The conditions guarantee that the multiplication \eqref{eq:group-multiplication-double-cover} is associative and continuous. Explicit formulas for $\eta(M_1,M_2)$ are available in~\eqref{eq:eta-explicit} and~\eqref{eq:etaDefinitionLog}.

The construction of the double cover with~\eqref{eq:group-multiplication-double-cover} was done in~\cite{rawnsley2012universal} for the symplectic group, but we argue that an essentially identical construction can also be used for the double cover of the special orthogonal group. However, there is an important caveat: The circle function from~\eqref{eq:circle-function} is only defined where $\overline{\det}(C_M)\neq 0$. As discussed in section~\ref{sec:circle-and-cocycle}, this will always the case for bosons, but for fermions there exist a submanifold of group elements $M$ with $\overline{\det}(C_M)=0$. Luckily, this is a set of measure zero, so the construction will still work for \emph{almost all} group elements, but just not everywhere. Moreover, we argue here that these ``holes'' will not affect the winding numbers, \ie given a loop in $\mathcal{G}$ we can always continuously deform this loop and/or change our reference complex structure $J$ used in the definition of the circle function, such that $\varphi$ is defined everywhere on the loop and maps it to $\mathrm{U}(1)$. We can then compute the winding number of this loop in $\mathrm{U}(1)$ modulo $2$, which agrees with the homotopy of the original loop in $\mathcal{G}$.

We can define the surjective homomorphism
\begin{align}
    \sigma: \widetilde{\mathcal{G}}\to\mathcal{G}
    \quad\text{with}\quad
    (M,\psi)\mapsto M\,,\label{eq:sigma}
\end{align}
whose kernel is given by $\mathbb{Z}_2=\bigl\{(\id,\pm 1)\bigr\}$, which shows that $\widetilde{\mathcal{G}}$ is a double cover and $\mathcal{G}=\widetilde{\mathcal{G}}/\mathbb{Z}_2$.

Let us recall that our construction of $\widetilde{\mathcal{G}}$ and the definitions of the circle function $\varphi(M)$ and the cocycle function $\eta(M_1,M_2)$ are all constructed with respect to a reference complex structure $J$. The following proposition now explains what happens when we move to a different reference complex structure $\tilde{J}$.

\begin{proposition}\label{prop:switch-J}
Given a group element $(M,\psi)\in\widetilde{\mathcal{G}}$ where $\psi$ is computed with respect to $J$, we can compute
\begin{align}
    \tilde{\psi}=\psi e^{\ii(\eta(T^{-1},M)+\eta(T^{-1}M,T))/2}\,,
\end{align}
where $(M,\tilde{\psi})$ represents the same group element with respect to $\tilde{J}=TJT^{-1}$, while $\eta(M_1,M_2)$ is still the cocycle function defined with respect to $J$.
\end{proposition}
\begin{proof}
Our definition~\eqref{eq:circle-function} of $\varphi$ is invariant under basis changes if we re-express both $M$ and $J$ into a new basis. A change of basis $T$, which we can extend to a double cover element $(T,\psi_T)$ with inverse $(T^{-1},\psi_T^*)$, can be seen both as a passive transformation of the basis and as an active transformation of the group elements. That is, the $\tilde{\psi}$ of an element $(M,\psi)$ with respect to a transformed complex structure $\tilde{J}=TJT^{-1}$ must be the same as the $\psi$ of the group element $(M,\psi)$ transformed by $(T^{-1},\psi_T^*)$, that is
\begin{align}
    &(T^{-1},\psi_{T}^*)\cdot (M,\psi)\cdot(T,\psi_T)\nonumber\\
    &\hspace{10mm}=(T^{-1}MT,\psi e^{\ii(\eta(T^{-1},M)+\eta(T^{-1}M,T))/2}) \,.
\end{align}
This gives the result by direct inspection of the expression above, in which we used the multiplication rule from~\eqref{eq:group-multiplication-double-cover}. We could have also multiplied in reverse order yielding $\tilde{\psi}=\psi e^{\ii(\eta(M,T)+\eta(T^{-1},MT))/2}$, which is also valid and yields the same result.
\end{proof}

\textbf{Fermions.} The fact that $\varphi(M)$ and thus also $\psi$ for $(M,\psi)\in\widetilde{\mathcal{G}}$ is not defined for those group elements $M\in\mathcal{B}_{\mathcal{G}}$ may appear to be a severe problem for general calculations in the double cover, even if it is a set of measure zero. However, the following proposition shows that this something that we can easily deal with in practice.

\begin{proposition}
Given two group elements $M_1,M_2\in\mathcal{G}$, we can always choose a reference complex structure $\tilde{J}$ (rather than $J$), such that $\tilde{\varphi}(M_1)$, $\tilde{\varphi}(M_2)$, $\tilde{\varphi}(M_1M_2)$ and $\tilde{\eta}(M_1,M_2)$ are defined, where $\tilde{\varphi}$ and $\tilde{\eta}$ are computed with respect to $\tilde{J}$. \label{prop:choose-J}
\end{proposition}
\begin{proof}
Every choice of a reference complex $J$ equips the group $\mathcal{G}$ with the fiber bundle structure discussed in section~\ref{sec:principal-fiber-bundles}, which then also has a quasi-boundary $\mathcal{B}_{\mathcal{G}}=\{M\in\mathcal{G}\,|\,\det(\id+\Delta_M)=0\}$ on which $\varphi$ is not defined. Changing $J$ continuously will change the quasi-boundary $\mathcal{B}_{\mathcal{G}}$ within $\mathcal{G}$ continuously. Furthermore, changing the matrix entries of $J$ will change the spectrum $\Delta_M=-MJM^{-1}J$ and its determinant continuously. This implies in particular that if for a given $M$ and $J$, we have that $\det(\id+\Delta_M)\neq 0$, we can always perturb $J$ in a neighborhood, where $\det(\id+\Delta_M)\neq 0$ throughout.\\
We now consider the case where $\det(\id+\Delta_M)=0$, which implies that $\Delta_M$ must have eigenvalue $-1$. Due to the spectral properties of $\Delta_M$ discussed in section~\ref{sec:cartan}, this implies that we can simultaneously block-diagonalize $\Delta_M$ and $J$ to have one or more blocks of the form
\begin{align}
    [\Delta_M]=\begin{pmatrix}
        -1 & &&\\
        &-1 & &\\
        &&-1&\\
        &&&-1
    \end{pmatrix},\, [J]=\begin{pmatrix}
        &&1&\\
        &&&1\\
        -1&&&\\
        &-1&&
    \end{pmatrix}\,.
\end{align}
Within this block and in this basis, we can then decompose $M=Tu$ with
\begin{align}
    [T]=\begin{pmatrix}
        & 1&&\\
        -1&&&\\
        &&&-1\\
        &&1
    \end{pmatrix}\,,
\end{align}
which allows us to write $\Delta_M=-TJT^{-1}J$. We will now show that we can continuously deform $J$ with a linear change $\delta J$, such that the eigenvalues of the block $[\Delta_M]$ moves away from $-1$. We consider the general linear change
\begin{align}
    [\delta J]=\begin{pmatrix}
    & a & & b\\
    -a &&-b&\\
    &b&&-a\\
    -b&&a&
    \end{pmatrix}\,,
\end{align}
which gives $[\delta\Delta_M]=[TJT^{-1}(\delta J)+T(\delta J)T^{-1}J]$ yielding
\begin{align}
    [\delta \Delta_M]=\begin{pmatrix}
    &  & & 2a\\
     &&-2a&\\
    &2a&&\\
    -2a&&&
    \end{pmatrix}\,.
\end{align}
The linear change of the eigenvalues of $\Delta_M$ is then given by $\lambda=-1\pm 2\ii a$, which allows us to move away from the eigenvalues $-1$. This means we can apply a finite change $J\to \tilde{J}$ in each block in the direction of $[\delta J]$ with $a\neq 0$, such that $\det(\id+\tilde{\Delta}_M)\neq 0$, where $\tilde{\Delta}_M$ was computed with respect $\tilde{J}$. We can apply this procedure first for $M=M_1$, then for $M=M_2$ and finally for $M=M_1M_2$. Each time, we perform a finite change of the reference complex structure, until $\det(\id+\tilde{\Delta}_{M})\neq 0$ for all $M\in\{M_1,M_2,M_1M_2\}$. Note further that there is no danger that we accidentally make the determinant for another $M$ vanish again, as long as our subsequent finite change of $J$ is sufficiently small, because we know that once the determinant is non-zero, there exists a local neighborhood of $J$ where the determinant is non-zero everywhere. For this choice of reference complex structure $\tilde{J}$, we can then be sure that $\tilde{\varphi}(M_1)$, $\tilde{\varphi}(M_2)$, $\tilde{\varphi}(M_1M_2)$ and $\tilde{\eta}(M_1,M_2)$ are all well-defined.
\end{proof}

What we have thus shown is that for any group elements $(M_1,\psi_1)$ and $(M_2,\psi_2)$, we can choose $\tilde{J}$, such that the product rule~\eqref{eq:group-multiplication-double-cover} can be executed without problems. To implement this change, we need to move from $\psi_i\to\tilde{\psi}_i$ according to proposition~\ref{prop:switch-J}, before performing the multiplication. In practice for numerical simulations, a Haar-randomly chosen complex structure $J$ will lead, with high probability, to a valid representation of the required unitary.

More generally, a more rigorous construction of $\widetilde{\mathcal{G}}$ for fermions can be achieved by representing group elements by equivalence classes\footnote{The same spirit can be used to define pure Gaussian states with $\ket{J}$ if we want to keep track of the phase, namely by equivalence classes $[(J,\arg\braket{J|\tilde{J}},\ket{\tilde{J}}]$, where we fix a reference Gaussian state vector $\ket{\tilde{J}}$ and then represent $\ket{J}$ by $J$ and the complex phase $\arg\braket{J|\tilde{J}}$ with respect to $\tilde{J}$, provided $\braket{J|\tilde{J}}\neq 0$. This approach was discussed in~\cite{Dias:2023arXiv230712912D}.} $[(M,\psi,J)]$, consisting of all possible different choices of the reference complex structure $J$ and the associated $M$ and $\psi$ satisfying $\psi^2=\varphi(M)$. Different representatives of the same equivalence class are related via proposition~\ref{prop:switch-J}. In this way, each group element has a valid representation. The group product can be defined according to~\eqref{eq:group-multiplication-double-cover} by choosing a reference $J$ such that the corresponding representative exists in the equivalence classes of all elements involved in the product. This is always possible as guaranteed by proposition~\ref{prop:choose-J}.

\section{Unitary representation and closed formulas}\label{sec:unitary-rep-results}
The present section aims to build on the previous two sections to understand the structure of Gaussian unitary transformations. While section~\ref{sec:Setup} introduced Gaussian unitaries $\mathcal{U}(M)$ defined up to an overall sign, section~\ref{sec:DoubleCover} put the quantum theory aside and constructed the double cover $\widetilde{\mathcal{G}}$ based on the circle function $\varphi(M)$, which was purely defined in terms of objects of the classical phase space. Now we will see how circle function $\varphi(M)$ and the object $\psi_M$ for double cover group elements $(M,\psi_M)\in\widetilde{\mathcal{G}}$ are naturally related to expectation values of Gaussian unitaries and allow us to construct a natural parametrization of the Gaussian unitary group.

\subsection{\label{sec:squared-expectation}Squared expectation value}
The circle function $\varphi(M)$ was introduced in~\eqref{eq:circle-function} with respect to a complex structure $J$ without any reference to quantum theory. The following result shows that the circle function and more generally also its unnormalized version $\overline{\det}(\frac{M-JMJ}{2})$ is directly relates to the squared expectation value $\braket{J|\mathcal{U}(M)|J}^2$. This provides an explicit link between our construction of the double cover and the projective unitary representation $\mathcal{U}(M)$.

\begin{result1}
\label{res:squared-expectation}
Given a Gaussian state $\ket{J}$ in a bosonic or fermionic system and a Gaussian unitary $\pm\mathcal{U}(M)$, only defined up to an overall sign, the squared expectation value is
\begin{align}
    \langle J|\mathcal{U}(M)|J\rangle^2
    =\begin{cases}
        \displaystyle\frac{1}{\overline{\det}\left(\frac{M-JMJ}{2}\right)} & \normalfont{\textbf{(bosons)}}\\[5mm]
        \,\overline{\det}\left(\frac{M-JMJ}{2}\right) & \normalfont{\textbf{(fermions)}}
        \end{cases}\,,\label{eq:exp-squared}
\end{align}
where $\overline{\det}$ was defined in~\eqref{eq:detbar}.
\end{result1}
\begin{proof}
We first Cartan decompose our group element $M$ as $M=Tu$ with respect to $J$, as explained in section~\ref{sec:cartan}. The properties~\eqref{eq:U(M)-product} of $\mathcal{U}(M)$ then imply
\begin{align}
    \mathcal{U}(M)=\pm\mathcal{U}(T)\mathcal{U}(u)\,.
\end{align}
With this, we can compute
\begin{align}
    \braket{J|\mathcal{U}(M)|J}^2&=\braket{J|\mathcal{U}(T)\mathcal{U}(u)|J}^2\\
    &=\braket{J|\mathcal{U}(T)|J}^2\braket{J|\mathcal{U}(u)|J}^2\\
    &=\braket{J|e^{\widehat{K}_+}|J}^2\braket{J|e^{\widehat{K}_-}|J}^2\,,\label{eq:prodeKpeKm}
\end{align}
where we showed in~\eqref{eq:U(u)-relation} that $\ket{J}$ is an eigenvector of $\mathcal{U}(u)$. There exist $K_+\in\mathfrak{u}_\perp(N)$ and $K_-\in\mathfrak{u}(N)$, such that $T=e^{K_+}$ and $u=e^{K_-}$. This is because the exponential map is surjective on $\mathfrak{u}_\perp(N)\to\exp(\mathfrak{u}_\perp(N))$ by definition and on $\mathfrak{u}(N)\to\mathrm{U}(N)$ due to maximal torus theorem for compact connected Lie groups~\cite{hall2013lie}. In practice, we can compute $K_+=\log(T)$ and $K_-=\log(u)$, where we only need to pay special attention at the branch cut, when $u$ has eigenvalue pairs $-1$, as $K_-$ then needs to have eigenvalues $\pm \ii \tfrac{\pi}{2}$ to be a real linear map.\\
We now evaluate the two factors from~\eqref{eq:prodeKpeKm} individually:
\begin{itemize}
    \item \textbf{First factor.} We have $\mathcal{U}(T)=\pm e^{\widehat{K}_+}$. It was shown in equation~(174) of~\cite{hackl2021bosonic} that
    \begin{align}
        \hspace{5mm}\braket{J|e^{\widehat{K}_+}|J}=\begin{cases}
        \det^{\frac{1}{8}}(\id-L^2) & \textbf{(bosons)}\\
        \det^{-\frac{1}{8}}(\id-L^2) & \textbf{(fermions)}
        \end{cases},
    \end{align}
    where $L=\tanh{K_+}$. Using $1-\tanh^2{x}=1/\cosh^2{x}$, lemma~\eqref{lem:coshKplus} then implies $\det^{-\frac{1}{4}}(\id-L^2)=\overline{\det}(\cosh{K_+})$, which finally yields
    \begin{align}
        \hspace{5mm}\braket{J|e^{\widehat{K}_+}|J}^2=\begin{cases}
        \displaystyle\frac{1}{\overline{\det}(\cosh{K_+})} & \textbf{(bosons)}\\[5mm]
        \,\overline{\det}(\cosh{K_+}) & \textbf{(fermions)}
        \end{cases}.\label{eq:JTJ2}
    \end{align}
    \item \textbf{Second factor.} We now consider the expectation value $\braket{J|e^{\widehat{K}_-}|J}$. Due to $[K_-,J]=0$, we can quasi-diagonalize $K_-$ while retaining $J$'s standard form~\eqref{eq:J-standard-form} as
    \begin{align}
        K_-\equiv\left(\begin{array}{c|c}
             0 & D \\
             \hline
             -D & 0
        \end{array}\right)\,,
    \end{align}
    where $D=\mathrm{diag}(\omega_1,\dots,\omega_N)$ with some $\omega_i$ being potentially negative. We thus have
    \begin{align}
        e^{\widehat{K}_-}=\begin{cases}
        e^{-\ii \sum_i(\hat{n}_i+\frac{1}{2})\omega_i} & \textbf{(bosons)}\\[1mm]
        e^{-\ii \sum_i(\hat{n}_i-\frac{1}{2})\omega_i} & \textbf{(fermions)}
        \end{cases}\,.
    \end{align}
    Using $\hat{n}_i\ket{J}=0$ in this basis, leads to
    \begin{align}
        \hspace{5mm}\braket{J|e^{\widehat{K}_-}|J}^2=\begin{cases}
        e^{-\ii \sum_{i}\omega_i}=e^{-\overline{\mathrm{tr}}K_-} & \textbf{(bosons)}\\[2mm]
        e^{+\ii \sum_{i}\omega_i}=e^{+\overline{\mathrm{tr}}K_-} & \textbf{(fermions)}
        \end{cases}\,.
    \end{align}
    Using $\overline{\det}(u)=e^{\overline{\mathrm{tr}}(K_-)}$ then yields
    \begin{align}
        \hspace{3mm}\braket{J|\mathcal{U}(u)|J}^2=\begin{cases}
        \displaystyle\frac{1}{\overline{\det}(u)} & \normalfont{\textbf{(bosons)}}\\[5mm]
        \,\overline{\det}(u) & \normalfont{\textbf{(fermions)}}
        \end{cases}\,.\label{eq:JUuJ2}
    \end{align}
\end{itemize}
We use the determinant property to find
\begin{align}
    \overline{\det}(\cosh{K_+})\overline{\det}(u)=\overline{\det}(\cosh(K_+)u)\,.
\end{align}
We can write out the argument explicitly to find
\begin{align}
\begin{split}\label{eq:combine-pieces}
    \cosh(K_+)u&=\tfrac{1}{2}(e^{K_+}u+e^{-K_+}u)\\
    &=\tfrac{1}{2}(e^{K_+}u+e^{-K_+}(-J^2)u)\\
    &=\tfrac{1}{2}(e^{K_+}u-Je^{K_+}uJ)\\
    &=\tfrac{1}{2}(M-JMJ)\,,
\end{split}
\end{align}
where we used $J^2=-\id$, $[J,u]=0$, $\{J,K_+\}=0$ and $M=e^{K_+}u$. Plugging~\eqref{eq:JTJ2} and~\eqref{eq:JUuJ2} into~\eqref{eq:prodeKpeKm} and using~\eqref{eq:combine-pieces} then yields the desired result~\eqref{eq:exp-squared}.
\end{proof}

If we are only interested in the absolute value, we can ignore the sign ambiguity and use Corollary~\ref{lem:determinant} to find
\begin{align}
    \hspace{-3mm}\bigl|\braket{J|\mathcal{U}(M)|J}\bigr|=\begin{cases}
    \det^{-\frac{1}{4}}(\frac{M-JMJ}{2}) & \textbf{(bosons)}\\
    \det^{+\frac{1}{4}}(\frac{M-JMJ}{2}) & \textbf{(fermions)}
    \end{cases}
    .\label{eq:abs-val}
\end{align}
This is equivalent to the respective expression in~\cite{hackl2021bosonic}.

Vice versa, if one is only interested in the complex phase, formula~\eqref{eq:exp-squared} can be used to assign physical meaning to the circle function introduced in~\eqref{eq:circle-function}, namely
\begin{align}
    \label{eq:phase-exp-squared}
    \frac{\braket{J|\mathcal{U}(M)|J}^2}{|\braket{J|\mathcal{U}(M)|J}|^2}=\begin{cases}
    \varphi^*(M) & \textbf{(bosons)}\\
    \varphi(M) & \textbf{(fermions)}
    \end{cases}\,.
\end{align}
Rawsley used the circle function in~\cite{rawnsley2012universal} to construct the universal cover of the symplectic group, but we can use the newly gained physical intuition to understand the key trouble when trying to repeat this construction for fermionic systems. The action of $\mathcal{U}(M)$ on $\ket{J}$ always yields another Gaussian state $\ket{J'}=\mathcal{U}(M)\ket{J}$. For bosons, $\braket{J|J'}\neq 0$ as there are no two orthogonal bosonic Gaussian states, while for fermions we have 
states with $\braket{J|J'}=0$. For such $\mathcal{U}(M)$, the circle function $\varphi$ is not defined, as zero does not have a uniquely defined complex phase. We already discussed in section~\ref{sec:construction-double-cover} how we can deal with this by moving to a different reference $\tilde{J}$, such that $\braket{\tilde{J}|J'}\neq0$. However, for computing $\braket{J|\mathcal{U}(M)|J}$ or its square, this is not a problem, as we do not need to know a complex phase whenever $\braket{J|\mathcal{U}(M)|J}=0$ and, as discussed previously, such $M$ form a set of measure zero and we can always more to a different complex structure $\tilde{J}$ with $\braket{\tilde{J}|\mathcal{U}(M)|\tilde{J}}\neq 0$.

\subsection{Unitary representation of the double cover}\label{sec:rep-double-cover}

With the results established in the previous section, we are now able to address one of the main aspects of our analysis, that is how our description of group elements in the double cover, introduced in section~\ref{sec:construction-double-cover}, can be used to parametrize a proper representation in terms of Gaussian unitaries.

More specifically, the relation~\eqref{eq:abs-val} implies that for any Gaussian unitary $\mathcal{U}(M)$ and Gaussian state $\ket{J}$, we have that $\braket{J|\mathcal{U}(M)|J}=\psi D(M)$. Here $D(M)$ is a positive number that depends solely on the group element $M$ as
\begin{equation}
    D(M)={\det}^{\mp \frac{1}{4}} \left(\frac{M-JMJ}{2}\right) \,, \label{eq:def-D}
\end{equation}
where the $-/+$ signs refer to bosons and fermions respectively.
The phase $\psi$, on the other hand, is not fully determined by $M$. The only constraint that $\psi$ has to respect is that $\psi^2=\varphi(M)$ (or $\varphi^*(M)$ for bosons) according to~\eqref{eq:phase-exp-squared}, meaning that there are always two possible values of $\psi$ for each $M$ corresponding to two inequivalent unitaries. This is a direct consequence of the fact that, as already discussed, the unitaries $\mathcal{U}$ cannot be fully parametrized just by elements of the group $\mathcal{G}$, of which they do not form a proper representation.

To resolve this, we extend the parametrisation of a unitary $\mathcal{U}(M)$ to also include the parameter $\psi$. For every choice of $M\in\mathcal{G}$ and compatible $\psi$ there now exists a single, fully-defined Gaussian unitary $\mathcal{U}(M,\psi)$. As $\psi$ has to satisfy relation~\eqref{eq:phase-exp-squared}, we see that $(M,\psi)$ is an element of the double cover group $\widetilde{\mathcal{G}}$ as defined in~\eqref{eq:double-cover-def}. So we have effectively established an isomorphism between $\widetilde{\mathcal{G}}$ and the set of Gaussian unitaries. As we will now show, this actually gives a proper group representation, meaning that the unitaries $\mathcal{U}(M,\psi)$ follow the same product rules~\eqref{eq:group-multiplication-double-cover} as the double cover elements. This construction is expressed more rigorously by the following result.

\begin{result2}
The map $\mathcal{U}: \widetilde{\mathcal{G}}\to \mathrm{Lin}(\mathcal{H})$ characterized by the conditions
\begin{align}
    \mathcal{U}^\dagger(M,\psi)\hat{\xi}^a\mathcal{U}(M,\psi)&=M^a{}_b\hat{\xi}^b\,,\label{eq:cond1}\\
    \frac{\braket{J|\mathcal{U}(M,\psi)|J}}{|\braket{J|\mathcal{U}(M,\psi)|J}|}&=\begin{cases}
    \psi^* & \normalfont{\textbf{(bosons)}}\\
    \psi & \normalfont{\textbf{(fermions)}}
    \end{cases}\,,\label{eq:cond2}
\end{align}
forms a unitary representation satisfying
\begin{align}
    \hspace{-3mm}\mathcal{U}(M_1,\psi_1)\mathcal{U}(M_2,\psi_2)=\mathcal{U}(M_1M_2,\psi_1\psi_2e^{\ii\eta(M_1,M_2)/2})\,.\label{eq:group-multiplication-representation}
\end{align}
This in particular implies $\mathcal{U}^\dagger(M,\psi)=\mathcal{U}(M^{-1},\psi^*)$.
\end{result2}
\begin{proof}
As explained in proposition 6 of~\cite{hackl2021bosonic}, condition~\eqref{eq:cond1} characterizes $\mathcal{U}(M,\psi)$ uniquely up to a complex phase. This complex phase is then fixed by condition~\eqref{eq:cond2}, which is compatible with~\eqref{eq:phase-exp-squared} for any $(M,\psi)$ in the double cover group $\widetilde{\mathcal{G}}$ as defined in~\eqref{eq:double-cover-def}. So the unitary $\mathcal{U}(M,\psi)$ is unique and well-defined\footnote{Strictly speaking, this is the case for all group elements for bosons, while for fermions there is the measure-zero quasi-boundary $\mathcal{B}_{\mathcal{G}}$ from~\eqref{eq:B_G}, where $\braket{J|\mathcal{U}(M,\psi)|J}=0$ and relation~\eqref{eq:cond2} is ill-defined.} and $\mathcal{U}$ is indeed an injective map on $\widetilde{\mathcal{G}}$.\\
To prove the product rule~\eqref{eq:group-multiplication-representation}, consider two fixed elements $(M_1,\psi_1)$ and $(M_2,\psi_2)$. For any such choice it is always possible to find a continuous deformation $(M_2(\tau),\psi_2(\tau))$ which connects the second double cover element to plus or minus identity. More precisely, we can find the functions $M_2(\tau)$ and $\psi_2(\tau)$ continuous in $\tau$ such that $(M_2(1),\psi_2(1))=(M_2,\psi_2)$ and $(M_2(0),\psi_2(0))=(\id,\pm 1)$. This can be done, for instance, by considering the Cartan decomposition $M_2=e^{K_+}u$ and defining $M_2(\tau)=e^{\tau K_+} e^{\tau \log u}$. We then define $\psi_2(\tau)=\sigma_2(\tau)\sqrt{\varphi(M_2(\tau))}$, where $\sigma_2(\tau)$ is just a sign that flips whenever $\varphi(M_2(\tau))$ crosses the branch cut of the square root in order to keep $\psi_2(\tau)$ continuous.\\
Consider now the product $\mathcal{U}(M_1,\psi_1)\mathcal{U}(M_2(\tau),\psi_2(\tau))$. This is for sure a Gaussian unitary that satisfies condition~\eqref{eq:cond1} with respect to $M=M_1 M_2(\tau)$. Regarding condition~\eqref{eq:cond2}, let us define $\psi_{12}(\tau)$ to be the phase of $\braket{J|\mathcal{U}(M_1,\psi_1)\mathcal{U}(M_2(\tau),\psi_2(\tau))|J}^*$ for bosons and of $\braket{J|\mathcal{U}(M_1,\psi_1)\mathcal{U}(M_2(\tau),\psi_2(\tau))|J}$ for fermions. Result~\eqref{eq:phase-exp-squared} implies that it must satisfy 
\begin{align}
    \psi_{12}(\tau)^2&=\varphi(M_1 M_2(\tau)) \\
    &=\varphi(M_1)\,\varphi(M_2(\tau))\,e^{\ii \eta(M_1,M_2(\tau))} \\
    &=\psi_1^2\, \psi_2(\tau)^2 \, e^{\ii \eta(M_1,M_2(\tau))}\,,
\end{align}
where we have used~\eqref{eq:cocycle-condition}. This means that we must have
\begin{equation}
    \psi_{12}(\tau)=\sigma_{12}(\tau)\,\psi_1\, \psi_2(\tau) \, e^{\ii \eta(M_1,M_2(\tau))/2}\,, \label{eq:psi-multiplication-representation-tau}
\end{equation}
where $\sigma_{12}(\tau)$ is a sign. Notice now that $\psi_{12}(\tau)$ must be continuous in $\tau$, because we have explicitly defined $\mathcal{U}(M_2(\tau),\psi_2(\tau))$ to be a continuous function of $\tau$. Given that $M_2(\tau)$ and $\psi_2(\tau)$ are continuous by construction and that $e^{\ii \eta(M_1,M_2)/2}$ is continuous by proposition~\ref{prop:eta-continuity}, this can only be possible if the sign $\sigma_{12}$ does not change with $\tau$. The constant value $\sigma_{12}(\tau)=+1$ can be fixed by evaluating~\eqref{eq:psi-multiplication-representation-tau} for $\tau=0$ and observing that $e^{\ii \eta(M_1,M_2(0))/2}=e^{\ii \eta(M_1,\id)/2}=+1$ according to~\eqref{eq:eta2} and that
\begin{align}
    \psi_{12}(0)&=\frac{\braket{J|\mathcal{U}(M_1,\psi_1)\,\mathcal{U}(\id,\psi_2(0))|J}}{|\braket{J|\,\mathcal{U}(M_1,\psi_1)\mathcal{U}(\id,\psi_2(0))|J}|} \\
    &=\frac{\braket{J|\mathcal{U}(M_1,\psi_1)\,(\psi_2(0)\cdot\id)|J}}{|\braket{J|\mathcal{U}(M_1,\psi_1)\,(\psi_2(0)\cdot\id)|J}|} \\
    &=\frac{\braket{J|\mathcal{U}(M_1,\psi_1)|J}}{|\braket{J|\mathcal{U}(M_1,\psi_1)|J}|} \psi_2(0)\\
    &=\psi_1\,\psi_2(0) 
\end{align}
for fermions (and the same result can be derived for bosons by including the appropriate complex conjugations).
Evaluating~\eqref{eq:psi-multiplication-representation-tau} for $\tau=1$ shows that $\mathcal{U}(M_1,\psi_1)\mathcal{U}(M_2,\psi_2)$ satisfies condition~\eqref{eq:cond2} with $\psi=\psi_1\, \psi_2\, e^{\ii \eta(M_1,M_2)/2}$, implying the final result~\eqref{eq:group-multiplication-representation}.
Relation $\mathcal{U}^\dagger(M,\psi)=\mathcal{U}(M^{-1},\psi^*)$ then follow from~\eqref{eq:group-multiplication-representation} and~\eqref{eq:eta3} by evaluating $\mathcal{U}(M^{-1},\psi^*)\mathcal{U}(M,\psi)=\mathcal{U}(\id,+1)=\id$.
\end{proof}

The main insight of the proof above is that the rule~\eqref{eq:group-multiplication-double-cover}, that we have introduced on purely abstract grounds, is actually the only possible choice to define a consistent and \emph{continuous} multiplication for objects parametrised by $(M,\psi)$ which must satisfy $\psi^2=\varphi(M)$. This means in particular that the function $e^{\ii \eta(M_1,M_2)/2}$ also has a concrete definition in terms of quantum mechanical operators. 
More precisely, using the shorthand $\mathcal{U}_i=\mathcal{U}(M_i,\psi_i)$ we have
\begin{align}
    &\braket{J|\mathcal{U}_1|J}\braket{J|U_1^\dag\,\mathcal{U}_2|J}\braket{J|\mathcal{U}_2^\dag|J}=\nonumber\\
    &\hspace{10mm} = D(M_1) D(M_1^{-1}M_2) D(M_2^{-1}) \, e^{\ii\eta(M_1^{-1},M_2)/2} \,. \label{eq:A-as-eta}
\end{align}
That is, $e^{\ii\eta(M_1^{-1},M_2)/2}$ is the phase of the product $\braket{J|\mathcal{U}_1|J}\braket{J|U_1^\dag\,\mathcal{U}_2|J}\braket{J|\mathcal{U}_2^\dag|J}$, which turns out to depend only on the group elements $M_1$ and $M_2$.

\subsection{Full expectation value}\label{sec:exact-expectation}
Given a unitary operator $e^{\widehat{K}}$, where $\widehat{K}$ is a quadratic operator, it is clear that this corresponds to a Gaussian unitary $\mathcal{U}(M,\psi)$ with $M=e^K$. In order to also determine $\psi$, however, it is necessary to evaluate $\braket{J|e^{\widehat{K}}|J}$, including its correct phase.

One possible strategy to do this is to observe that, for fixed $J$, the trajectory $\widetilde{\gamma}(t)=e^{t\widehat{K}}$ is a continuous path on the double cover $\widetilde{\mathcal{G}}$, such that its projection under the map $\sigma$ from~\eqref{eq:sigma} is given by $\gamma(t)=e^{tK}$. To determine the correct sign in the square root of~\eqref{eq:exp-squared}, we would need to follow the path $\langle J|e^{t\widehat{K}}|J\rangle^2$ for $t\in[0,1]$ written as a determinant and count how often we transverse through the branch cut and take the correct sign to keep the path continuous (plus sign for an even number, minus sign for an odd number of crossings). This is cumbersome or requires numerical integration.

Instead, we will take a different route leading to a closed formula for the complex phase $\arg\braket{J|e^{\widehat{K}}|J}$. 
While the result will not be as compact as~\eqref{eq:exp-squared}, it will still be easy to evaluate numerically and in principle also allows further analytical studies. We would like to stress though that depending on the desired application the previous compact result~\eqref{eq:exp-squared} may already be everything that is needed.

The basic idea is that everything would be easy if $\widehat{K}=-\ii\sum^{N}_{i=1}\omega_i(\hat{n}_i\pm\frac{1}{2})$ for bosons ($+$) and fermions ($-$), respectively, such that $\ket{J}$ is the ground state of all these number operators with $\hat{n}_i\ket{J}=0$. In this case, we would immediately find\footnote{Throughout, we mean by $\arg$ the multi-valued function whose values are only defined modulo $2\pi$ rather than the principal value defined on a specific range.}
\begin{align}
    \arg\braket{J|e^{\widehat{K}}|J}=\mp\sum^N_{i=1}\frac{\omega_i}{2}=\pm\frac{1}{4}\Tr(JK)\,,\label{eq:arg-eigenstate}
\end{align}
where we used that $\braket{J|\widehat{K}|J}=\pm\frac{\ii}{4}\Tr(JK)$. Unfortunately, we cannot expect in general that $\ket{J}$ is the ground state or more generally an eigenstate of $\widehat{K}$. However, we will argue that we can always find a basis transformation $M$ that turns $\widehat{K}$ into a different operator $\widehat{K_J}=\widehat{MKM^{-1}}$,
such that $\ket{J}$ is an eigenstate, or in the case of bosons, such that $\ket{J}$ is an eigenstate of the piece of $\widehat{K_J}$ related to the imaginary eigenvalues of $K$. We can then use the cocycle group multiplication~\eqref{eq:group-multiplication-double-cover} to compute $\braket{J|e^{\widehat{K}}|J}$ from $\braket{J|e^{\widehat{K_J}}|J}$.

It is interesting that both, bosons and fermions, come with their own advantages and drawbacks for this calculation, as indicated in table~\ref{tab:ad-draws}. We demonstrate the utility of our respective formulas for randomly generated quadratic Hamiltonians in figure~\ref{fig:example-trajectories}.

\renewcommand{\arraystretch}{1.2}
\begin{table}[t]
    \centering
    \caption{\emph{Advantages and drawbacks for the bosonic and fermionic calculation, respectively.}}
    \label{tab:ad-draws}
    \begin{tabular}{l p{.2cm} p{3cm} p{.2cm}p{3cm}}
    \hline
    \hline
    && \textbf{Advantage}     && \textbf{Drawback} \\
    \hline
    \textbf{Bosons} &&  $\varphi$ is defined everywhere.    && $K$ may have a non-diagonalizable part.\\
    \textbf{Fermions} && $K$ can always be diagonalized. && $\varphi$ is not defined everywhere.\\
    \hline
    \hline
    \end{tabular}
\end{table}

\subsubsection{Bosonic systems}
The challenge of the bosonic case is that the generator $K\in\mathfrak{sp}(2N,\mathbb{R})$ may not be diagonalizable. Moreover, this means that $\widehat{K}$ may also not be fully diagonalizable with a Gaussian eigenstate. To give some intuition, we can consider the following three example Hamiltonians
\begin{align}
    \hat{H}_1=\frac{1}{2}(\hat{p}^2+\hat{q}^2)\,,\quad\hat{H}_2=\frac{1}{2}(\hat{p}^2-\hat{q}^2)\,,\quad\hat{H}_3=\frac{1}{2}\hat{p}^2\,,
\end{align}
associated to the symplectic generators
\begin{align}
    K_1=\begin{pmatrix}
        0 & 1\\
        -1 & 0
    \end{pmatrix}\,,\,K_2=\begin{pmatrix}
        0 & 1\\
        1 & 0
    \end{pmatrix}\,,\,K_3=\begin{pmatrix}
        0 & 1\\
        0 & 0
    \end{pmatrix}\,.
\end{align}
Only $K_1$ has imaginary eigenvalues implying that $\hat{H}_1$ has a Gaussian eigenstate, which is also its ground state. In contrast, $K_2$ has real eigenvalues $\pm 1$ and $K_3$ is nilpotent, which implies that neither $\hat{H}_2$ nor $\hat{H}_3$ has Gaussian eigenstates. It was shown in~\cite{hackl2018entanglement} that evolving a Gaussian state with each of these three types of Hamiltonians yields different characteristic behavior of the entanglement entropy, namely bounded oscillations, linear growth and logarithmic growth.

The basic idea is that we can always decompose $K=K_I+K'$, where $K_I$ is of type $K_1$, \ie has purely imaginary eigenvalues and there exists a Gaussian eigenstate---not necessarily ground state, as the eigenmodes of the Hamiltonian $\omega (\hat{n}_i+\tfrac{1}{2})$ could also have negative $\omega$. The remaining piece $K'$ can be decomposed as $K'=K_R+K_N$, where $K_R$ is of type of $K_2$ with purely real eigenvalues and $K_N$ is nilpotent, so of type $K_3$. Bringing quadratic Hamiltonians and thus, equivalently, symplectic Lie generators $K$ into a normal form is a well-known problem in classical mechanics, as already discussed by Arnolds in~\cite{arnol2013mathematical}. More recently, the authors of~\cite{Kustura:2019PhRvA..99b2130K} presented a modern review of this problem, from which we will draw some key ingredients as discussed in appendix~\ref{app:normal-form}. Let us emphasize that only pieces of type $K_3$ (nilpotent part) make this careful consideration necessary, while for diagonalizeble generators containing only parts of type $K_1$ and $K_2$ (no nilpotent parts), the resulting calculation simplifies dramatically and the analysis of appendix~\ref{app:normal-form} based on~\cite{Kustura:2019PhRvA..99b2130K} is not needed.

\begin{result3a}[Bosons]
Given $K\in\mathfrak{sp}(2N,\mathbb{R})$ and a pure Gaussian state $\ket{J}$, we first decompose $K=K_I+K'$, such that $[K_I,K']=0$ and $K_I$ has purely imaginary eigenvalues. Furthermore, there exists a complex structure $\tilde{J}$ that takes the standard form~\eqref{eq:J-standard-form} in the basis, where $K$ is quasi-diagonal according to proposition~\ref{prop:K-normal-form} and appropriately rescaled according to lemma~\ref{lem:rescaling}, such that $[K_I,\tilde{J}]=0$, from which we can compute
\begin{align}
\begin{split}\label{eq:main-bosons}
    \arg\braket{J|e^{\widehat{K}}|J}&=\tfrac{1}{4}\Tr(K_I\tilde{J})+\tfrac{1}{2}\eta(T^{-1},e^{K})\\
    &\quad -\arg\overline{\det}\sqrt{T^{-1}\tfrac{e^{K'}-\tilde{J}e^{K'}\tilde{J}}{2}T}\\
    &\quad +\frac{1}{2}\eta(T^{-1}e^{K},T)\,,
\end{split}
\end{align}
where $T=\sqrt{-\tilde{J}J}$ and the square root in the second line must be evaluated, such that the square roots of negative real eigenvalue pairs are conjugate to each other (and thus do not contribute towards the argument). 
\end{result3a}
\begin{proof}
The decomposition $K=K_I+K'$ follows from the Jordan–Chevalley decomposition of a linear map on a finite-dimensional vector space, as explained in appendix~\ref{app:normal-form}. We then use proposition~\ref{prop:K-normal-form} and lemma~\ref{lem:rescaling} to choose a symplectic basis, where $K$ takes an appropriate standard form, and choose the complex structure $\tilde{J}$ to have the standard form~\eqref{eq:J-standard-form} in this basis, such that $[K_I,\tilde{J}]=0$.\\
If we consider $\mathcal{U}(e^{K},\psi_{e^{K}})=e^{\widehat{K}}$, relation~\eqref{eq:cond2} states that $\arg{\braket{J|e^{\widehat{K}}|J}}=-\arg{\psi_{e^K}}$. We can then use proposition~\ref{prop:switch-J} to relate $\psi_{e^K}$, computed with respect to $J$, with $\tilde{\psi}_{e^{K}}$, computed with respect to $\tilde{J}$, as
\begin{align}
\tilde{\psi}_{e^{K}}=\psi_{e^K}e^{\ii[\eta(T^{-1},M)+\eta(T^{-1}M,T)]/2}\,.\label{eq:tildepsi-psi0}
\end{align}
We can then solve for $\psi_{e^{K}}$ after we have computed $\tilde{\psi}_{e^K}$. For this, we compute
\begin{align}
\begin{split}\label{eq:simplify-exp}
    -\arg\tilde{\psi}_{e^K}&=\arg\braket{\tilde{J}|e^{\widehat{K}}|\tilde{J}}\\
    &=\arg\braket{\tilde{J}|e^{\widehat{K}_I}e^{\widehat{K}'}|\tilde{J}}\\
    &=\tfrac{1}{4}\Tr(K_I\tilde{J})+\arg\braket{\tilde{J}|e^{\widehat{K}'}|\tilde{J}}\,,
\end{split}
\end{align}
where we used that $\ket{\tilde{J}}$ is an eigenstate of $\widehat{K}_I$ (due to $[K_I,\tilde{J}]=0$) with eigenvalues $\braket{\tilde{J}|\widehat{K}|\tilde{J}}=\frac{\ii}{4}\Tr(K_I\tilde{J})$. From~\eqref{eq:exp-squared}, we know that
\begin{align}
    \arg\braket{\tilde{J}|e^{\widehat{K}'}|\tilde{J}}^2=-\arg\widetilde{\overline{\det}}(\tfrac{e^{K'}-\tilde{J}e^{K'}\tilde{J}}{2})\,,
\end{align}
where the $\widetilde{\overline{\det}}$ indicate that we decompose according to appendix~\ref{app:normal-form} with respect to $\tilde{J}$. By inserting $T^{-1}$ and $T$, we can change this to the original decomposition with respect to $J$, \ie
\begin{align}\label{eq:det-expressions}
\widetilde{\overline{\det}}(\tilde{C}_{e^{K'}})=\overline{\det}(T^{-1}\tilde{C}_{e^{K'}}T)\,,
\end{align}
where $\tilde{C}_M=\frac{1}{2}(M-\tilde{J}M\tilde{J})$. In order to evaluate~\eqref{eq:simplify-exp}, we would need to compute the square root of~\eqref{eq:det-expressions}, but doing this naively would only yield $\arg$ only in the interval $(-\tfrac{\pi}{2},\tfrac{\pi}{2}]$. However, proposition~\ref{prop:eigenvalues-C} states that instead the function
\begin{align}
    \widetilde{\overline{\det}}\sqrt{\tilde{C}_{e^{tK'}}}=\overline{\det}\sqrt{T^{-1}\tfrac{e^ {K'}-Je^{K'}J}{2}T}\label{eq:dettildeCeK}
\end{align}
is continuous for $t\in[0,1]$ for an appropriately chosen $\tilde{J}$. By first applying the square root to find $\sqrt{\tilde{C}_{e^{tK'}}}$, we effectively apply the square root of each individual eigenvalue and then take the argument of their products. Therefore, the resulting argument can lie in the full interval $(-\pi,\pi]$ and proposition~\ref{prop:eigenvalues-C} ensures that either none or an even number of eigenvalues will cross the branch cut of the square root along the negative real axis. The only subtlety is that we need to ignore any negative eigenvalues which will appear with even multiplicity, as the correct treatment would be to assign $+\pi$ and $-\pi$ as their complex arguments, which will exactly cancel, so we can also just remove these eigenvalues from our calculation.\\
Finally, we can plug~\eqref{eq:dettildeCeK} and~\eqref{eq:simplify-exp} into~\eqref{eq:tildepsi-psi0}, which we solve for $\psi_{e^K}$ and use $\arg\braket{J|e^ {\widehat{K}}|J}=-\arg(\psi_{e^{K}})$ to arrive at~\eqref{eq:main-bosons}.
\end{proof}

Note that it is still crucial that in formula~\eqref{eq:main-bosons}, we compute the last term as $\overline{\det}\,{\sqrt{C_{e^{K'}}}}$ rather than $\sqrt{\overline{\det}\,{C_{e^{K'}}}}$. The reason is that while the individual eigenvalues of $e^{K'}$ do not cross the negative axis in the complex plane, the determinant as the product of all these eigenvalues may very well cross this axis, so that we must first take the square root of each eigenvalue.

In practice, the decomposition into $K_I$ and $K'$ can be done by first finding a complete set of (potentially generalized) eigenvectors $e^a_i$ with $i=1,\dots, 2N$ and eigenvalues $\lambda_i$. We refer to the associated dual eigenvectors as $\tilde{e}_a^i$ satisfying $e^a_i\tilde{e}_a^j=\delta_i^j$. With this, we have 
\begin{align}
    (K_I)^a{}_b=\sum^{2N}_{i=1}\ii\operatorname{Im}(\lambda_i)\,e^a_i\tilde{e}_b^i
\end{align}
and can then compute $K'=K-K_I$.

\begin{figure*}[t]
    \centering
    \begin{tikzpicture}
    \begin{scope}[xshift=3mm]
    \draw (-4.65,0) node[inner sep=0pt]{\includegraphics[width=.95\columnwidth]{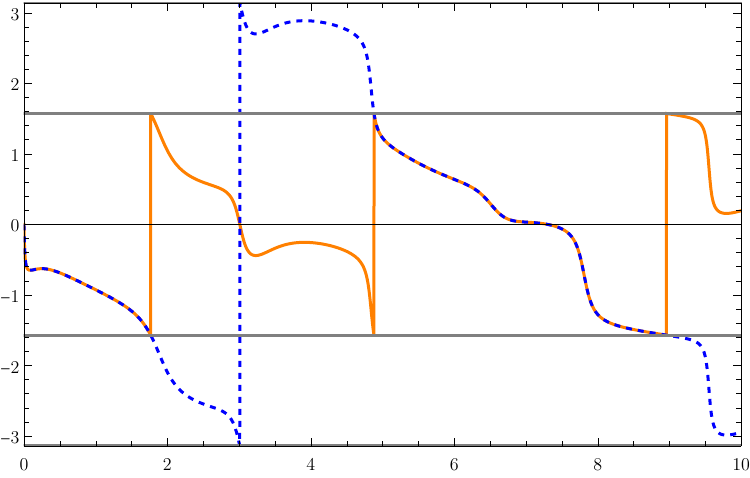}};
    \draw (-4.65,3) node{(a) bosons, $4$ degrees of freedom};
    \draw (-4.58,-2.6) node[font=\footnotesize]{$t$};
    \draw (-9.1,0.2) node[font=\footnotesize,rotate=90]{$\arg{\braket{J|e^{\widehat{K}}|J}}$};
    \end{scope}
    
    \draw (4.65,0) node[inner sep=0pt]{\includegraphics[width=.95\columnwidth]{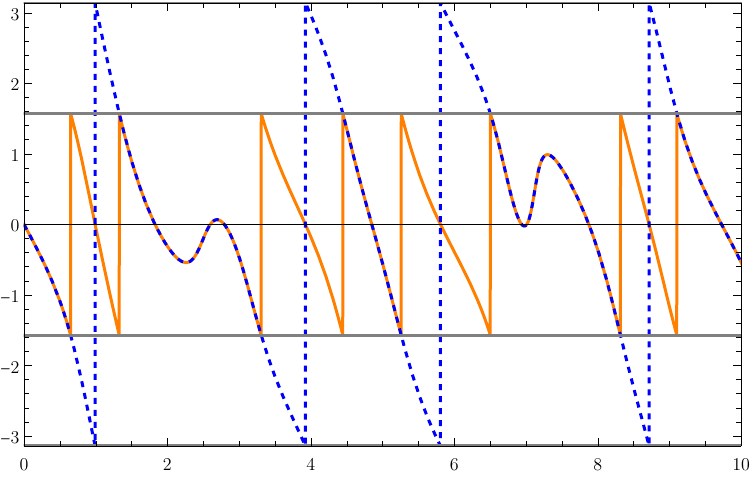}};
    \draw (4.65,3) node{(b) fermions, $4$ degrees of freedom};
    \draw (4.72,-2.6) node[font=\footnotesize]{$t$};
    \draw (0.2,0.2) node[font=\footnotesize,rotate=90]{$\arg{\braket{J|e^{\widehat{K}}|J}}$};
    \end{tikzpicture}
    \ccaption{Time evolution of complex phase}{We plot the complex phase $\arg{\braket{J|e^{\widehat{K}}|J}}\in[-\pi,\pi]$ as a function of $t$ for randomly generated quadratic Hamiltonians $\hat{H}=\ii \widehat{K}$ in (a) bosonic and (b) fermionic systems. We compare our closed formulas for bosons~\eqref{eq:main-bosons} and fermions~\eqref{eq:main-fermions} as blue, dashed line with the naive formula based on the square root of~\eqref{eq:exp-squared} as orange line. The naive formula misses the branch cuts and only gives values in the interval $[-\tfrac{\pi}{2},\tfrac{\pi}{2}]$ based on our square root definition.}
    \label{fig:example-trajectories}
\end{figure*}

\subsubsection{Fermionic systems}
We now consider the fermionic case. Here, we do not need to worry about decomposing $K$ into different commuting pieces, but there exist choices of $K$, such that $\braket{J|e^{\widehat{K}}|J}=0$, in which case its complex phase would not be well-defined. As discussed in section~\ref{sec:construction-double-cover}, in this case we could just move to a different reference complex structure.

\begin{result3b}[Fermions]\label{prop:expectation}
Given $K\in\mathfrak{so}(2N,\mathbb{R})$ and a pure Gaussian state $\ket{J}$, we can always find another complex structure $\tilde{J}$, such that $[K,\tilde{J}]=0$ and $\det(\id+\tilde{J}J)\neq 0$, from which we can compute
\begin{align}\label{eq:main-fermions}
    \arg\braket{J|e^{\widehat{K}}|J}=-\tfrac{1}{4}\Tr(K\tilde{J})-\tfrac{1}{2}\eta(T^{-1},e^{K})\,,
\end{align}
where $T=\sqrt{-\tilde{J}J}$.
\end{result3b}
\begin{proof}
As $K\in\mathfrak{so}(2N,\mathbb{R})$ implies that $K$ is an anti-symmetric matrix with respect to $G$, we can always use a special orthogonal transformation to bring $K$ into a quasi-diagonal real form and define a complex structure $\tilde{J}$ in the same orthonormal basis
\begin{align}
    K\equiv\bigoplus^{N}_{i=1}\begin{pmatrix}
        0 & \omega_i\\
        -\omega_i
    \end{pmatrix}\,,\quad \tilde{J}\equiv\bigoplus^{N}_{i=1}\begin{pmatrix}
        0 & \sigma_i\\
        -\sigma_i
    \end{pmatrix}\,,
\end{align}
such that $\omega_i\geq 0$ and we have the freedom of choosing each $\sigma_i=\pm 1$ individually. Clearly, we have $[K,\tilde{J}]=0$. Recall from our discussion around~\eqref{eq:cartan-fermion-defined} that $\Delta=-\tilde{J}J$ can have $-1$ eigenvalue pairs or quadruples, which we will be able to avoid by appropriately choosing $\sigma_i$. If there are eigenvalues $-1$, $J$ will have $2$-by-$2$ and $4$-by-$4$ blocks that only differ from the respective blocks in $\tilde{J}$ by an overall minus sign. By flipping the sign of these blocks in $\tilde{J}$ using our freedom to choose $\sigma_i$, we can ensure that $-\tilde{J}J$ does not have eigenvalues $-1$, such that $\det(\id+\tilde{J}J)\neq 0$. With this choice, there exists a unique $T=\sqrt{-\tilde{J}J}$, such that $\tilde{J}=TJT^{-1}$.\\
If we consider $\mathcal{U}(e^{K},\psi_{e^{K}})=e^{\widehat{K}}$, relation~\eqref{eq:cond2} states that $\arg{\braket{J|e^{\widehat{K}}|J}}=\arg{\psi_{e^K}}$. We can then use proposition~\ref{prop:switch-J} to relate $\psi_{e^K}$, computed with respect to $J$, with $\tilde{\psi}_{e^{K}}$, computed with respect to $\tilde{J}$, as
\begin{align}
\tilde{\psi}_{e^{K}}&=\psi_{e^K}e^{\ii[\eta(T^{-1},e^K)+\eta(T^{-1}e^K,T)]/2}\,.\label{eq:tildepsi-psi}
\end{align}
Lemma~\ref{lem:cocycle-properties} and setting $u=T^{-1}e^{-K}T$ allows us to find
\begin{align}
\begin{split}\label{eq:remove-one-eta}
    \eta(T^{-1}e^K,T)&=\eta(T^{-1}e^K,Tu)\\
    &=\eta(T^{-1}e^K,e^{-K}T)\\
    &=\eta(M,M^{-1})=0\,,
\end{split}
\end{align}
where we used~\eqref{eq:eta1} with $u\in\mathrm{U}(N)$ due to $[u,J]=T^{-1}[e^{-K},\tilde{J}]T=0$ in the first line and~\eqref{eq:eta3} in the last.\\
Relation $[K,\tilde{J}]=0$ implies that $\ket{J}$ is an eigenvector of $\widehat{K}$ and we can thus compute $\tilde{\psi}_{e^K}$ directly as
\begin{align}
    \tilde{\psi}_{e^K}=\braket{\tilde{J}|e^{\widehat{K}}|\tilde{J}}=e^{-\frac{\ii}{4}\Tr(K\tilde{J})}\,.
\end{align}
Plugging this and~\eqref{eq:remove-one-eta} into~\eqref{eq:tildepsi-psi}, solving it for $\psi_{e^K}$ and taking its argument then yields~\eqref{eq:main-fermions}.
\end{proof}

\subsection{Generalized Wick theorem}
\label{sec:generalized-wick}
So far we have discussed in depth how to evaluate expressions of the form $\braket{J|e^{\widehat{K}}|J}$ or $\braket{J|\mathcal{U}(M,\psi)|J}$. Here we will now consider an extension of these expressions, namely
\begin{equation}
    \braket{J|\, \hat{\xi}^{a_1}\cdots\hat{\xi}^{a_d} \, \mathcal{U}(M,\psi) |J}\,. \label{eq:generalized-wick}
\end{equation}
This can be seen as a generalization of the well-known Wick theorem~\cite{wick_evaluation_1950}, which deals with the expectation of monomials like $\hat{\xi}^{a_1}\cdots\hat{\xi}^{a_d}$ on a single Gaussian state, expressing them just as functions of the covariance matrix. Here we deal with the overlap of such monomials between two different Gaussian states related by the transformation $\mathcal{U}(M,\psi)$.

This generalized Wick expectation can be evaluated as follows.
\begin{result4}
    The expression $\braket{J|\, \hat{\xi}^{a_1}\cdots\hat{\xi}^{a_d} \, \mathcal{U}(M,\psi) |J}$ vanishes for odd $d$, while if $d$ is even it evaluates to 
    \begin{equation}
        \braket{J| \mathcal{U}(M,\psi) |J}\sum_\pi  \frac{|\pi|}{d!}\widetilde{C}^{\pi(a_1) \pi(a_2)} \cdots \widetilde{C}^{\pi(a_{d-1}) \pi(a_{d})} \,, \label{eq:generalized-wick-result}
    \end{equation}
    where the sum runs over all permutations with $\pi(2i-1)<\pi(2i)$ and $|\pi|$ is $1$ for bosons and equal to the sign of the permutation for fermions. The effective covariance matrix $\widetilde{C}$ is given by
    \begin{align}
        \widetilde{C}^{ab}&= R^{a}{}_{c} R^{b}{}_{d} \braket{J|\hat{\xi}^{c} \hat{\xi}^{d}|J} \\
        &=\frac{1}{2} R^{a}{}_{c} (G^{cd}+\ii \Omega^{cd})  R^{b}{}_{d} \,,
    \end{align}
    with $R=1+\frac{1}{4}(1-\ii J)\, \mathrm{tanh}K_+ \,(1+\ii J)$, given that $M=e^{K_+}u$ is a valid Cartan decomposition of $M$, as discussed in section~\ref{sec:cartan}.
\end{result4}
\begin{proof}
To compute~\eqref{eq:generalized-wick} let us suppose that the group element $M$ can be decomposed according to the Cartan decomposition $M=Tu$, with $T=e^{K_+}$. 
This means that we can also decompose $\mathcal{U}(M,\psi)=\sigma \,e^{\widehat{K}_+} \,\mathcal{U}(u)$, where the sign factor $\sigma\in\{+1,-1\}$ arises as a consequence of dealing with a projective representation. It turns out that we do not need to evaluate $\sigma$ explicitly, as we will see that we can reabsorbe it later. Similarly it does not matter what phase is chosen for the unitary $\mathcal{U}(u)$, as long as it is consistent. Observing, as in section~\ref{sec:squared-expectation}, that $\ket{J}$ is an eigenstate of $\mathcal{U}(u)$, we have
\begin{equation}
    \mathcal{U}(M,\psi)\ket{J}=\sigma \,e^{\widehat{K}_+}\ket{J}\braket{J|\mathcal{U}(u)|J}\,. \label{eq:cartan-for-wick}
\end{equation}
 
Using the results of reference~\cite{hackl2021bosonic} (Section 3.1.3) one can further decompose $e^{\widehat{K}_+}$ in terms of the operators $\hat{\xi}_\pm^a$ and in particular show that
\begin{equation}
    e^{\widehat{K}_+}\ket{J}= \mathcal{U}_+ \ket{J}\braket{J|e^{\widehat{K}_+}|J}\,,
\end{equation}
where we have defined
\begin{align}
\mathcal{U}_+=\begin{cases}
    e^{-\frac{\ii}{2}\omega_{ac}L^c{}_b \hat{\xi}_+^a\hat{\xi}_+^b}   &\textbf{(bosons)}\\
    e^{\frac{1}{2}g_{ac}L^c{}_b \hat{\xi}_+^a\hat{\xi}_+^b}   &\textbf{(fermions)}
\end{cases}\,.
\end{align}
Here, we have that $L=\tanh K_+$ and $\hat{\xi}_+^a=\tfrac{1}{2}(\delta^a{}_b+\ii J^a{}_b)\hat{\xi}^b$, such that $\bra{J} \hat{\xi}^a_+=0$ according to~\eqref{eq:annihilation-op}.
Combining these results we have
\begin{align}
    &\braket{J| \,\hat{\xi}^{a_1}\cdots\hat{\xi}^{a_d} \, \mathcal{U}(M,\psi) |J}  \nonumber \\
    &\hspace{15mm}=\braket{J| \,\hat{\xi}^{a_1}\cdots\hat{\xi}^{a_d} \, \mathcal{U}_+ |J} \times \nonumber \\
    &\hspace{30mm} \times \sigma \braket{J|e^{\widehat{K}_+}|J}\braket{J|U(u)|J}  \\
    &\hspace{15mm}=\braket{J| \,\hat{\xi}^{a_1}\cdots\hat{\xi}^{a_d} \, \mathcal{U}_+ |J} \braket{J|\mathcal{U}(M,\psi)|J}\,,
\end{align}
where in the last step we have used~\eqref{eq:cartan-for-wick} again.

We see therefore that in order to evaluate generalized Wick expectation values it is necessary to compute $\braket{J|\mathcal{U}(M,\psi)|J}$, including its complex phase, to which the previous results of this paper are dedicated. It is also necessary to compute $\braket{J|\hat{\xi}^{a_1}\cdots\hat{\xi}^{a_d} \mathcal{U}_+ |J}$, but this can be achieved easily.

Notice, indeed, that, although $\mathcal{U}_+$ is not unitary, it is still the exponential of a quadratic combination of linear operators $\hat{\xi}$. Therefore its action can be evaluated in closed form using the canonical (anti)commutation relations:
\begin{equation}
    U^{-1}_+ \, \hat{\xi}^a \, \mathcal{U}_+ = R^a{}_b \, \hat{\xi}^b\,,
\end{equation}
where $R=\id+P_- L P_+$ with $P_\pm=\tfrac{1}{2}(\id\pm\ii J)$. Using this to commute $\mathcal{U}_+$ to the left and observing that $\bra{J}\mathcal{U}_+=\bra{J}$, we have 
\begin{equation}
    \braket{J|\hat{\xi}^{a_1}\cdots\hat{\xi}^{a_d} \mathcal{U}_+ |J} = R^{a_1}{}_{b_1} \cdots R^{a_d}{}_{b_d} \braket{J|\hat{\xi}^{b_1}\cdots\hat{\xi}^{b_d}|J}
\end{equation}
which can be evaluated using the standard Wick theorem, leading to the result~\eqref{eq:generalized-wick-result}.
\end{proof}

This result can also be understood in the following way. The generalized Wick expectation values~\eqref{eq:generalized-wick} can be also rewritten as
\begin{equation}
    \braket{J|\mathcal{U}(M,\psi) |J} \tr \left[\frac{ \mathcal{U}(M,\psi)\ket{J}\bra{J}}{\braket{J|\mathcal{U}(M,\psi) |J}}\,\hat{\xi}^{a_1}\cdots\hat{\xi}^{a_{d}} \right]\,.
    \label{eq:phase-space-rep}
\end{equation}
From the perspective of the phase space representation of Quantum Mechanics this trace corresponds to computing the $d$-th moment of the characteristic function of 
\begin{equation}
    \hat{O}=\frac{ \mathcal{U}(M,\psi)\ket{J}\bra{J}}{\braket{J|\mathcal{U}(M,\psi) |J}}\,.
\end{equation}
Although $\hat{O}$ is not the density operator of a well-defined Gaussian state (it is not Hermitian), its characteristic function is still a normalized Gaussian distribution. Indeed, it is the product of two objects that have Gaussian characteristic functions: $\ket{J}\bra{J}$, because it is a Gaussian state, and $\mathcal{U}(M,\psi)$, because it is a Gaussian operator. The characteristic function of a product is the convolution of the characteristic function of the factors, so it is still Gaussian. The denominator of the fraction ensures the normalization. 

The moments of a normalized Gaussian distribution follow Wick's theorem (or Isserlis' theorem, as it is normally called in the context of probability distributions), so~\eqref{eq:phase-space-rep} must have the form~\eqref{eq:generalized-wick-result}. To compute the value of two-point functions $\widetilde{C}^{ab}$ one could also compute the phase space Gaussian integrals to find the characteristic function of $\hat{O}$. In this paper we avoid this step by exploiting group theory arguments instead.

\section{Case studies: Bosonic and fermionic systems}\label{sec:CaseStudies}
We will illustrate our results for the simplest non-trivial systems, namely one bosonic mode and two fermionic modes. The system of a single fermionic mode is trivial, as the manifold of Gaussian states only consists of two points, often denoted by $\ket{0}$ and $\ket{1}$ with $\braket{0|1}=0$.

\subsection{One bosonic mode}
We consider a single bosonic degree of freedom with Hermitian operators $\hat{\xi}^a\equiv(\hat{q},\hat{p})$, such that $[\hat{\xi}^a,\hat{\xi}^b]=\ii\Omega^{ab}$ with $\Omega$ taking the standard form~\eqref{eq:Omega-G-standard-form}, that act on the infinite dimensional Hilbert space $\mathcal{H}$, a bosonic Fock space.

\textbf{Lie group.} We consider the symplectic group $\mathrm{Sp}(2,\mathbb{R})$, which is identical to $\mathrm{SL}(2,\mathbb{R})$. An elegant parametrization of this group is given by
\begin{align}
    M\!=\!\left(
    \begin{array}{cc}
    \cos{\tau}\,\ch{\rho}+\sin{\theta}\,\sh{\rho} & \sin{\tau}\,\ch{\rho}-\cos{\theta}\,\sh{\rho}\\
    -\sin{\tau}\,\ch{\rho}-\cos{\theta}\,\sh{\rho} & \cos{\tau}\,\ch{\rho}-\sin{\theta}\,\sh{\rho}
    \end{array}\right),
\end{align}
where $\sh=\sinh$ and $\ch=\cosh$. It is easy to verify that $\det M=1$, but it takes some more effort to show that every such matrix can be written in this form. Furthermore, we can verify that $(\rho,\theta)$ are polar coordinates describing a plane, while $\tau$ is a periodic coordinate describing a circle. Topologically, we therefore have $\mathrm{Sp}(2,\mathbb{R})\cong\mathbb{R}^2\times S^1$, which we visualize in figure~\ref{fig:case-study-bosonic}. Given a reference complex structure
\begin{align}
    J\equiv\begin{pmatrix}
    0 & \id\\
    -\id & 0
    \end{pmatrix}\,,
\end{align}
we have the stabilizer of $J$ given by
\begin{align}
    \mathrm{U}(1)=\{u\in\mathrm{Sp}(2,\mathbb{R})\,|\,uJu^{-1}=J\}\,.
\end{align}
This subgroup is parametrized by $M$ for $\rho=0$, such that $\theta$ becomes irrelevant, and $\tau\in[0,2\pi]$. We can use the identification from~\eqref{eq:ComplexDecomposition} to parametrize $u$ as the $1$-by-$1$ matrix $\overline{u}=e^{\ii \tau}$ equal to its determinant.

\begin{figure*}[t]
    \centering
    \begin{tikzpicture}[scale=1.1]
    \tikzset{->-/.style={decoration={
      markings,
      mark=at position #1 with {\arrow{>}}},postaction={decorate}}}
    \coordinate (A) at (0,0,-1);
    \coordinate (B) at (0,0,-6);
    \coordinate (C) at (6,0,-6);
    \coordinate (D) at (6,0,-1);
    \coordinate (S) at (0,-2.8,0);
    \coordinate (E) at ($(A)+(S)$);
    \coordinate (F) at ($(B)+(S)$);
    \coordinate (G) at ($(C)+(S)$);
    \coordinate (H) at ($(D)+(S)$);
    \coordinate (S2) at (0,-1.4,0);
    \coordinate (E2) at ($(A)+(S2)$);
    \coordinate (F2) at ($(B)+(S2)$);
    \coordinate (G2) at ($(C)+(S2)$);
    \coordinate (H2) at ($(D)+(S2)$);
    \coordinate (S3) at (0,-5,0);
    \coordinate (E3) at ($(A)+(S3)$);
    \coordinate (F3) at ($(B)+(S3)$);
    \coordinate (G3) at ($(C)+(S3)$);
    \coordinate (H3) at ($(D)+(S3)$);
    \coordinate (dd) at (.05,0,0);
    \coordinate (IB) at (3,-1.4,0);
    \coordinate (S4) at (0,-1.8,0);
     \draw[densely dashed] (C) -- (G);
    \draw[densely dashed,blue,opacity=0.2] (E) -- (F) -- (G);
    \draw[opacity=0.2] (B) -- (F);
    \draw[densely dashed,blue] (E) -- (H) -- (G);
    \draw[densely dashed] (C) node[right]{$\mathcal{G}$} -- (G);

    \shadedraw[densely dashed] (E3) -- (F3) -- (G3) node[right]{$\mathcal{M}$} -- (H3) -- cycle;
    
    \draw[thick] ($(A)-(dd)$) node[left]{$\pi$} -- ($(A)+(dd)$);
    \draw[thick] ($(A)-(0,1.4,0)-(dd)$) node[left]{$0$} -- ($(A)-(0,1.4,0)+(dd)$);
    \draw[thick] ($(E)-(dd)$) node[left]{$-\pi$} -- ($(E)+(dd)$);
    \draw[thick] (3,-2.8,-3) -- (3,-2.17,-3);
    \draw[thick,opacity=.4] (3,-1.4,-3) -- (3,-2.17,-3);

    \begin{scope}[yshift=-1.4cm]
    \fill[red,opacity=.2] plot[domain=0:3,samples=91,smooth] ({3+\x},{-1.4+.01*(-atan(sinh(\x))+180)*1.4/1.8},{-3}) -- (6,0,-3);
    \fill[blue,opacity=.2] plot[domain=3:0,samples=91,smooth] ({3+\x},{-1.4+.01*atan(sinh(\x))*1.4/1.8},{-3}) -- (6,-1.4,-3);
    \fill[green,opacity=.2] plot[domain=0:3,samples=91,smooth] ({3+\x},{-1.4+.01*(-atan(sinh(\x))+180)*1.4/1.8},{-3}) -- plot[domain=3:0,samples=91,smooth] ({3+\x},{-1.4+.01*atan(sinh(\x))*1.4/1.8},{-3});
    \end{scope}
    
    \fill[gray,opacity=.5] (E2) -- (F2) -- (G2) -- (H2) -- cycle;

    \fill[white] (3,-1.4,-3) -- (6,-1.4,-3) -- (6,0,-3) -- (3,0,-3) -- cycle;
    \fill[blue,opacity=.2] plot[domain=0:3,samples=91,smooth] ({3+\x},{-1.4+.01*(-atan(sinh(\x))+180)*1.4/1.8},{-3}) -- (6,0,-3);
    \fill[red,opacity=.2] plot[domain=3:0,samples=91,smooth] ({3+\x},{-1.4+.01*atan(sinh(\x))*1.4/1.8},{-3}) -- (6,-1.4,-3);
    \fill[green,opacity=.2] plot[domain=0:3,samples=91,smooth] ({3+\x},{-1.4+.01*(-atan(sinh(\x))+180)*1.4/1.8},{-3}) -- plot[domain=3:0,samples=91,smooth] ({3+\x},{-1.4+.01*atan(sinh(\x))*1.4/1.8},{-3});

    \draw[thick] (3,-1.4,-3) -- (3,0,-3) node[above]{$\mathrm{U}(1)$};
    \draw[densely dashed] (A) -- (E) (D) -- (H);
    
    \draw[red,->] (3,-1.4,-3) -- (3.5,-1.4,-3) node[right]{$X$};
    \draw[red,->] (3,-1.4,-3) -- (3,-.9,-3) node[right]{$Z$};
    \draw[red,->] (3,-1.4,-3) -- (3,-1.4,-2.3) node[below]{$Y$};

    \fill (3,-1.4,-3) node[left]{$\id$} circle (1pt);
    \draw (-.1,-.3,0) node[right]{$\tau$};
    \draw plot[domain=0:135,samples=91,smooth]
    ({3+.4*cos(\x)},{-4.6},{-3+.4*sin(\x)});
    \fill ($(3,-2.8,-3)+(S4)$) node[above]{$J$} circle (1pt);
   \draw[->] ($(3,-2.8,-3)+(S4)$) -- node[below,xshift=-2mm]{$\theta$} ($(4,-2.8,-3)+(S4)$) node[above,xshift=2mm]{$x=\rho \cos\theta$};
   \draw[->] ($(3,-2.8,-3)+(S4)$) -- ($(3,-2.8,-1.7)+(S4)$) node[right,xshift=1mm]{$y=\rho \sin\theta$};
   \draw ($(3,-2.8,-3)+(S4)$) -- node[above]{$\rho$} ($(2.2,-2.8,-2.2)+(S4)$);
   \fill ($(2.2,-2.8,-2.2)+(S4)$) node[below]{$J_{M(\rho,\theta,\tau)}$} circle (1pt);
   \draw (2.2,-2.8,-2.2) -- (2.2,-1.85,-2.2) (2.2,-1.4,-2.2) -- (2.2,0,-2.2);
   \draw[opacity=.4]  (2.2,-1.85,-2.2) -- (2.2,-1.4,-2.2);
   \draw[blue,densely dashed] (A) -- (B) -- (C) -- (D) -- cycle;
   \fill[blue] (3,0,-3) circle (1pt) (3,-2.8,-3) circle (1pt);
   \fill (2.2,-1.2,-2.2) node[left]{$M(\rho,\theta,\tau)$} circle (1pt);

    \draw[dotted] (3,0,-3) -- (9.6,1.65)  (3,-2.8,-3) -- (9.6,-1.95) ;

    \begin{scope}[xshift=9.6cm,yshift=-.15cm]
        \draw[very thick,blue] plot[domain=0:6,samples=91,smooth] ({\x},{.01*(-atan(sinh(\x))+180)});
        \draw[very thick,blue] plot[domain=0:6,samples=91,smooth] ({\x},{.01*(atan(sinh(\x))-180)});
        \draw[very thick,red] plot[domain=0:6,samples=91,smooth] ({\x},{-.01*atan(sinh(\x))});
        \draw[very thick,red] plot[domain=0:6,samples=91,smooth] ({\x},{.01*atan(sinh(\x))});

        \fill[red,opacity=.2] plot[domain=0:6,samples=91,smooth] ({\x},{-.01*atan(sinh(\x))}) -- plot[domain=6:0,samples=91,smooth] ({\x},{.01*atan(sinh(\x))});

        \fill[blue,opacity=.2] plot[domain=0:6,samples=91,smooth] ({\x},{.01*(-atan(sinh(\x))+180)}) -- (6,1.8);
        \fill[blue,opacity=.2] plot[domain=0:6,samples=91,smooth] ({\x},{.01*(atan(sinh(\x))-180)}) -- (6,-1.8);

        \fill[green,opacity=.2] plot[domain=0:6,samples=91,smooth] ({\x},{.01*(-atan(sinh(\x))+180)}) -- plot[domain=6:0,samples=91,smooth] ({\x},{.01*atan(sinh(\x))});
        \fill[green,opacity=.2] plot[domain=0:6,samples=91,smooth] ({\x},{.01*(atan(sinh(\x))-180)}) -- plot[domain=6:0,samples=91,smooth] ({\x},{-.01*atan(sinh(\x))});

        \draw (.4,.9) node{(I)}  (3,.45) node{(II)} (3,1.35) node{(III)};
        \draw (5,-.25) node{$\rho$};
        
        \draw[thick,->] (0,-1.8) node[left]{$-\pi$} -- (0,2);
        \draw (0,1.8) node[left]{$\pi$};
        \draw[thick,->] (0,0) node[left]{$0$} -- (6.2,0);
        \draw (0,-1.8) rectangle (6,1.8);

    \end{scope}
    
    \end{tikzpicture}
    \ccaption{Bosonic fiber bundle for $\mathrm{Sp}(2,\mathbb{R})$}{We visualize the symplectic group $\mathrm{Sp}(2,\mathbb{R})$, which is topologically $\mathbb{R}^2\times S^1$, so that we need to identify the upper and lower plane. A general group element $M$ is parametrized by $(\rho,\theta,\tau)$, where $\rho$ and $\theta$ are polar coordinates of a horizontal plane and $\tau$ is a $2\pi$-periodic vertical coordinate. We also recognize the Lie algebra $\mathfrak{sp}(2,\mathbb{R})=\mathrm{span}(X,Y,Z)$ as tangent space at the identity. We further illustrate the different regions (I), (II) and (III) characterized by the relations in~\eqref{eq:region-(I)} to~\eqref{eq:region-(III)}. We show a planar section of these regions inside the group and an enlarged version.}
    \label{fig:case-study-bosonic}
\end{figure*}

\textbf{Lie algebra.} The Lie algebra $\mathrm{sl}(2,\mathbb{R})$ is generated by the three matrices
\begin{align}
    X=\begin{pmatrix}
    0 & 1\\
    1 & 0
    \end{pmatrix}\,,\,\,Y=\begin{pmatrix}
    1 & 0\\
    0 & -1
    \end{pmatrix}\,,\,\,Z=\begin{pmatrix}
    0 & 1\\
    -1 & 0
    \end{pmatrix}\,.
\end{align}
While $X$ and $Y$ are tangential to the horizontal plane at $M=\id$, we note that $Z$ points upwards along the $\mathrm{U}(1)$ circle, such that we have
\begin{align}
    \mathfrak{u}(1)=\mathrm{span}(Z)\quad\text{and}\quad\mathfrak{u}_\perp(1)=\mathrm{span}(X,Y)\,.
\end{align}
The Killing form\footnote{The Killing form is constructed from the adjoint representation of a Lie algebra, but can often also be expressed in the fundamental representation. In the case of $A,B\in\mathfrak{sp}(2,\mathbb{R})$, the Killing form is proportional to $\mathrm{Tr}(AB)$, where $A,B$ are matrices in the fundamental representation.} induces a natural invariant semi-definite metric given by $\langle A,B\rangle=\tfrac{1}{2}\mathrm{Tr}(AB)$ with norm
\begin{align}
    \lVert aX+bY+cZ\rVert^2=a^2 + b^2 - c^2\,.
\end{align}

\textbf{Exponential map.} We now discuss which elements $M$ of the group can be reached by exponentiation of a Lie algebra element $K$, such that $M=e^K$. We find the following three sectors, where we use a language inspired from special relativity (light cone, past, future), while it does not have an intrinsic physical meaning here. The group is rotational symmetric with respect to the angle $\theta$ (corresponding to rotations in the $X$-$Y$-plane in the Lie algebra) implies that we can describe the different regions purely in terms of $\rho$ and $\tau$, as illustrated in figure~\ref{fig:case-study-bosonic}:
\begin{itemize}
    \item[(I)] \textbf{Future/Past.} All points that are reachable by $e^K$ with $\lVert K\rVert<0$. This is equivalent to
    \begin{align}
    \begin{split}\label{eq:region-(I)}
        \tan^{-1}{\sinh{\rho}}&<\tau<-\tan^{-1}{\sinh{\rho}}+\pi\,,\\
        -\pi+\tan^{-1}{\sinh{\rho}}&<\tau<-\tan^{-1}{\sinh{\rho}}\,.
    \end{split}
    \end{align}
    Note that there are closed timelike curves and all of these points can be reached with both, future and past pointing curves. We have $\mathrm{spec}(M)=(\lambda,\overline{\lambda})$ with complex $\lambda$ satisfying $|\lambda|=1$.
    \item[(II)] \textbf{Reachable space.} All points that are reachable by $e^K$ with $\lVert K\rVert>0$. This is equivalent to
    \begin{align}
        -\tan^{-1}{\sinh{\rho}}<\tau<\tan^{-1}{\sinh{\rho}}\,.\label{eq:region-(II)}
    \end{align}
    We have $\mathrm{spec}(M)=(\lambda,1/\lambda)$ with real $\lambda>1$. Its boundary consists of group elements $M=e^{K}=(\id+K)$ with $\lVert K\rVert=0$, which are evidently all reachable by exponentiating a single Lie algebra element.
    \item[(III)] \textbf{Unreachable space.} This is equivalent to
    \begin{align}
        \tan^{-1}{\sinh{\rho}}<\tau<\pi-\tan^{-1}{\sinh{\rho}}\,.\label{eq:region-(III)}
    \end{align}
    These points cannot be reached by a single exponential $e^K$. We have $\mathrm{spec}(M)=(-\lambda,-1/\lambda)$ with real $\lambda>1$. Its boundary consists of group elements $M=-e^{K}=-(\id+K)$ with $\lVert K\rVert=0$, which cannot be reached except for the single point $M=-\id$.
\end{itemize}

\textbf{Bosonic fiber bundle.} Given an arbitrary group element $M$, we can compute its Cartan decomposition $M=Tu$, where $T=\sqrt{-MJM^{-1}J}$ with $T\in\exp(\mathfrak{u}_\perp(1))$. This set corresponds exactly to the group elements $M(\rho,\theta,\tau=0)$. We see explicitly from figure~\ref{fig:case-study-bosonic} that $\mathrm{Sp}(2,\mathbb{R})\simeq \mathbb{R}^2\times \mathrm{U}(1)$, where $(\rho,\theta)$ parametrize the plane $\mathbb{R}^2$, while $\tau$ is the circular coordinate for $\mathrm{U}(1)$. In particular, we see that the fiber bundle is trivial, as the base space $\mathbb{R}^2$ is contractible. The fundamental group $\pi_1(\mathrm{Sp}(2,\mathbb{R}))=\mathbb{Z}$ is purely due to the compact direction along $\mathrm{U}(1)$.

\textbf{Representation theory.} The symplectic Lie algebra $\mathfrak{sp}(2,\mathbb{R})$ is naturally represented by linear combinations of the following three operators:
\begin{align}
\begin{split}\label{eq:case-study-XYZ-operators}
    \widehat{X}&=-\ii\,\frac{\hat{p}^2-\hat{q}^2}{2}=\frac{\ii}{2}(\hat{a}^{\dagger 2}-\hat{a}^2)\,,\\
    \widehat{Y}&=-\ii\,\frac{\hat{p}\hat{q}+\hat{q}\hat{p}}{2}=-\frac{\hat{a}^{\dagger 2}+\hat{a}^2}{2}\,,\\
    \widehat{Z}&=-\ii\,\frac{\hat{p}^2+\hat{q}^2}{2}=-\ii(\hat{a}^\dagger\hat{a}+\tfrac{1}{2})\,.
\end{split}
\end{align}
A general Lie algebra element $K=a X+bY+cZ$ is represented by the quadratic operator $\widehat{K}=a\widehat{X}+b\widehat{Y}+c\widehat{Z}$.

\textbf{Double cover.} When exponentiating $\widehat{Z}$, we observe the behavior discussed in section~\ref{sec:GaussianUnitaries}. While we have $e^{2\pi Z}=\id$, we have $e^{2\pi \widehat{Z}}\neq \id$. This follows from $\widehat{Z}=-\ii(\hat{n}+\frac{1}{2})$ using the standard number operator $\hat{n}=\hat{a}^\dagger\hat{a}$ with non-negative integer spectrum. We therefore see that $\widehat{Z}$ generates $\widetilde{\mathrm{U}}(1)$ which is the double cover $\mathrm{U}(1)$, as $e^{t\widehat{Z}}$ is $4\pi$-periodic, while $e^Z$ is only $2\pi$-periodic. While $\widetilde{\mathrm{U}}(1)$ is itself isomorphic to $\mathrm{U}(1)$, the resulting full group $\widetilde{\mathcal{G}}$ will be double cover $\mathrm{Mp}(2,\mathbb{R})$.

For every symplectic group element $M$, there exist two distinct unitary transformations $\pm \mathcal{U}$ satisfying
\begin{align}
    \mathcal{U}^\dagger\hat{\xi}^a\mathcal{U}=M^a{}_b\hat{\xi}^b\,.
\end{align}
In particular, we find the following representations:
\begin{align}
    \mathcal{U}(\id)=\pm\id\,,\quad \mathcal{U}(-\id)=\pm\ii e^{-\ii\pi \hat{n}}\,.
\end{align}

\textbf{Cocycle function.} Given two symplectic transformations $M_i(\rho_i,\theta_i,\tau_i)$ with relative complex structures
\begin{align}
    \Delta_{M_i^{-1}}=\begin{pmatrix}
    \cosh{\rho_i}+\sin{\theta_i}\sinh{\rho_i} & \cos{\theta_i}\sinh{\rho_i}\\
    \cos{\theta_i}\sinh{\rho_i} & \cosh{\rho_i}-\sin{\theta_i}\sinh{\rho_i}
    \end{pmatrix}\,,
\end{align}
we find the single complex eigenvalue
\begin{align}
    \overline{\id-Z_{M_1}Z_{M_2^{-1}}}=1+e^{\ii(\theta_1-\theta_2)}\tanh{\tfrac{\rho_1}{2}}\tanh{\tfrac{\rho_2}{2}}\,,
\end{align}
which has a positive real part and yields
\begin{align}
    \eta=\arg(1+e^{\ii(\theta_1-\theta_2)}\tanh{\tfrac{\rho_1}{2}}\tanh{\tfrac{\rho_2}{2}})\,.
\end{align}

\textbf{Complex phase calculation.} Without loss of generality, we can consider the reference complex structure $J$ in its standard form from~\eqref{eq:J-standard-form} and a general Lie algebra element $K$ to compute $\braket{J|e^{\widehat{K}}|J}$. The action of our $\mathrm{U}(1)$ subgroup onto group and Lie algebra corresponds to a rotation around the $\tau$- and $Z$-axis (on group vs. algebra level, respectively), which is a symmetry. Therefore, we can consider without loss of generality a Lie algebra element $K=aX+cZ$. We can use formula~\eqref{eq:exp-squared} to find
\begin{align}
    \braket{J|e^{\widehat{K}}|J}^2=\cosh{\sqrt{a^2-c^2}}+\frac{\ii c}{\sqrt{a^2-c^2}}\sinh{\sqrt{a^2-c^2}}\,,
\end{align}
from which we immediately find the absolute value
\begin{align}
    |\braket{J|e^{\widehat{K}}|J}|^2=\frac{2(a^2-c^2)}{a^2-2c^2+a^2\cosh(2\sqrt{a^2-c^2})}\,.\label{eq:boson-modulus}
\end{align}
For the complex phase, we need to distinguish the cases $a^2\leq c^2$, $a^2=c^2$ and $a^2\geq c^2$ to find
\begin{align}
    \arg\braket{J|e^{\widehat{K}}|J}\!=\!\begin{cases}
    -\frac{1}{2}\arctan\!\frac{(c-R)\cos(R)\sin(R)}{c+(R-c)\cos(R)}\!-\!\frac{R}{2}& a^2\leq c^2\\
    -\frac{1}{2}\arctan(c) & a^2=c^2\\
    -\frac{1}{2}\arctan(\tfrac{c \tanh(R)}{R}) & a^2\geq c^2    
    \end{cases}\label{eq:boson-phase}
\end{align}
with $R=\sqrt{|a^2-c^2|}$. We show both, complex phase (color) and modulus (brightness) in figure~\ref{fig:bosonic-<J|U|J>}. In particular, we see the light cone structure explicitly, where evolution in the direction of $K$ with $\lVert K\rVert>0$ (``spacelike'') yields exponential suppression of $|\braket{J|e^{\widehat{K}}|J}|$, while for $\lVert K\rVert<0$ (``timelike''), we find oscillations with relatively little suppression.

\begin{figure}[t]
    \centering
    \begin{tikzpicture}
    \draw (0,0) node[inner sep=0pt]{\includegraphics[width=\linewidth]{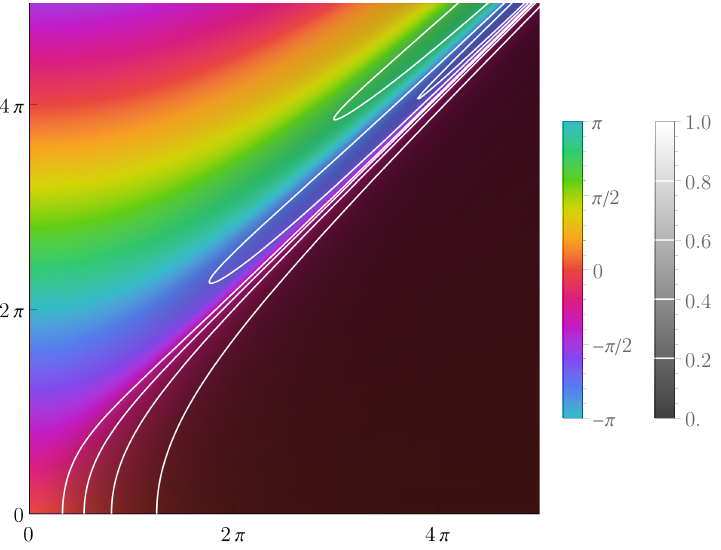}};
    \draw (-4.2,.7) node[font=\footnotesize]{$c$};
    \draw (-.3,-3.1) node[font=\footnotesize]{$a$};
    \draw (2.55,2.05) -- (3.25,2.6) node[font=\footnotesize,above]{$\arg(\braket{J|e^{\widehat{K}}|J})$};
    \draw (3.65,-1.9) -- (3.25,-2.3) node[font=\footnotesize,below]{$|\braket{J|e^{\widehat{K}}|J}|$};
    \end{tikzpicture}
    \ccaption{Bosonic $\braket{J|e^{\widehat{K}}|J}$ for $\widehat{K}=a\widehat{X}+c\widehat{Z}$ from~\eqref{eq:case-study-XYZ-operators}}{We show $\braket{J|e^{\widehat{K}}|J}$, where we indicate its complex phase from~\eqref{eq:boson-phase} by color and its modulus from~\eqref{eq:boson-modulus} by brightness. White height lines represent the respective modulus as indicated in the brightness legend.}
    \label{fig:bosonic-<J|U|J>}
\end{figure}

\subsection{Two fermionic modes}\label{sec:case-two-fermionic-modes}
For fermions, a system consisting of a single mode is trivial as the manifold of Gaussian states consists only of two states, typically denoted by $\ket{0}$ and $\ket{1}$ in the fermionic Fock space. The space of quadratic Hamiltonians is (apart from an energy offeset) one-dimensional and given by $\hat{H}=\omega(\hat{n}-\frac{1}{2})$, which implies $\braket{0|e^{-\ii\hat{H}}|0}=e^{\ii\omega/2}$ and $\braket{1|e^{-\ii\hat{H}}|1}=e^{-\ii\omega/2}$. We will therefore consider the simplest non-trivial case of two fermionic modes. As we will see, the manifold of Gaussian states then consists of two disconnected components which are both topologically equivalent to $2$-spheres. The relevant Lie group will be $\mathrm{SO}(4,\mathbb{R})$. We are able to see the fermionic principal fiber bundle structure from section~\ref{sec:principal-fiber-bundles} explicitly.

\textbf{Lie group.} We represent an element of $\mathrm{SO}(4,\mathbb{R})$ as real $4$-by-$4$ matrix with $MM^\intercal=\id$, where we are in a basis with $G\equiv \id$. Given a reference complex structure
\begin{align}
    J\equiv\begin{pmatrix}
    0 & \id\\
    -\id & 0
    \end{pmatrix}\,,
\end{align}
we have the stabilizer of $J$ given by
\begin{align}
    \mathrm{U}(2)=\{u\in\mathrm{SO}(4,\mathbb{R})| uJu^{-1}=J\}\,.
\end{align}
This implies $u J=J u$, \ie $u$ must be complex linear with respect to $J$. This implies that $u$ must have the block structure
\begin{align}
    u=\begin{pmatrix}
    u_1 & u_2\\
    -u_2 & u_1
    \end{pmatrix}
\end{align}
and via the relation~\eqref{eq:ComplexDecomposition}, we have $\overline{u}=u_1+\ii u_2$. The orthogonality $u\in\mathrm{SO}(4,\mathbb{R})$ implies $u_1u_1^\intercal-u_2u_2^\intercal=\id$ and $u_1u_2^\intercal+u_2u_1^\intercal=0$, which together is equivalent to $\overline{u}\overline{u}^\dagger=\id$, thus demonstrating explicitly how this group is $\mathrm{U}(2)$.

\textbf{Lie algebra.} The Lie algebra $\mathfrak{so}(4,\mathbb{R})$ is $6$-dimensional and we decompose it into $\mathfrak{so}(4,\mathbb{R})=\mathfrak{u}(2)\oplus \mathfrak{u}_\perp(2)$. While $\mathfrak{u}(2)$ is relative straight-forward, as $\overline{K}_-\in \mathfrak{u}(2)$ can be easily described by anti-Hermitian $2$-by-$2$ matrices, it will be essential to characterize $\mathfrak{u}_\perp(2)$. Its elements $K_+$ are Lie algebra elements with $\{K_+,J\}=0$, which spans a $2$-dimensional space
\begin{align}
    K_+=\frac{\theta}{2}\left(
\begin{array}{cccc}
 0 & \cos (\phi ) & 0 & \sin (\phi ) \\
 -\cos (\phi ) & 0 & -\sin (\phi ) & 0 \\
 0 & \sin (\phi ) & 0 & -\cos (\phi ) \\
 -\sin (\phi ) & 0 & \cos (\phi ) & 0
\end{array}\right)\,,
\end{align}
where we should think of $(\theta,\phi)$ as polar coordinates with radius $\theta$ and angle $\phi$. The eigenvalues of $K_+$ are given by $\pm\ii\frac{\theta}{2}$ with multiplicity $2$ each. We can also compute the norm $\lVert K_+\rVert_{\infty} =\frac{\theta}{2}$ for $\theta\geq 0$. Therefore, we have
\begin{align}
    \mathcal{I}_{\mathfrak{u}_\perp(2)}=\{K_+(\theta,\phi)\,|\,0\leq\theta<\pi,0\leq \phi<2\pi\}\,,
\end{align}
meaning the inner part of $\mathfrak{u}_\perp(2)$ and thus also $\mathcal{I}_{\mathcal{M}}$ via the map $K_+\mapsto e^{2K_+}J$ are given by disks parametrized by $(\theta,\phi)$ where $\theta$ describes the radial coordinate.

\begin{figure*}[t]
    \centering
    \begin{tikzpicture}

    \draw[domain=-180:-90, smooth, variable=\x, purple,thick,opacity=.3]  plot ({3+2*sin(\x)}, {2.75-0.7*cos(\x)+1.25*\x/180});
    \draw[domain=-135:-90, smooth, variable=\x, blue,thick,opacity=.3]  plot ({3+2*sin(\x)}, {2.4375-0.7*cos(\x)+1.25*\x/180});
    
    \draw[domain=90:225, smooth, variable=\x, blue,thick,opacity=.3]  plot ({3+2*sin(\x)}, {2.4375-0.7*cos(\x)+1.25*\x/180});
    \draw[domain=90:180, smooth, variable=\x, purple,thick,opacity=.3]  plot ({3+2*sin(\x)}, {2.75-0.7*cos(\x)+1.25*\x/180});
    
    \draw[thick,orange] (3,4) ellipse (2cm and .7cm);
    \draw[thick,orange](3,1.5) ellipse (2cm and .7cm);
    \draw[thick,orange] (1,1.5) -- (1,4) node[black,font=\footnotesize,left]{$\mathcal{G}$} (5,1.5) -- (5,4) node[right,xshift=0mm,yshift=4mm]{$\mathcal{B}_{\mathcal{G}}\simeq\mathrm{U}(2)$};

    \draw[domain=-90:90, smooth, variable=\x, dgreen,thick,opacity=.3]  plot ({3+2*sin(\x)}, {2.75+0.7*cos(\x)});
    

    \draw[thick] (3,1.5) -- (3,4) node[font=\footnotesize,left]{$\mathrm{U}(2)$};
    \fill (3,2.75) node[font=\footnotesize,right]{$\id$} circle (1pt);
    
    \draw[domain=90:270, smooth, variable=\x, dgreen,thick]  plot ({3+2*sin(\x)}, {2.75+0.7*cos(\x)});
    
    \draw[domain=-90:90, smooth, variable=\x, purple,thick]  plot ({3+2*sin(\x)}, {2.75-0.7*cos(\x)+1.25*\x/180});

    \draw[domain=-90:90, smooth, variable=\x, blue,thick]  plot ({3+2*sin(\x)}, {2.4375-0.7*cos(\x)+1.25*\x/180});
    

    \fill[purple] (3.9,2.3) node[above,purple]{$T_0$}; circle (1pt);
    \fill[blue] (3.9,1.9) node[xshift=2mm,below,blue]{$T_0u_\phi$}; circle (1pt);
    

    \begin{scope}[xscale=.97]
    \draw[thick,orange] (6,1.5) rectangle (14,4);
    \draw[orange] (10,1.5) -- (10,4);
    \draw[thick,dgreen] (6,2.75) -- (14,2.75);
    \draw[thick,blue] (6,1.8125) -- (13,4) node[right,font=\footnotesize,yshift=3mm,xshift=-4mm]{
    } (13,1.5) -- (14,1.8125);
    \draw[thick,purple] (6,2.125) -- (12,4) node[left,font=\footnotesize,yshift=3mm,xshift=4mm]{
    } (12,1.5) -- (14,2.125);
    
    \draw[dgreen] (12,2.75) node[xshift=1mm,above,font=\footnotesize]{$\exp(\mathcal{B}_{\mathfrak{u}_\perp(2)})$};

    \draw[thick,red] (6,-.3) -- (14,-.3);
    \fill (6,-.3) node[below,font=\footnotesize]{$-\frac{\pi}{2}$} circle (1pt) (14,-.3) node[below,font=\footnotesize]{$\frac{3\pi}{2}$} circle (1pt);
    \fill (10,-.3) node[below,black,font=\footnotesize]{$\frac{\pi}{2}$} circle (1pt);
    \fill (8,-.3) node[below,black,font=\footnotesize]{$0$} circle (1pt) (12,-.3) node[below,black]{$\pi$} circle (1pt);
    
    \draw[thick,red] (6,-.3) -- (14,-.3);
    \fill (6,-.3) node[below,font=\footnotesize]{$-\frac{\pi}{2}$} circle (1pt) (14,-.3) node[below,font=\footnotesize]{$\frac{3\pi}{2}$} circle (1pt);
    \fill (10,-.3) node[below,black,font=\footnotesize]{$\frac{\pi}{2}$} circle (1pt);
    \fill (8,-.3) node[below,black,font=\footnotesize]{$0$} circle (1pt) (12,-.3) node[below,black]{$\pi$} circle (1pt);

    \fill[purple] (8,2.75) node[above,font=\footnotesize,xshift=0mm]{$T_0$}; circle (1pt);
    \fill[blue] (9,2.75) node[below,font=\footnotesize,xshift=1mm]{$T_0u_\phi$}; circle (1pt);
    \end{scope}


    \draw[dgreen] (2,2.8) node[font=\footnotesize]{$\exp(\mathcal{B}_{\mathfrak{u}_\perp(2)})$};
    \draw[thick,orange] (3,4) ellipse (2cm and .7cm);
    \shadedraw[thick,red] (3,-.3) ellipse (2cm and .7cm);
    \draw[dotted] (3,-.3) ellipse (1cm and .35cm);
    \draw[domain=0:45, smooth, variable=\x]  plot ({3+1*sin(\x)}, {-.3-0.7/2*cos(\x)});
    \draw (3.4,-.8) node{$\phi$};
    \draw (3,-.3) -- (3,-1);
    \draw[thick,->] (3,-.3) node[font=\footnotesize,left]{$J$} -- (3.9,-.6) node[above,xshift=-1.5mm]{$\theta$};
    \draw[dashed] (3,-.3) -- + ({2*sin(45)},{-.7*cos(45)});
    \fill (1,-.3) node[left,xshift=-.5mm]{$=$} circle (1pt) (5,-.3) circle (1pt) (3,-1) circle (1pt) (3,.4) circle (1pt) (3,-.3) circle (1pt);
    \draw (.7,.3) node{$\mathcal{M}$} (.8,-1) node[red]{$\mathcal{B}_{\mathcal{M}}$};
    
    \draw (9,-.7) node{$\phi$};

    \begin{scope}
    \clip (6,1.5) rectangle (16,4);
    \end{scope}


    \draw[domain=-180:180, smooth, variable=\x, black,thick]  plot ({15+2*sin(\x)/4}, {1.5+0.7*cos(\x)/4});

    \draw[domain=270:450, smooth, variable=\x, dgreen,thick,opacity=.5]  plot ({15+2*sin(\x)/4}, {2.7+0.7*cos(\x)/4+.3*cos(\x/2)});
    \draw[domain=-270:270, smooth, variable=\x, dgreen,thick]  plot ({15+2*sin(\x)/4}, {2.7+0.7*cos(\x)/4+.3*cos(\x/2)});
    
    \draw[dgreen,thick] (15,3.5) node[font=\footnotesize]{$u_\phi\in\mathrm{U}(2)$};

    \draw[thick,->] (15,2.4) -- node[right,font=\footnotesize]{$\overline{\det}$} (15,1.8);
    \draw (15,1) node[font=\footnotesize]{$e^{\ii 2\phi}\in\mathrm{U}(1)$};

    \fill[red] (15,-.3) circle (1pt);
    \fill[dgreen] (14,2.75) node[font=\footnotesize]{$\simeq$};
    \fill[orange] (5.4,2.75) node[font=\footnotesize]{$=$};
    \fill[red] (5.4,-.3) node[font=\footnotesize]{$=$};
    \fill[red] (14.3,-.3) node[font=\footnotesize]{$=$};

    \pgfmathsetmacro{\rvec}{.8}
    \pgfmathsetmacro{\thetavec}{30}
    \pgfmathsetmacro{\phivec}{50}
    \tdplotsetmaincoords{100}{-50}
    \begin{scope}[scale=1.5,xshift=-5mm,yshift=-2mm]
        \shadedraw[tdplot_screen_coords,ball color = white] (0,0) circle (\rvec);
        \draw (.2,-.3) node{$\phi$};
        \draw (.6,-.05) node{$\theta$};
        \coordinate (O) at (0,0,0);
        \tdplotsetthetaplanecoords{\phivec}
        
        \draw[dotted,tdplot_main_coords] (\rvec,0,0) arc (0:360:\rvec);
        \draw[tdplot_main_coords] (\rvec,0,0) arc (0:30:\rvec);
        
        \fill[black,tdplot_main_coords] (0,0,\rvec) node[above,font=\footnotesize]{$J$} circle (.75pt);
    \begin{scope}[scale=.8,tdplot_main_coords]
	\coordinate (O) at (0,0,0);
 
	\foreach \angle in {0}
	{
		\tdplotsinandcos{\sintheta}{\costheta}{\angle}%

		\coordinate (P) at (0,0,\sintheta);

		\tdplotdrawarc{(P)}{\costheta}{0}{45}{}{}
		
		\tdplotsetthetaplanecoords{\angle}
		
		\tdplotdrawarc[tdplot_rotated_coords,style={thick,->}]{(O)}{1}{0}{100}{}{}
  \tdplotdrawarc[tdplot_rotated_coords,style={dashed}]{(O)}{1}{0}{180}{}{}
	}
        \foreach \angle in {45}
	{
		\tdplotsinandcos{\sintheta}{\costheta}{\angle}%

		\coordinate (P) at (0,0,\sintheta);

		
		\tdplotsetthetaplanecoords{\angle}
		
		\tdplotdrawarc[tdplot_rotated_coords]{(O)}{1}{0}{180}{}{}
	}
    \end{scope}
    \fill[red,tdplot_main_coords] (0,0,-\rvec) 
    circle (.75pt);
    \end{scope}
    
    \end{tikzpicture}
    \ccaption{Fermionic fiber bundle for $\mathrm{SO}(4,\mathbb{R})$}{We visualize the group $\mathrm{SO}(4,\mathbb{R})$ as principal fiber bundle over $\mathrm{SO}(4,\mathbb{R})/\mathrm{U}(2)$. The base space is a disk parametrized by $(\theta,\phi)$, where we identify quasi-boundary (at $\theta=\pi$) with a single point (turning the disk into a sphere $S^2$ with this point being the south pole). Considering that the quasi-boundary of the base manifold is identified with a single point, the quasi-boundary of the bundle must be identified as a single $\mathrm{U}(2)$ fiber. Choosing the reference point $T_0\in\mathcal{B}_\mathcal{G}$ from~\eqref{eq:Tphi} fixes the isomorphism $\mathcal{B}_\mathcal{G}\simeq \mathrm{U}(2)$, \ie we can identify any other point $M\in\mathcal{B}_\mathcal{G}$ by the appropriate $u\in\mathrm{U}(2)$ with $M=T_0u$. We only illustrate one dimension of this fiber, as the different spirals represent single elements $T_0$ (purple spiral) and $T_0u_\phi$ (blue spiral), so our $2$-dimensional illustration of $\mathcal{B}_{\mathcal{G}}$ as a torus is actually a $1$-dimensional circle, once different spirals are identified as single points. We see that $\exp(\mathcal{B}_{\mathfrak{u}_\perp(2)})=\{T_0u_\phi\,|\,\phi\in[0,2\pi)\,\}$. We see that $\exp(\mathcal{B}_{\mathfrak{u}_\perp(2)})$ represents a circle within the boundary $\mathcal{B}_{\mathcal{G}}\simeq\mathrm{U}(2)$. Recall that we are only interested in the fundamental group of $\mathrm{U}(2)$, so we can use the map $\overline{\det}$ to map this circle onto $\mathrm{U}(1)$, while preserving the winding number. We see that $\overline{\det}(u_\phi)$ wraps around $\mathrm{U}(1)$ twice. This circle is contractible in $\mathcal{G}$, as we can just shrink it to a point, while in the fiber it has winding number $2$. 
    }
    \label{fig:case-study-fermionic}
\end{figure*}

\textbf{Exponential map.} It is well-known that the exponential map from the Lie algebra to the Lie group is surjective, when the Lie group is semi-simple, compact and connected. This is the case for $\mathrm{SO}(4,\mathbb{R})$, so for every $M\in\mathrm{SO}(4,\mathbb{R})$, there exist a $K\in\mathfrak{so}(4,\mathbb{R})$ with $M=e^K$. Moreover, $K$ is not unique, as we can shift its eigenvalues quadruples, which are all imaginary, by a multiple of $\pm 2\pi\ii$ without changing $M$.

As discussed in section~\ref{sec:principal-fiber-bundles}, it suffices to exponentiate elements of $\mathcal{I}_{\mathfrak{u}_\perp(N)}$ to intersect every fiber within $\mathcal{I}_{\mathcal{G}}$. We compute $T=e^{K_+(\theta,\phi)}$ and $\pi(T)=J_T=TJT^{-1}=T^2J$ yielding
\begin{align}
T=\left(\begin{smallmatrix}
 \cos \frac{\theta }{2} & \sin \frac{\theta }{2} \cos \phi  & 0 & \sin \frac{\theta }{2} \sin \phi  \\
- \sin \frac{\theta }{2} \cos \phi  & \cos \frac{\theta }{2} & -\sin \frac{\theta }{2} \sin \phi  & 0 \\
 0 & \sin \frac{\theta }{2} \sin \phi  & \cos \frac{\theta }{2} & -\sin \frac{\theta }{2} \cos \phi  \\
 -\sin \frac{\theta }{2} \sin \phi  & 0 & \sin \frac{\theta }{2} \cos \phi  & \cos \frac{\theta }{2}
\end{smallmatrix}\right),\\
J_T=\left(
\begin{smallmatrix}
 \cos\theta  & \sin\theta  \cos\phi  & 0 & \sin\theta  \sin\phi  \\
 -\sin \theta  \cos \phi  & \cos \theta  & -\sin \theta  \sin \phi  & 0 \\
 0 & \sin \theta  \sin \phi  & \cos \theta  & -\sin \theta  \cos \phi  \\
 -\sin \theta  \sin \phi  & 0 & \sin \theta  \cos \phi  & \cos \theta 
\end{smallmatrix}\right).\label{eq:JT-fermions}
\end{align}
Taking the limit $\theta\to\pi$ yields $J_T=-J$ regardless of $\phi$, which means the quasi-boundary $\mathcal{B}_{\mathcal{M}}$ consists of the single point $-J$. As we will later see, this corresponds to the fermionic Gaussian state $\ket{-J}=\ket{11}$ that is orthogonal to $\ket{J}=\ket{00}$ when written in the Fock basis. We can also evaluate $T_\phi=\lim_{\theta\to\pi}T$ to find
\begin{align}
    T_\phi=\left(\begin{array}{cccc}
 0 & \cos \phi  & 0 & \sin \phi  \\
-  \cos \phi  & 0 & -\sin \phi  & 0 \\
 0 & \sin \phi  & 0 & - \cos \phi  \\
 - \sin \phi  & 0 & \cos \phi  & 0
\end{array}\right)\,.\label{eq:Tphi}
\end{align}

\textbf{Fermionic fiber bundle.} Applying our understanding from section~\ref{sec:principal-fiber-bundles}, we know that $\mathrm{SO}(4,\mathbb{R})$ is a $\mathrm{U}(2)$ principal fiber bundle with base manifold $\mathcal{M}=\mathrm{SO}(4,\mathbb{R})/\mathrm{U}(2)$, which is topologically equal to $S^2$. We have the fundamental groups
\begin{align}
    \pi_1(\mathrm{SO}(4,\mathbb{R}))=\mathbb{Z}_2\quad\text{and}\quad\pi_1(\mathrm{U}(2))=\mathbb{Z}\,,
\end{align}
while $\pi_1(S^2)$ is trivial. This immediately implies that the fermionic fiber bundle is non-trivial, as the $\mathrm{U}(2)$ fibers must be glued together in such a way that a loop with even winding number within a single fiber can be contracted when deforming it continuously across several fibers. In order to see this explicitly, we will use the reduction $\overline{\det}: \mathrm{U}(2)\to\mathrm{U}(1)$. This reduction induces an isomorphism between fundamental groups, \ie a loop in $\mathrm{U}(2)$ with given winding number is mapped to a loop in $\mathrm{U}(1)$ with the same winding number. It relies on the fact that we have topologically $\mathrm{U}(2)=\mathrm{SU}(2)\times\mathrm{U}(1)$, where the fundamental group $\pi_1(\mathrm{SU}(2))$ is trivial. This means the fundamental group of $\mathrm{U}(2)$ is fully due to the $\mathrm{U}(1)$ subgroup.\footnote{We can embed $e^{\ii\theta}\in\mathrm{U}(1)$ back into $\mathrm{U}(2)$ as $\overline{u}=\mathrm{diag}(e^{\ii\theta},1)$.}

Recall from section~\ref{sec:principal-fiber-bundles} that we have $\mathcal{I}_{\mathcal{G}}=\mathcal{I}_{\mathcal{M}}\times\mathrm{U}(N)$ and we already saw that $\mathcal{I}_{\mathcal{M}}$ is given by a disk, which we parametrized by $(\theta,\phi)$. We already saw in~\eqref{eq:JT-fermions} that the apparent perimeter of this disk is mapped to a single point $J_T=TJT^{-1}=-J$. We therefore need to identify the perimeter $\mathcal{B}_{\mathcal{M}}$ of our disk $\mathcal{I}_{\mathcal{M}}$ with a single point, thereby turning our base manifold $\mathcal{M}$ into a $2$-sphere, as sketched in figure~\ref{fig:case-study-fermionic}. Consequently, the quasi-boundary $\mathcal{B}_{\mathrm{SO}(4,\mathbb{R})}$ is given by a single fiber $\mathrm{U}(2)$.

To understand the global topological structure of $\mathrm{SO}(4,\mathbb{R})$, we thus need to analyze how the fibers of $\mathcal{I}_{\mathrm{SO}(4,\mathbb{R})}=\mathcal{I}_{\mathcal{M}}\times\mathrm{U}(2)$ are glued together with the single $\mathrm{U}(2)$ fiber over the boundary point $\mathcal{B}_{\mathcal{M}}$. We saw that if we approach the quasi-boundary with $T=e^{K_+(\theta,\phi)}$ in the limit $\theta\to\pi$ we arrive at the point $T_\phi$ from~\eqref{eq:Tphi}. Considering that for all $\phi$, we have $J_{T_\phi}=T_\phi JT_\phi^{-1}=-J$ implies that they are all in the same $\mathrm{U}(2)$ fiber and thus related via $T_\phi=T_{\phi'}u$ with $u\in\mathrm{U}(2)$. If we choose the reference $T_0$, we find the group element
\begin{align}
    \overline{u}_\phi=\begin{pmatrix}
    0 & \cos\phi+\ii \sin\phi\\
    -\cos\phi-\ii \sin\phi & 0
    \end{pmatrix}
\end{align}
explicitly, such that $T_\phi=T_0 u_\phi$, where we used the notation from~\eqref{eq:ComplexDecomposition}. The circle
\begin{align}
    \exp(\mathcal{B}_{\mathfrak{u}_\perp(2)})=\{T_\phi\,|\,0\leq \phi<2\pi\}
\end{align}
describes a circle that is fully contained in the fiber above $\mathcal{B}_{\mathcal{M}}=\{-J\}$ and the relation $T_\phi=T_0u_\phi$ tells us how the nearby fibers are glued together at this point.

As $\mathrm{U}(2)$ is $4$-dimensional, we cannot draw it explicitly. In order to still visualize this gluing, we can remember that we are mostly interested in the homotopy of our fibers and the whole fiber bundle, as this what determines the fundamental group relevant for constructing the double cover. Therefore, we can collapse $\mathrm{U}(2)$ to $\mathrm{U}(1)$ using the determinant function $\overline{\det}$ introduced above. Note that $\mathrm{U}(2)$ has fundamental group $\mathbb{Z}$, while the higher homotopy groups can be found from the fact that $\mathrm{U}(2)$ contains $\mathrm{SU}(2)$ with the topology of the $3$-sphere, so all higher homotopy groups of $\mathrm{U}(2)$ apart from the fundamental group coincide with the ones of $\mathrm{SU}(2)\simeq S^3$. Collapsing $\mathrm{U}(2)$ to $\mathrm{U}(1)$ via $\overline{\det}$ will not preserve these groups, so we should remember that there is still a $3$-sphere hiding, but for the consideration of the fundamental group it suffices to consider how the respective $\mathrm{U}(1)$ fibers are glued together.

While the Cartan decomposition of group elements in $\mathcal{B}_{\mathcal{G}}$ is not unique, we can just choose $T_0$ as reference and then Cartan decompose all other $M=T_0u$ with $u\in\mathrm{U}(2)$ and evaluate $\overline{\det}(u)$ to find a group element in $\mathrm{U}(1)=\{c\in\mathbb{C}\,|\,|c|=1\}$. Doing this for $M=T_\phi=T_0u_\phi$ yields
\begin{align}
    \overline{\det}(u_\phi)=e^{2\ii \phi}\,,
\end{align}
so varying $\phi$ from $0$ to $2\pi$ will wrap around $\mathrm{U}(1)$ twice. This means that the circle parameterized by $T_\phi$ contained in $\mathrm{U}(2)$ has winding number $2$, but recall that this circle can be continuously contracted to the identity if we move $e^{K_+(\theta,\phi)}$ to $\theta=0$. This suggests that any loop with even winding number within a $\mathrm{U}(2)$ fiber becomes contractible if we allow deformations within the full bundle manifold $\mathrm{SO}(4,\mathbb{R})$. The fundamental group of $\mathrm{SO}(4,\mathbb{R})$ is therefore $\mathrm{Z}_2=\mathbb{Z}/2\mathbb{Z}$, as expected. This topological structure is visualized in figure~\ref{fig:case-study-fermionic}.

\textbf{Double cover.} Moving from the special orthogonal group $\mathrm{SO}(4,\mathbb{R})$ to its double cover $\mathrm{Spin}(4,\mathbb{R})$ amounts to gluing two copies of our reduced picture on top of each other. In this case, any loop becomes contractible, which is exactly what we expect for the spin group that is known to be simply connected. In our reduced picture, we find the structure of a so-called Hopf fibration~\cite{hopf1931abbildungen}, so our reduced manifold is nothing else than a $3$-sphere written as a $\mathrm{U}(1)$ fiber bundle over the $2$-sphere $S^2$. Note, however, that we ignored here the other $3$-sphere contained in $\mathrm{U}(2)$, so that it should not come as a surprise that globally, we have
\begin{align}
    \mathrm{SO}(4,\mathbb{R})&=\frac{\mathrm{SU}(2)\times \mathrm{SU}(2)}{\mathbb{Z}_2}\,,\\
    \mathrm{Spin}(4,\mathbb{R})&=\mathrm{SU}(2)\times \mathrm{SU}(2)\,.
\end{align}
We can construct this relation explicitly in analogy to the well-known spin-$1$ example, where $\mathrm{Spin}(3)=\mathrm{SU}(2)$ is the double cover of $\mathrm{SO}(3,\mathbb{R})$. For this, it is useful to introduce the Pauli matrix $4$-vector
\begin{align}
    \sigma_\mu=(\id,\ii\sigma_x,\ii\sigma_y,\ii\sigma_z)\,,
\end{align}
which allows us to map any real $4$-vector $x^\mu=(x_0,x_1,x_2,x_3)\in\mathbb{R}^4$ to the matrix
\begin{align}\label{eq:special}
   x\cdot \sigma= x^\mu\sigma_\mu=\begin{pmatrix}
    x_0+\ii x_3 & \ii x_1+x_2\\
    \ii x_2-x_2 & x_0-\ii x_3
    \end{pmatrix}\,.
\end{align}
We can compute the norm $\lVert x\rVert^2=\det(x^\mu\sigma_\mu)=\sum_ix_i^2$. Given a pair $(u,v)\in\mathrm{SU}(2)\times \mathrm{SU}(2)$, we can define the action
\begin{align}
    u (x\cdot \sigma)v^\dagger\,,
\end{align}
which can itself be again written as a linear combination $y\cdot \sigma$. This can be seen explicitly by checking that the action of $u$ and $v$ preserves the special structure of~\eqref{eq:special}. 
This implies that for every $(u,v)$, there exists a matrix $M_{u,v}$, such that
\begin{align}
    y\cdot\sigma=u (x\cdot \sigma)v^\dagger
\end{align}
with $y^\mu=(M_{u,v})^\mu{}_\nu x^\nu$. We can calculate
\begin{align}
    \det(y\cdot \sigma)=\det(u (x\cdot \sigma)v^\dagger)=\det(x\cdot \sigma)\,,
\end{align}
which implies that $M$ preserves the norm. One can further show that it also preserves the orientation, so $M\in\mathrm{SO}(4,\mathbb{R})$. Finally, one can show that the map is surjective and each element $M_{u,v}\in\mathrm{SO}(4,\mathbb{R})$ has two pre-images given by $(u,v)$ and $(-u,-v)$. With this construction, we recognize $\mathrm{SU}(2)\times\mathrm{SU}(2)$ as double cover of $\mathrm{SO}(4,\mathbb{R})$.

\textbf{Cocycle function.} Given two orthogonal transformations $M_1$ and $M_2$, such that $\Delta_{M_i^{-1}}=-M^{-1}_i JM_iJ$ parametrized by $(\phi_i,\theta_i)$ describing spheres, we find the $2$-by-$2$ matrix
\begin{align}
    \hspace{-2mm}\overline{\id-Z_{M_1}Z_{M_2^{-1}}}=\id(1-4e^{\ii (\phi_2-\phi_1)}\tanh{\tfrac{\theta_1}{2}}\tanh{\tfrac{\theta_2}{2}})\,.
\end{align}
which yields the cocycle function
\begin{align}
    \eta=2\arg(1-4e^{\ii (\phi_2-\phi_1)}\tanh{\tfrac{\theta_1}{2}}\tanh{\tfrac{\theta_2}{2}})\,.
\end{align}

\textbf{Complex phase calculation of $e^{\widehat{K}}$.} We can use result~\eqref{eq:exp-squared} directly or use the small size of the problem to evaluate $\braket{0|e^{-\ii\hat{H}}|0}$ explicitly, as follows. The most general quadratic Hamiltonian $\hat{H}=\frac{\ii}{2}h_{ab}\hat{\xi}^a\hat{\xi}^b$ with $\widehat{K}=-\ii \hat{H}$ can be characterized by
\begin{align}
    h_{ab}\equiv\begin{pmatrix}
        0 & x_1 & x_2 & x_3\\
        -x_1 & 0 & x_4 & x_5\\
        -x_2 & -x_4 & 0 & x_6\\
        -x_3 & -x_5 & -x_6 & 0
    \end{pmatrix}\,,
\end{align}
where we have six real parameters $x_i$. We can evaluate $\hat{H}=\frac{\ii}{2}h_{ab}\hat{\xi}^a\hat{\xi}^b$ explicitly, which turns out to be a block diagonal matrix over the even sector spanned by $(\ket{00},\ket{11})$ and the odd sector spanned by $(\ket{01},\ket{10})$. In these bases, we find
\begin{align}\label{eq:fermionic-H}
    \hat{H}=(\vec{n}\cdot \vec{\sigma})\oplus(\vec{r}\cdot \vec{\sigma})\,\,\text{with}\,\,\begin{array}{l}
    \vec{n}=(\tfrac{x_4-x_3}{2},\tfrac{x_1-x_6}{6},-\tfrac{x_2+x_5}{2})\\
    \vec{r}\hspace{.45mm}=(\tfrac{x_4+x_3}{2},\tfrac{x_1+x_6}{2},-\tfrac{x_2-x_5}{2})
    \end{array}\,,
\end{align}
where $\vec{\sigma}=(\sigma_x,\sigma_y,\sigma_z)$ refers to the Pauli matrices. This means the most general quadratic Hamiltonian for two fermions corresponds to direct sum of Pauli operators acting on the even and odd sectors separately. Moreover, we have $e^{\widehat{K}}=(e^{-\ii \vec{n}\cdot \vec{\sigma}})\oplus (e^{-\ii \vec{r}\cdot \vec{\sigma}})$, which shows explicitly the relation $\mathrm{Spin}(4,\mathbb{R})=\mathrm{SU}(2)\times\mathrm{SU}(2)$. To evaluate $\braket{J|e^{\widehat{K}}|J}$, we can use the well-known formula
\begin{align}
    e^{-\ii \vec{n}\cdot \vec{\sigma}}=\cos(\tfrac{\theta}{2})\id-\ii\sin(\tfrac{\theta}{2})\frac{\vec{n}\cdot\vec{\sigma}}{\theta}\,,
\end{align}
where $\theta=\sqrt{n_1^2+n_2^2+n_3^2}$. This and $\ket{J}=\ket{00}$ yields
\begin{align}\label{eq:<J|U|J>}
    \braket{J|e^{-\ii \hat{H}}|J}=\cos(\tfrac{\theta}{2})-\ii\sin(\tfrac{\theta}{2})\frac{n_3}{\theta}\,,
\end{align}
from which we can compute its modulus and complex phase (which is defined whenever $\braket{J|e^{-\ii \hat{H}}|J}\neq 0$). We show both, complex phase (color) and modulus (brightness) in figure~\ref{fig:fermionic-<J|U|J>}.

\begin{figure}[t]
    \centering
    \begin{tikzpicture}
    \draw (0,0) node[inner sep=0pt]{\includegraphics[width=\linewidth]{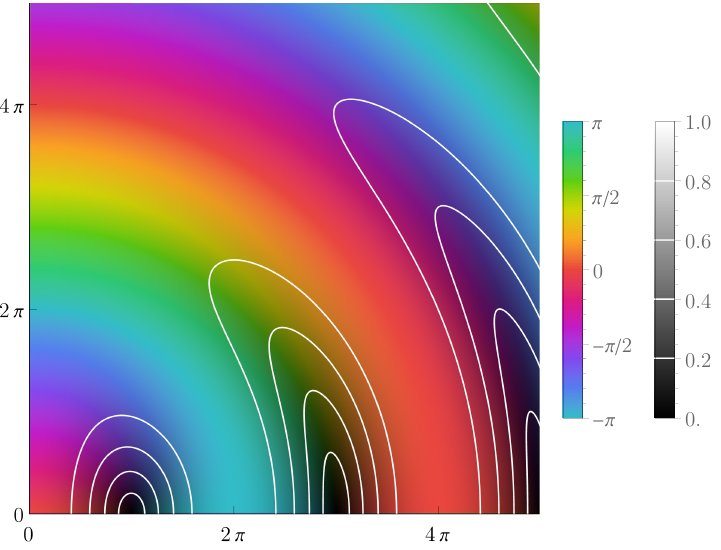}};
    \draw (-4.2,.7) node[font=\footnotesize]{$n_3$};
    \draw (-.3,-3.1) node[font=\footnotesize]{$n_1$};
    \draw (2.55,2.05) -- (3.25,2.6) node[font=\footnotesize,above]{$\arg(\braket{J|e^{\widehat{K}}|J})$};
    \draw (3.65,-1.9) -- (3.25,-2.3) node[font=\footnotesize,below]{$|\braket{J|e^{\widehat{K}}|J}|$};
    \end{tikzpicture}
    \ccaption{Fermionic $\braket{J|e^{\widehat{K}}|J}$ for $\widehat{K}=-\ii \hat{H}$ from~\eqref{eq:fermionic-H} with $\vec{n}=(n_1,0,n_3)$}{We show $\braket{J|e^{\widehat{K}}|J}$ from~\eqref{eq:<J|U|J>}, where we indicate its complex phase by color and its modulus by brightness. White height lines represent the respective modulus as indicated in the brightness legend.}
    \label{fig:fermionic-<J|U|J>}
\end{figure}

\section{\label{sec:Applications}Applications}
In this section we will discuss some common situations in quantum computation and many-body physics where quantities of the form~$\braket{J|e^{\widehat{K}}|J}$ arise naturally. We will show how the methods described in this paper enable techniques for the numerical simulation of these systems.

\subsection{Variational ansätze}
One setting where bosonic and fermionic Gaussian states are frequently applied is as ansatz states for many-body quantum systems. In these systems the dimension of the Hilbert space grows exponentially with the number of elementary degrees of freedom that constitute them, making it impossible to describe a generic quantum state in a numerically exact way even for moderately large systems. Gaussian states however represent a subset of relevant states that can be described efficiently in terms of a small number of parameters. 

Indeed, they can be parametrized by a matrix $J$ whose size scales polynomially with the system size. Furthermore, they are particularly well suited for describing ground states or perturbed states of systems with weak interactions and many frequently used methods rely on this type of Gaussian approximation (\eg the Hartree-Fock method for ground states of fermionic systems, or the Gross-Pitaevskii equations for describing the dynamics of Bose-Einstein condensates).  

Not all systems, however, are well-suited to be described by Gaussian states. Especially in systems where interactions between modes become important, the ground states or evolved states cannot be captured by states with only Gaussian correlations. For these systems it is therefore useful to introduce other ansätze that better approximate the states of interest. Some of these can be seen as extensions of the Gaussian ansatz and their usage relies on the ability to compute quantities such as~$\braket{J|e^{\widehat{K}}|J}$ or~\eqref{eq:generalized-wick}. Here we will briefly present a couple of them.

\subsubsection{Superpositions of Gaussian states}\label{sec:gaussian-superpositions}
The simplest way to go beyond Gaussian states is to consider superpositions of multiple Gaussian states. This class of states has been extensively studied as variational ansatz for ground states~\cite{bravyi_complexity_2017}, for the simulation of gate-based~\cite{Dias:2023arXiv230712912D,Dias:2024arXiv240319059D,reardonsmith2024improved} or analogue evolution and for the tomography of quantum states~\cite{mele2024efficient}.

The basic idea is to consider states of the form 
\begin{equation}
    \ket{\Psi}=\sum_{k=1}^m c_k \,\mathcal{U}(M_k,\psi_k)\ket{J} \,,
    \label{eq:gaussian-superposition}
\end{equation}
where the complex numbers $c_k$ and the double cover elements $(M_k,\psi_k)\in\widetilde{\mathcal{G}}$ represent the variational parameters of the state. If the total number of elements $m$ of the superposition scales polynomially in the system size, this is an efficient parametrisation. It is clear that here it is important to parametrise unambiguously the phase of the unitaries $\mathcal{U}(M_k,\psi_k)$, as these contribute to relative phases in the state which are physically relevant.

This is very closely related to the problem of parametrising Gaussian states in way that also defines their global phase. This has been thoroughly studied in references~\cite{Dias:2023arXiv230712912D,Dias:2024arXiv240319059D} where such a parametrisation is given, together with algorithms to evolve the states under generators of the Gaussian unitary group and to compute overlaps of different states in a phase-sensitive way. The overlap computation in particular is crucial even for simple tasks like evaluating the norm of a state of the form~\eqref{eq:gaussian-superposition}. In terms of the representation theory of Gaussian unitaries that we have introduced, these overlaps can be easily expressed as
\begin{align}
    &\braket{J|\,\mathcal{U}^\dag(M_k,\psi_k)\,\mathcal{U}(M_{k'},\psi_{k'})\,|J}= \nonumber\\
    &\hspace{30mm}=\braket{J|\,\mathcal{U}(M_k^{-1}M_{k'},\psi_{kk'})|J} \,,
\end{align}
where the right hand side can be directly evaluated with the techniques discussed in Sections~\ref{sec:squared-expectation} and~\ref{sec:exact-expectation}. Here $\psi_{kk^\prime}$ is defined as $\psi^*_k\psi_{k^\prime} e^{\ii\eta(M_k^{-1},M_{k^\prime})/2}$ according to the group multiplication rule~\eqref{eq:group-multiplication-double-cover} in the double cover $\widetilde{\mathcal{G}}$.

In conclusion, the results of~\cite{Dias:2023arXiv230712912D,Dias:2024arXiv240319059D} can be recovered as a direct application of the results that we have presented. Let us point out, however, that Gaussian superpositions could be used also in methods that go beyond simply evaluating evolutions and overlaps as in~\cite{Dias:2023arXiv230712912D,Dias:2024arXiv240319059D}. For instance, one may try approximating a system's ground state by optimising the energy function $E=\braket{\psi|\hat{H}|\psi}$ with respect to the variational parameters, for some Hamiltonian $\hat{H}$. In this case we have
\begin{align}
    E&=\sum_{kk^\prime} c_k^* c_{k^\prime} \braket{J|\,\mathcal{U}^\dag(M_k,\psi_k)\hat{H}\,\mathcal{U}(M_{k'},\psi_{k'})|J}\\
    &=\sum_{kk^\prime} c_k^* c_{k^\prime}\braket{J| \,\hat{\widetilde{H}}_k \, \mathcal{U}(M_k^{-1} M_{k^\prime},\psi_{kk^\prime}) |J}\,,
\end{align}
where $\hat{\widetilde{H}}_k\equiv \mathcal{U}^\dag(M_k,\psi_k)\hat{H}\,\mathcal{U}^\dag(M_k,\psi_k)$.
If $\hat{H}$ can be written as a polynomial of linear observables $\hat{\xi}^a$, then also $\hat{\widetilde{H}}_k$ can be written as such a polynomial just by appropriately rotating the coefficients of $\hat{H}$ with the matrix $M_k$.

We thus observe that to evaluate the energy function we need to be able to compute terms of the form
\begin{equation}
    \braket{J| \,\hat{\xi}^{a_1}\cdots\hat{\xi}^{a_d} \, \mathcal{U}(M,\psi) |J} \,, \label{eq:generalized-wick-applications}
\end{equation}
for arbitrary monomials $\hat{\xi}^{a_1}\cdots\hat{\xi}^{a_d}$ of linear operators and double cover elements $(M,\psi)$. We have named these objects \textit{generalized Wick expectation values} and we have described in detail in section~\ref{sec:generalized-wick} how to compute them explicitly.

Finally, let us note that more advanced variational methods~\cite{hackl_geometry_2020} might require to compute, on top of the energy function, also the objects
\begin{equation}
    \braket{V_\mu|\hat{H}|\psi}\,, \quad \braket{V_\mu|V_\nu}\,,
    \label{eq:other-variational-quantities}
\end{equation}
related respectively to the derivative of the energy function and to the geometric structure of the variational manifold. Here $\ket{V_\mu}=\partial_\mu\ket{\Psi}$ represent a tangential vectors of such variational manifold, where by $\partial_\mu$ we indicate the derivative with respect to the $\mu$-th variational parameter. In the case of Gaussian states $\ket{\Psi}$ one can see that the tangent vectors are linear combinations of $\hat{\xi}^{a}\hat{\xi}^{b}\ket{\Psi}$, so it is easy to convince oneself that also the quantities~\eqref{eq:other-variational-quantities} can be expressed in terms of generalized Wick expectations.

\subsubsection{Generalized Gaussian states}
Another ansatz that extends the structure of Gaussian states are generalized Gaussian states. They were introduced in~\cite{shi2018variational} and can be understood in terms of the paradigm of generalized group-theoretic coherent states~\cite{guaita_generalization_2021}.

They can be represented in the form
\begin{equation}
    \ket{\Psi}=\mathcal{U}(M)\mathcal{V}(W)\ket{J}\,,
\end{equation}
where $\ket{J}$ is a Gaussian state, $\mathcal{U}(M)$ is a Gaussian unitary\footnote{Here it is not important to fix the phase of this unitary, as this only contributes an unobservable global phase to the state.} and the special unitary $\mathcal{V}(W)$ is defined as
\begin{equation}
    \mathcal{V}(W)=\exp\left[\ii \sum_{ij} W_{ij} (\hat{a}^\dag_i\hat{a}_i \pm \frac{1}{2})(\hat{a}^\dag_j\hat{a}_j \pm \frac{1}{2}) \right]\,,
\end{equation}
where the $+$ applies to bosons and the $-$ applies to fermions. The group element $M$, the real symmetric matrix $W$ and the choice of the state $\ket{J}$ are all free variational parameters.

The unitary $\mathcal{V}(W)$ endows the state with interesting non-Gaussian correlations. At the same time, its specific form is designed in such a way that relevant variational quantities, such as the Hamiltonian expectation value, can all be computed through closed and numerically efficient algorithms that just exploit properties of Gaussian states and unitaries.  

Indeed, using the methods of reference~\cite{guaita_generalization_2021} one can rewrite the expectation values of any polynomial of linear observables $\hat{\xi}^a$ as a linear combination of generalized Wick expectation values~\eqref{eq:generalized-wick-applications}. Similarly, the quantities~\eqref{eq:other-variational-quantities} can also be written in the same way.
In conclusion, also for applications of Generalized Gaussian States it is important to be able to evaluate expressions of the form~\eqref{eq:generalized-wick-applications}, including information about their phase. This can be achieved using the results we derive in this paper.

\subsection{Numerical evaluation}
Although we have so far dedicated a detailed analysis to the issue of computing expectation values like $\braket{J|\mathcal{U}(M,\psi)|J}$ including the correct sign, let us now point out that in some practical applications it is not always necessary to compute this sign every time explicitly. In other words, it is often sufficient to use the result derived in section~\ref{sec:squared-expectation} and not the more cumbersome calculation described in section~\ref{sec:exact-expectation}. It is also only necessary to keep track of the quantity $M$ that parametrizes the unitary and not of the full double cover element.

Indeed, in some applications, such as most variational methods described above, the key task will be to perform some step-wise evolution of the parametrized state $\ket{\Psi}$. This could be, for instance, the implementation of a gradient descent algorithm that ideally converges to an approximate ground state for the problem. Or it could be an algorithm that approximates within the variational manifold the time evolution dynamics of a state under a given many-body Hamiltonian. 

In these cases, the discretized evolution is an approximation of a theoretically continuous evolution. For optimal results it will be necessary to keep the discretization time step sufficiently small. As we update the unitary $\mathcal{U}(M,\psi)$ from time step $t_i$ to time step $t_{i+1}=t_i+dt$ we expect 
\begin{equation}
    \braket{J|\mathcal{U}(M_{t_i},\psi_{t_i})|J}\longrightarrow \braket{J|\mathcal{U}(M_{t_{i+1}},\psi_{t_{i+1}})|J}
\end{equation}
to also evolve in an approximately continuous manner. We can therefore compute $\braket{J|\mathcal{U}(M_{t_{i+1}},\psi_{t_{i+1}})|J}$ up to a sign, using the results of section~\ref{sec:squared-expectation}, and then pick the sign that makes this result closest to $\braket{J|\mathcal{U}(M_{t_i},\psi_{t_i})|J}$.  A similar observation will apply also to overlaps $\braket{J|\mathcal{U}^\dagger(M_k,\psi_k)\mathcal{U}(M_{k'},\psi_{k'})|J}$ or even generalized Wick expectation values~\eqref{eq:generalized-wick}.

Using this method we are guaranteed to obtain the right result in many applications of practical interest, evaluating at each time step only the simple matrix expression~\eqref{eq:exp-squared}. The exact value of $\braket{J|\mathcal{U}(M,\psi)|J}$ needs to be computed fully only at the initial time step, for which the more involved algorithm described in section~\ref{sec:exact-expectation} may be used. For this it is important to correctly parametrize the initial unitary within the double cover $\widetilde{\mathcal{G}}$, but for subsequent steps it is sufficient to keep track only of the evolution of the matrix $M$.

To give a more concrete example of this, suppose we want to parametrize a state as a superposition of Gaussian states as in~\eqref{eq:gaussian-superposition}. As discussed above, a relevant quantity in this case are the overlaps
\begin{align}
    X_{kk'}=\braket{J|\mathcal{U}^\dagger(M_k,\psi_k)\mathcal{U}(M_{k'},\psi_{k'})|J} \,,
\end{align}
which we can initially evaluate as
\begin{align}
    X_{kk'}=\left(\det\frac{M_{kk'}-JM_{kk'}J}{2}\right)^{\pm1/4} \psi_k^*\psi_{k'}e^{\eta(M_k^ {-1},M_{k'})}\,, \label{eq:overlaps-explicit}
\end{align}
where $M_{kk'}=M_k^{-1}M_{k'}$ and the plus or minus signs in the exponent refer to fermions or bosons respectively.

Assume now that we want to evolve such a state by an infinitesimal step such that each $M_k$ is shifted in the direction $K_k$ in the Lie algebra, that is 
\begin{align}
    M_k\to M_k e^{\epsilon K_k} \,,
\end{align}
for a small $\epsilon$. We want to do this in such a way that the evolution of the overall state is continuous. To do this we can simply update the group elements $M_{kk'}\to e^{\epsilon K_k} M_{kk'}e^{-\epsilon K_{k'}}$ and evaluate the new overlaps $X_{kk'}$ according to~\eqref{eq:exp-squared} as
\begin{align}
    X_{kk'}=\begin{cases}
        \pm\sqrt{\overline{\det}\left(\frac{M_{kk'}-JM_{kk'}J}{2}\right)^{-1}} & \normalfont{\textbf{(bosons)}}\\[5mm]
        \pm\sqrt{\overline{\det}\left(\frac{M_{kk'}-JM_{kk'}J}{2}\right)} & \normalfont{\textbf{(fermions)}}
        \end{cases}
\end{align}
where we now choose the sign $\pm$ based on which value is closer the previous $X_{kk'}$. In particular we no longer need the more elaborate formula~\eqref{eq:overlaps-explicit}.

\section{Discussion}\label{sec:Discussion}
In this paper, we have shown how to explicitly construct the double cover groups of $\mathrm{Sp}(2N,\mathbb{R})$ and $\mathrm{SO}(2N,\mathbb{R})$. This can be used in particular to parametrize the elements of the unitary representation of these double cover groups given by bosonic and fermionic Gaussian unitaries. Having access to this parametrization makes it easy to compute correctly many relevant quantities, such as overlaps of Gaussian states or $\braket{J|e^{\widehat{K}}|J}$, where $\widehat{K}$ is a quadratic Hamiltonian, or even $\braket{J|\, \hat{\xi}^{a_1}\cdots\hat{\xi}^{a_d} \, e^{\widehat{K}} |J}$, where $\hat{\xi}$ are linear phase space operators. This in turn enables many applications, ranging from the development of variational ansätze to the simulation of quantum circuits or many-body dynamics.

\subsection*{Connection to previous results}
Many previously studied problems can be tackled with this approach. The most well-known of these is how to consistently parametrize Gaussian states, such that also their global phase is well-defined. A natural way to do this using our results is to parametrize the states by a Gaussian unitary $\mathcal{U}(M,\psi)$ acting on a reference Gaussian state $\ket{J_0}$, as discussed in section~\ref{sec:gaussian-superpositions}. Indeed, the states $\ket{\Psi(M,\psi)}=\mathcal{U}(M,\psi)\ket{J_0}$ are well-defined, including their global phase and can be used as a building blocks for many applications. Evolving these states by Gaussian evolutions is a simple application of the product rule~\eqref{eq:group-multiplication-double-cover} and computing their relative overlaps can be done using Results 3a and 3b discussed in section~\ref{sec:exact-expectation}.

A more conventional approach to this problem, studied in detail by Dias and König~\cite{Dias:2023arXiv230712912D,Dias:2024arXiv240319059D}, would be to observe that, if one neglects global phases, any Gaussian state $\ket{J}$ can be reached from $\ket{J_0}$ through a Gaussian transformation, identified just by a group element $M\in \mathcal{G}$. To fix the global phase, including the ill-defined phase of the Gaussian transformation, one can use $\ket{J_0}$ again and specify the value of the relative phase $\psi=\braket{J_0|J}/|\!\braket{J_0|J}\!|$. With this in place, the main difficulty is to compute overlaps $\braket{J_1|J_2}$ between two states $\ket{J_1}$ and $\ket{J_2}$ parametrized by $M_1, \psi_1$ and $M_2, \psi_2$. Typically this is solved by observing that the quantity $A(M_1,M_2)=\braket{J_0|J_1}\braket{J_1|J_2}\braket{J_2|J_0}$
cannot depend on the global phases of the three states $\ket{J_0}$, $\ket{J_1}$ and $\ket{J_2}$, as each state vector appears both as a bra and as a ket, and indeed only depends on the matrices $M_1$ and $M_2$. Then the relative overlap can be computed as
\begin{equation}
    \braket{J_1|J_2}=\frac{A(M_1,M_2)}{\braket{J_0|J_1} \braket{J_2|J_0} } = \frac{A(M_1,M_2) \, \psi_1^* \psi_2}{D(M_1) D(M_2)} \,, \label{eq:conventional-2-state-overlap}
\end{equation}
where we observed that $\left|\braket{J_0|J_1}\right|$ and $\left|\braket{J_2|J_0}\right|$ must also be expressible as a function $D$ of just $M_1$ and $M_2$, respectively. In conclusion, everything is derived in terms of the quantities $D(M)$ and $A(M_1,M_2)$ for which closed-form expressions are known. The latter have been derived for bosons in~\cite{PhysRevA.61.022306} and for fermions in~\cite{bravyi_complexity_2017}, in both cases using methods involving the phase space formalism. In our notation, the resulting expressions are for bosons 
\begin{align}
   A(M_1,M_2)&=\left[\det(J_0+M_2J_0M_2^{-1})\right. \nonumber\\
   &\hspace{15mm}\left. \times\det(\tilde{J}_0+M_1J_0M_1^{-1})\right]^{-1/2}\label{eq:traditionalA-boson}
\end{align}
where $\tilde{J}_0=J_0-(J_0-\ii)(J_0+M_2J_0M_2^{-1})^{-1}(J_0+\ii)$, and for fermions
\begin{align}
    A(M_1,M_2)&={\left(\frac{\ii}{4}\right)}^{\!N} \!\mathrm{Pf}\begin{pmatrix}
        \ii J_0 & \id & -\id \\
        -\id & \ii M_1J_0M_1^{-1} & \id \\
        \id &  -\id & \ii M_2J_0M_2^{-1}
    \end{pmatrix} \,. \label{eq:traditionalA-fermions}
\end{align}

Notice that this approach directly relates to our formalism. Once we parametrize a state as $\ket{\Psi(M,\psi)}=\mathcal{U}(M,\psi)\ket{J_0}$, it has a phase overlap with $\ket{J_0}$ given by 
\begin{align}
    \frac{\braket{J_0|\Psi(M,\psi)}}{|\braket{J_0|\Psi(M,\psi)}|}&=\frac{\braket{J_0|\mathcal{U}(M,\psi)|J_0}}{|\braket{J_0|\mathcal{U}(M,\psi)|J_0}|} \\
    &= \begin{cases}
    \psi^* & \normalfont{\textbf{(bosons)}}\\
    \psi & \normalfont{\textbf{(fermions)}}
    \end{cases}
\end{align}
so our parametrization exactly coincides with the one discussed above (up to a different complex conjugate convention for bosons). In fact, we can read off the functions $D(M)$ and $A(M_1,M_2)$ as introduced in~\eqref{eq:def-D} and in~\eqref{eq:A-as-eta}, which are related to $\det(C_M)$ and $e^{\ii \eta(M_1,M_2)/2}$. We conclude that our expression derived in~\eqref{eq:etaDefinitionLog} and in~\eqref{eq:A-as-eta} must coincide with~\eqref{eq:traditionalA-boson} and~\eqref{eq:traditionalA-fermions}, so our formalism naturally contains these previous approaches when our parametrization $\mathcal{U}(M,\psi)$ of the unitary representation of $\widetilde{\mathcal{G}}$ is applied to a reference state $\ket{J}$. Interestingly, however, we observe that our findings have very different analytical forms and have been derived using completely independent methods. In particular, they have been derived by referring exclusively to the group theoretic structure of Gaussian unitaries without any reference to the phase space formalism.

Finally, let us highlight that the main goal of our formalism is to conveniently parametrize the double cover of $\mathcal{G}$, so parametrizing the set of Gaussian state vectors $\ket{J}$ and their relative complex phases is thus a byproduct of our formalism.

\subsection*{Extensions and outlook}
Several possible extensions of our construction can be taken into consideration. The first one comes from observing that, for fermions, we have restricted our attention to the group $\mathrm{SO}(2N,\mathbb{R})$, that is the component of the orthogonal group connected to the identity. It is known, however, that actually the whole orthogonal group $\mathrm{O}(2N,\mathbb{R})$ admits a (projective) representation as unitary operators on the fermionic Hilbert space. To extend our construction to this case it is sufficient to add a representation of a single element $M_d$ belonging to the disconnected component of $\mathrm{O}(2N,\mathbb{R})$. Indeed, any other element of this disconnected component can be parametrized as $M=M_d M_c$, for some $M_c$ that lies in the connected component and can be treated with the already developed methods. Computing overlaps involving such Gaussian unitaries does not introduce further complications, as $\braket{J|\mathcal{U}(M)|J}=0$ whenever $M$ is in the disconnected component. Most importantly, when computing the overlap of two states $\mathcal{U}(M_{i})\ket{J_0}$ with $M_i$ in the disconnected component, we can use group multiplication to reduce this calculation to only involve $M_1^{-1}M_2$ which will be back in the connected component.

Another extension, in this case for bosons, is the inclusion of displacement operators $\mathcal{D}(z)=\exp(\ii z_a \hat{\xi}^a)$. We do not expect this extension to introduce particular issues from the point of view of the double cover and indeed it has already been successfully incorporated into the formalism of Dias \textit{et al.}~\cite{Dias:2024arXiv240319059D} for the purpose of parametrizing Gaussian states. From a group theory perspective, the displacement operators form a representation of the $(2N+1)$-dimensional Heisenberg group. This is a proper representation and does not have any issues related to coverings. The full Gaussian group is then built by extending the metaplectic group by this Heisenberg group through a semidirect product. 

\begin{acknowledgments}
We thank Eugenio Bianchi, Pavlo Bulanchuk, Ignacio Cirac, Dougal Davis, Robert Jonsson, John Rawnsley, Norbert Schuch, Robert Schuhmann, Tao Shi and Jan Philip Solovej for helpful discussions. TG was supported by the German Federal Ministry for Education and Research (BMBF) under the project FermiQP. LH acknowledges support by the Alexander von Humboldt Foundation, by the grant $\#$62312 from the John Templeton Foundation, as part of the \href{https://www.templeton.org/grant/the-quantuminformation-structure-ofspacetime-qiss-second-phase}{‘The Quantum Information Structure of Spacetime’ Project (QISS)}, by grant $\#$63132 from the John Templeton Foundation and an Australian Research Council Australian Discovery Early Career Researcher Award (DECRA) DE230100829 funded by the Australian Government. The opinions expressed in this publication are those of the authors and do not necessarily reflect the views of the respective funding organization.
\end{acknowledgments}


\appendix

\section{Normal form of symplectic Lie generators}\label{app:normal-form}
The derivation of Result 3a (Bosons), where we provide in~\eqref{eq:main-bosons} a formula for $\braket{J|e^{\widehat{K}}|J}$, requires a careful analysis of the symplectic generator $K$, which may have both real and imaginary eigenvalues and potentially even a non-diagonalizable nilpotent part. Normal forms of symplectic generators or, equivalently, of quadratic bosonic Hamiltonians have been studied in the literature~\cite{williamson1936algebraic,laub1974canonical,arnol2013mathematical} before, but we follow the conventions of the modern review~\cite{Kustura:2019PhRvA..99b2130K}, which we briefly summarize in the following.

\begin{proposition}[based on~\cite{Kustura:2019PhRvA..99b2130K}]\label{prop:K-normal-form}
Given a symplectic generator $K\in\mathfrak{sp}(2N,\mathbb{R})$, there exists a basis, such that $\Omega$ takes the standard form~\eqref{eq:Omega-G-standard-form} and $K$ is given by
\begin{align}\label{KN}
K \equiv \begin{pmatrix} O_I & O_R \\ O_L & -O_I^\intercal  \end{pmatrix}\,,
\end{align}
where $O_\chi \in \mathbb{R}^{N\times N} $ ($\chi=I,L,R$), $O_R= O_R^\intercal$, and $O_L= O_L^\intercal$. The matrices $O_\chi$ can each be expressed as a direct sum of blocks over different eigenvalue types
\begin{align}
   O_\chi \equiv  O_\chi^{(\mathfrak{\mathcal{R}})} \oplus O_\chi^{(\mathfrak{\mathcal{C}})} \oplus O_\chi^{(0)} \oplus O_\chi^{(\mathfrak{\mathcal{I}})}\,, 
\end{align}
where
\begin{align}
O_\chi^{(\mathfrak{\mathcal{R}})} &\equiv \bigoplus_{\mathcal{R}_i   }  \spare{\Moplus_{j=1}^{m_i} I_{\chi}^{(1)}(\lambda_i,D_{ij})},\\
O_\chi^{(\mathfrak{\mathcal{C}})} &\equiv \bigoplus_{\mathcal{C}_i }  \spare{\Moplus_{j=1}^{m_i} I_{\chi}^{(2)}(\lambda_i,D_{ij})},\\
O_\chi^{(0)} &\equiv  \bigoplus_{j=1}^{l_0} I_{\chi}^{(3)}(0,D_{0j}) \bigoplus_{j=1}^{n_0} I_{\chi}^{(4)}(0,D_{0j}),\\
O_\chi^{(\mathcal{I})} &\equiv \bigoplus_{\mathcal{I}_i  }  \spare{ \Moplus_{j=1}^{l_i} I_{\chi}^{(5)}(\lambda_i,D_{ij}) \Moplus_{j=1}^{n_i} I_{\chi}^{(6)}(\lambda_i,D_{ij})}
\end{align}
with eigenvalues $\lambda_i$ and associated Jordan block dimensions $D_{ij}$ for the Jordan block $j$. The matrices $I^{(\mathfrak{c})}_\chi(\lambda,D)$ can be read from table~\ref{TAB:NormalFormList} where $\mathfrak{c}=1,\dots,6$ labels six distinct cases.
\end{proposition}
\begin{proof}
A constructive proof is given in~\cite{Kustura:2019PhRvA..99b2130K}.
\end{proof}
\begin{table*}
  \caption{Blocks for the real Jordan normal form of the symplectic generator $K$. We denote $\lambda = \mu + \ii \nu$ (with  $\mu,\nu \in \mathbb{R}$) and $\sigma=\alpha_{\lambda}(e, \tilde{e})$, where $e$ is the generating generalized eigenvector (gGEV) associated to the block. Note that in $\mathfrak{c}=6$, $\ii\sigma \in \mathbb{R}$.}\label{TAB:NormalFormList}
	\begin{tabular}{|C{0.03\textwidth} | C{0.27\textwidth} | C{0.27\textwidth} |C{0.27\textwidth} | C{0.09\textwidth} |}
		\hline
		$\mathfrak{c}$ & $I_I^{(\mathfrak{c})} (\lambda, D)$ & $I_R^{(\mathfrak{c})}(\lambda,D)$ & $I_L^{(\mathfrak{c})}(\lambda,D)$ & dimension \\
		\hline
		1 
		&
            $\begin{pmatrix}
		\mu &  &  &   &  \\
		1 & \mu &  &  &  \\ 
		& \ddots & \ddots &  &  \\ 
		&  & & 1 & \mu \\
		\end{pmatrix}$
		&
		$\mathbb{0}$
		&
		$\mathbb{0}$
		&
		$D$
		\\
		\hline
		2
		&
		$\begin{pmatrix}
		\mu & \nu &  &  &  &  &  &  \\
		-\nu & \mu &  &  &  &  &  &  \\
		1 & 0 & \mu & \nu &  &  &  &  \\
		0 & 1 & -\nu & \mu &  &  &  &  \\
		&  &  & \ddots & \ddots &   &  &  \\
		&  &  &  & 1 & 0 & \mu & \nu \\
		&  &  &  & 0 & 1 & -\nu & \mu \\
		\end{pmatrix}$
		&
		$\mathbb{0}$
		&
		$\mathbb{0}$
		&
		$2D$
		\\
		\hline
		3
		&
		$\sigma
		\begin{pmatrix}
		0 &  &  &   \\ 
		1 & 0 &  &  \\ 
		& \ddots & \ddots & \\ 
		&   & 1 & 0\\
		\end{pmatrix} $
		&
		$\mathbb{0}$
		&
		$\sigma
		\begin{pmatrix}
		0 &  &  &   \\ 
		& 0 &  &  \\ 
		&  & \ddots & \\ 
		&   &  & (-1)^{D/2}\\
		\end{pmatrix} $
		&
		$D/2$ (integer)
		\\
		\hline
		4
		&
		$\begin{pmatrix}
		0 &  &  &   \\ 
		1 & 0 &  &  \\ 
		& \ddots & \ddots & \\ 
		&   & 1 & 0\\
		\end{pmatrix} $
		&
		$\mathbb{0}$
		&
		$\mathbb{0}$
		&
		$D$ (odd)
		\\
		\hline
		5
		&
		$\mathbb{0}$
		&
		$\sigma
		\begin{pmatrix}
		&  & &  &  & \nu\\
		&  &  & & \nu & 1\\
		&  &  & & -1 & \\
		&  &\iddots  & \iddots &  & \\
		& \nu & -1 &  & & \\
		\nu & 1 &  & &  & \\
		\end{pmatrix}$
		& 
		$\sigma
		\begin{pmatrix}
		&  & &  &  -1 & -\nu\\
		&  &  & 1& -\nu & \\
		&  &\iddots  & \iddots &  & \\
		& 1 &  & & & \\
		-1 & -\nu &  &  & & \\
		-\nu & &  & &  & \\
		\end{pmatrix}$
		&
		$D$ (even)
		\\
		\hline
		6
		&
		$\begin{pmatrix}
		0 &  &  &   \\ 
		1 & 0 &  &  \\
		& \ddots & \ddots & \\ 
		&   & 1 & 0\\
		\end{pmatrix} $
		&
		$ \ii \sigma
		\begin{pmatrix}
		&  &  &  & \nu\\
		&  &  & -\nu & \\
		&  & \iddots & & \\
		& -\nu &  &  & \\
		\nu &  &  &  & \\
		\end{pmatrix} $
		&
		$\ii \sigma
		\begin{pmatrix}
		&  &  &  & -\nu\\
		&  &  & \nu & \\
		&  & \iddots &  & \\
		& \nu &  &  & \\
		-\nu &  &  &  & \\
		\end{pmatrix}$
		&
		$D$ (odd)
		\\
		\hline
	\end{tabular}
\end{table*}

The following lemma will be needed for the subsequent proposition and establishes that we can always perform a change of basis to rescale a certain type of block.
\begin{lemma}\label{lem:rescaling}
Given a symplectic generator $K$ in its normal form of proposition~\ref{prop:K-normal-form}, we can always perform a further symplectic basis transformation, \ie thereby retaining the standard form~\eqref{eq:Omega-G-standard-form} of $\Omega$, such that blocks of type $\mathfrak{c}=3$ are rescaled by $\epsilon>0$.
\end{lemma}
\begin{proof}
We consider an individual block $K_i$ of type of $\mathfrak{c=3}$ in table~\ref{TAB:NormalFormList} of size $D$-by-$D$. Applying the symplectic transformation
\begin{align}
    X=\mathrm{diag}(\epsilon^{\frac{D-1}{2}},\dots,\epsilon^{\frac{1}{2}},\epsilon^{-\frac{D-1}{2}},\dots,\epsilon^{-\frac{1}{2}})
\end{align}
will preserve the standard form~\eqref{eq:Omega-G-standard-form} of $\Omega_i$ associated to this block, but rescales $K_i$ according to
\begin{align}
    K\to X^{-1}K_iX=\epsilon K_i\,,
\end{align}
which means that we can bring the entries of this block arbitrarily close to zero. We then choose a complex structure $J$, which takes the standard form~\eqref{eq:J-standard-form} in this new rescaled basis. This implies that for an $\epsilon$ chosen sufficiently small, all eigenvalues of $C_{e^{tK_i}}=\frac{1}{2}(e^{tK_i}-Je^{tK_i}J)$ will have a positive real part for all $t\in[0,1]$. This is due to continuity, as in the limit $\epsilon\to 0$, we have $e^{tK_i}\to\id$ and $C_{e^{tK_i}}\to\id$, where all eigenvalues approach $1$.
\end{proof}

According to the Jordan–Chevalley decomposition we can decompose any Lie algebra generator $K\in\mathfrak{sp}(2N,\mathbb{R})$ uniquely into three parts
\begin{align}
    K=K_I+K_R+K_N\,,
\end{align}
where all three parts commute with each other and $K_I$ has purely imaginary eigenvalues, $K_R$ has purely real eigenvalues and $K_N$ is non-diagonalizable and nilpotent. If we have brought $K$ into the normal form according to proposition~\ref{prop:K-normal-form}, $K_I$ corresponds to the matrix parts proportional to $\nu$ (different imaginary eigenvalues), $K_R$ those proportional to $\mu$ (different real eigenvalues) and $K_N$ the remaining pieces.

\begin{proposition}
\label{prop:eigenvalues-C}
Given a symplectic generator $K=K_I+K'\in\mathfrak{sp}(2N,\mathbb{R})$, such that $K_I$ is diagonalizable with purely imaginary eigenvalues in the Jordan–Chevalley decomposition, we consider a basis, such that $K$ takes the block structure from proposition~\ref{prop:K-normal-form}, where we use the rescaling from lemma~\ref{lem:rescaling} for blocks of types $\mathfrak{c}=3$. For the reference complex structure $J$, which takes the standard form~\eqref{eq:J-standard-form} in this basis, the function
\begin{align}
    \arg\overline{\det}\sqrt{C_{e^{tK'}}}:=\tfrac{1}{2}\sum_i\arg'(\lambda_i)\label{eq:cont-def}
\end{align}
is continuous in $t\in[0,1]$, where $\arg'(z)=0$ for $z\in(-\infty,0]$ and $(\lambda_1,\dots,\lambda_{2N})$ are the eigenvalues of $C_{e^{t K'}}$.
\end{proposition}
\begin{proof}
We will distinguish several cases. It turns out that the cases $\mathfrak{c}=1,2,4,6$ can be proven using the same simple argument, so we will do that first.
\begin{description}
\item[$\mathfrak{c}=1,2,4,6$] A close inspection of table~\ref{TAB:NormalFormList} shows that the respective matrices $K'=\left(\begin{smallmatrix}A&0\\0&-A^\intercal\end{smallmatrix}\right)$ are all block-diagonal. Note that for $\mathfrak{c}=6$, $K_I=K-K'$ contains off-diagonal blocks, but not $K'$. The block structure implies that the matrix exponential $e^{tK}\equiv\left(\begin{smallmatrix}e^{tA}&0\\0&e^{-tA^\intercal}\end{smallmatrix}\right)$ and thus $C_{e^{tK'}}\equiv\frac{1}{2}\left(\begin{smallmatrix}e^{tA}+e^{-tA^\intercal}&0\\0&e^{tA}+e^{-tA^\intercal}\end{smallmatrix}\right)$ are block-diagonal as well. We thus have $\overline{C}_{e^{tK'}}\equiv\frac{1}{2}(e^{tA}+e^{-tA^\intercal})$. This is a real $N$-by-$N$ matrix, whose determinant $\overline{\det}(C_{e^{K'}})=\det\overline{C}_{e^{K}}$ must be also real. As the circle function never vanishes for bosons and it is equal to $1$ for $t=0$, continuity implies that $\overline{\det}(C_{e^{tK'}})>0$ for all $A$.
\item[$\mathfrak{c}=3$] For this case, we can assume that we already rescaled the respective blocks according to lemma~\ref{lem:rescaling} to ensure that all eigenvalues of $C_{e^{tK'}}$ have positive real part. Of course, it would even suffice if all eigenvalues had some fixed minimal distance from the negative real axis, but we chose positive real part as a simple condition. This implies that we can take the square root of all eigenvalues of $\overline{C}_{e^{tK'}}$ without any discontinuities, as all eigenvalues will be far away from the branch cut along the negative real axis. Consequently, the $\overline{\det}\sqrt{C}_{e^{tK'}}$ as defined in~\eqref{eq:cont-def} will be continuous.
\item[$\mathfrak{c}=5$] Investigating the structure of the matrix $\overline{C}_{e^{tK'}}=\frac{1}{2}(e^{tK'}-Je^{tK'}J)$ for this case yields
\begin{align}
    \overline{C}_{e^{tK'}}\equiv\left(\begin{smallmatrix}
    c_{11} & 0 & c_{12} & \dots & 0 & c_{1n} & 0\\
    0 & c_{nn} & 0 & \dots & c_{2n} & 0 & c_{1n}\\
    c_{12} & 0 & c_{22} &\dots & 0 & c_{2n} & 0\\
    \vdots & \vdots & \vdots & \ddots & \vdots & \vdots & \vdots\\[2mm]
    0 & c_{2n} & 0 & \dots & c_{22} & 0 & c_{12}\\
    c_{1n} & 0 & c_{2n} & \dots & 0 & c_{nn} & 0\\
    0 & c_{1n} & 0 &\dots & c_{12}  & 0 & c_{11}\\
    \end{smallmatrix}\right)\,,
\end{align}
    where we recognize that $\overline{C}_{e^{tK'}}$ itself consists of two times the same complex symmetric $n$-by-$n$ matrix $c=(c_{ij})$ with $n=D/2$ (recall $D$ even for $\mathfrak{c}=5$), interlaced and mirrored with respect to the anti-diagonal. This implies $\overline{\det}\,{C_{e^{tK'}}}=\det^2(c)$ its square root is thus continuous, as each eigenvalue appears with even multiplicity, so there if at all there is an even number of complex eigenvalue passing over the branch cut along the negative real axis. By using $\arg'$ rather than $\arg$ (with $\arg(-1)=\pi$), it is ensured that if such an eigenvalue pair approaches the negative real axis, we do not get a discontinuity.
\end{description}
Together, this analysis of the individual cases implies that $\overline{\det}\sqrt{C_{e^{tK'}}}$ is continuous function for $t\in[0,1]$.
\end{proof}

Let us emphasize that this analysis is really only required if the respective symplectic generator $K$ contains non-trivial blocks of type $\mathfrak{c}=3$ or $\mathfrak{c}=5$, as otherwise there will be no contribution due to $\arg\overline{\det}\sqrt{C_{e^{tK'}}}=0$.
\vfill

\bibliography{references.bib}

\end{document}